\documentclass[11pt]{article}

\usepackage{amssymb,amsmath,amsthm}
\usepackage{bbm}
\usepackage[noend]{algpseudocode}
\usepackage{graphicx}
\usepackage{verbatim}%for comment blocks
\usepackage{url,xspace,hyperref}
\usepackage{cite}
\usepackage{caption}
\usepackage{subcaption}
\usepackage{multirow}
\usepackage{multicol}
\usepackage{latexsym}
\usepackage{amsmath,amssymb,enumerate}
\usepackage{xspace}
\usepackage{algorithm,algorithmicx}
\usepackage{float}
\usepackage{xcolor}
\usepackage{mathrsfs}
\usepackage{cleveref}

\usepackage{fullpage}
\usepackage{geometry}
\usepackage{setspace}
\newtheorem{assumption}{Assumption}[section]
\newtheorem{theorem}{Theorem}[section]
\newtheorem{definition}{Definition}[section]

\newtheorem{lemma}[theorem]{Lemma}

\newtheorem{corollary}[theorem]{Corollary}

\newcounter{note}[section]

\def\sse{\subseteq}

\newcommand{\pr}{\mathbf{P}} 
\newcommand{\E}{\mathbb{E}} 
\newcommand{\ZZ}{\mathbb{Z}}

\newcommand{\R}{\mathbb{R}}
\newcommand{\A}{\mathcal{A}}
\newcommand{\T}{\mathcal{T}}
\newcommand{\calE}{\mathcal{E}}
\newcommand{\D}{\mathcal{D}}

\newcommand{\OPT}{\mathtt{OPT}}
\newcommand{\NA}{\mathtt{NA}}

\newcommand{\AD}{\mathtt{AD}}
\newcommand{\ODT}{\mathtt{ODT}}
\newcommand{\Stem}{\mathtt{Stem}}
\newcommand{\sym}{\star}
\def\nil{\perp}

\newcommand{\ignore}[1]{}
\newcommand{\rv}[1]{\mathcal{{#1}}}
\newcommand{\rl}[2]{#1(#2)}

\newcommand{\asc}{\ensuremath{\mathsf{AdSubCov}}\xspace}
\newcommand{\ssp}{\ensuremath{\mathsf{SSPC}}\xspace}
\newcommand{\WISERu}{\mathtt{WISER}-\mathtt{U}}
\newcommand{\WISERr}{\mathtt{WISER}-\mathtt{R}}
\newcommand{\SYNu}{\mathtt{SYN}-\mathtt{U}}
\newcommand{\SYNr}{\mathtt{SYN}-\mathtt{R}}

\newcommand{\Qh}{\widehat{Q}}
\newcommand{\fh}{\widehat{f}}

\newif\ifICMLVersion

\title{The Power of Adaptivity for Stochastic Submodular Cover}
\author{Rohan Ghuge\thanks{Department of Industrial and Operations Engineering, University of Michigan, Ann Arbor, USA. Research supported in part by NSF grants CMMI-1940766 and CCF-2006778.} \and Anupam Gupta\thanks{Department of Computer Science, Carnegie Mellon University, Pittsburgh, USA. Supported in part by NSF awards CCF-1907820, CCF-1955785, and CCF-2006953.} \and Viswanath Nagarajan$^*$}
\begin{document}
\maketitle

\begin{abstract}
In the stochastic submodular cover problem, the goal is to select a
subset of stochastic items of minimum expected cost to cover a submodular function. Solutions in this setting correspond to sequential decision processes that select items one by one ``adaptively'' (depending on prior observations). While such adaptive solutions achieve the best objective, the inherently sequential nature makes them undesirable in many applications. We ask: \emph{how well can solutions with only a few adaptive rounds approximate fully-adaptive solutions?}   
We give nearly tight answers for both independent and correlated settings, proving smooth tradeoffs between the number of adaptive rounds and the solution quality, relative to fully adaptive solutions. Experiments on synthetic and real datasets  
show qualitative improvements in the solutions as we allow more rounds
of adaptivity; in practice, solutions with a few rounds of adaptivity are nearly as good as fully adaptive solutions.
\end{abstract}

\section{Introduction}\label{sec:intro}

Submodularity is a fundamental notion that arises in applications such
as image segmentation, data summarization~\cite{SimonSS07, LB11,
  SiposSS+12}, hypothesis identification~\cite{BarinovaLK12,
  ChenSM+14}, information gathering~\cite{RadanovicSK+18}, and social
networks~\cite{KempeKT15} . The {\em submodular cover} optimization
problem requires us to pick a minimum-cost subset $S$ of items to
cover a monotone submodular function $f$. Submodular cover has been
extensively studied in machine learning, computer science and
operations research~\cite{W82,GK11,MirzasoleimanKB15,BateniEM18}: here
are two examples from %viral marketing 
sensor deployment and medical diagnosis.

In the sensor deployment setting, we consider the problem of deploying a collection of sensors to monitor some phenomenon \cite{KrauseG07, MiniUS14, SunLL19}, for example: we may wish to monitor air quality or traffic situations. The area each sensor can cover depends on its sensing range. The goal of the problem is to deploy as few sensors as possible to cover a desired region entirely. The area covered as a function of the sensors deployed is a submodular function, and we can cast the sensor deployment problem as a special case of submodular cover. 

In the medical diagnosis example, we have $s$ possible conditions
(hypotheses) the patient may suffer from along with the priors on
their occurrence, and our goal is to perform tests to identify the
correct condition as quickly as possible~\cite{GG74,KPB99,D01,CLS14}. This can be cast as
submodular cover by viewing each test as eliminating all inconsistent
hypotheses. Hence we want a coverage value of $s-1$: once $s-1$
inconsistent hypotheses are eliminated, the remaining one must be
correct. 

Observe that both these applications involve uncertain data: 
the precise area covered by a sensor  is not known before the sensor is actually setup 
and the precise outcome (positive/negative) of a test is not known
until the action has been taken. This uncertainty can be modeled using
{\em stochastic submodular optimization}, where the items are
stochastic.   
  As a simple example of the stochastic nature, each item may
be active or inactive (with known probabilities), and only active
items contribute to the submodular function.

A solution for stochastic submodular cover is now a \emph{sequential
  decision process}. At each step, an item is {\em probed} and its
realization (e.g., active or inactive) is observed. The process is
typically \emph{adaptive}, where all the information from previously
probed items is used to identify the next item to probe. This process
continues until the submodular function is covered, and the goal is to
minimize the expected cost of probed items. Such adaptive solutions are inherently fully sequential, which is undesirable if probing an item
is time-consuming. E.g., in sensor deployment, probing a sensor corresponds to physically deploying a sensor and observing whether it functions as expected, which may take hours or days.
Or, probing/performing a test in medical diagnosis
may involve a long wait for test results. Therefore, we prefer
solutions with only few {\em rounds} of adaptivity.

Motivated by this, we ask: \emph{how well do solutions with only a few
  adaptive rounds approximate fully-adaptive solutions for the
  stochastic submodular cover problem?} We consider both cases where
realizations of different items are independent, and where they are
allowed to be correlated.  For both these situations we give nearly
tight answers, with smooth tradeoffs between the number $r$ of
adaptive rounds and the solution quality (relative to fully adaptive
solutions).

The main contribution of our work is an $r$-round adaptive solution for stochastic submodular cover in the ``set-based'' model for adaptive rounds. In this model, a fixed subset of items is probed in parallel every round (and their total cost is incurred). The decisions in the current round can depend on the realizations seen in all previous rounds. However, as noted in \cite{AAK19},  if we require  function $f$ to be covered with probability one then the $r$-round-adaptivity gap turns out to be very large. 
\ifICMLVersion
\else
 (See \S\ref{app:set-based} for an example.) 
 \fi  
Therefore, we focus on  set-based solutions that are only required to cover the function with high probability.

In designing algorithms, it turns out to be more convenient to work with the ``permutation'' model for adaptive rounds, where the function is covered with probability one. This model was also used in prior literature \cite{GV06, AAK19}. Here,  every round of an $r$-round-adaptive solution  specifies an ordering of all remaining items and probes them in this order until some stopping rule. See Definition~\ref{def:r-round} for a formal definition.  Moreover, our $r$-round adaptive algorithm in the permutation model can be transformed into an  $r$-round adaptive algorithm in the set-based model. We obtain algorithms in the set-based model that:
\begin{itemize}
    \item for any $\eta \in (0,1)$, finds an $r$-round adaptive solution that has expected cost at most $\frac{r\alpha}{\eta }\cdot \OPT$ and  covers the function with  probability at least $1-\eta$.
    \item finds an  $O(r)$-round adaptive solution that has expected cost at most $O({\alpha})\cdot \OPT$ and covers the function with probability at least $1-e^{-\Omega(r)}$.
\end{itemize}
Here $\OPT$ is the cost of an optimal fully-adaptive solution and $\alpha$ is the approximation ratio of our algorithm in the permutation model.

The first algorithm above is  for the case where $r$, the number of rounds of adaptivity, is small (say, a constant). In this, we keep the number of rounds the same, but we lose a factor $r$ in the expected cost. 
The second algorithm is for the case that $r$ is large, e.g., more than a constant. Here, the number of set-based rounds increases by a factor $2$, but  we only lose a constant factor in expected cost. We formalize and prove these results in \S\ref{app:set-based}. For the rest of the paper, an $r$-round adaptive algorithm refers to an an $r$-round adaptive algorithm in the permutation model (unless specified otherwise).

\ifICMLVersion
\subsection{Results}
\else
\subsection{Results and Techniques}\label{sec:results}
\fi
Consider a monotone submodular function $f: 2^U \to \ZZ_{\geq 0}$ with $Q := f(U)$. There are $m$ items, where each item $i$ is a random
variable $\rv{X}_i$ having cost $c_i$ and corresponding to a random
element of $U$. (Our results extend to the more general setting where each item realizes to a subset of $U$.) 
  The goal is to select a set of items such that the
union $S$ of their corresponding realizations satisfies $f(S) = Q$, and the
expected cost is minimized. Our first result is when the items have
\emph{independent distributions}.

\begin{theorem}[Independent Items]
  \label{thm:main-ssc}
  For any integer $r \geq 1$, there is an $r$-round adaptive algorithm
  for the stochastic submodular cover problem with cost
  $O(Q^{1/r} \cdot \log Q)$ times the cost of an optimal adaptive
  algorithm.
\end{theorem}

This improves over an $O(r\, Q^{1/r} \,\log Q \,\log (mc_{max}))$
bound from \cite{AAK19} by eliminating the dependence on the number of
items $m$ and the item costs (which could be much larger than
$Q$). Moreover, our result nearly matches the lower bound of
$\Omega(\frac{1}{r^3} Q^{1/r})$ from \cite{AAK19}. Setting $r=\log Q$
shows that $O(\log Q)$ adaptive rounds give an $O(\log
Q)$-approximation. By transforming this algorithm into a set-based
solution using \Cref{thm:set-large-r}, we get:
\begin{corollary}[Independent  Items: Set-Based Model]
  \label{cor:indep}
There is an $O(\log Q)$ round  algorithm
  for  stochastic submodular cover in the set-based model with cost
  $O(\log Q)$ times the optimal (fully) adaptive cost.
\end{corollary}
This approximation ratio of $O(\log Q)$ is the best possible (unless P=NP) even with an arbitrary number ($r=m$) of
adaptive rounds. 
\ifICMLVersion
Previously, one could only obtain an $O(\log^2 Q
\log(mc_{max}))$-approximation in a logarithmic number of rounds \cite{AAK19}. In the special case of  \emph{unit costs}, \cite{EsfandiariKM19} gave an
$O(\log (mQ))$-approximation for covering ``adaptive submodular'' functions using $O(\log m \, \log(Qm))$ set-based rounds. When costs are arbitrary, their result implies  an $O(\log (mQc_{\max}))$-approximation in $O(\log m \, \log(Qmc_{\max}))$ set-based rounds.
\else
Previously, one could only obtain an $O(\log^2 Q
\log(mc_{max}))$-approximation in a logarithmic number of rounds \cite{AAK19}, or  an $O(\log (mQc_{\max}))$-approximation in $O(\log m \, \log(Qmc_{\max}))$ set-based rounds \cite{EsfandiariKM19}.
\fi

Moreover, \Cref{thm:main-ssc} (with $r=1$) implies an $O(Q\log Q)$ adaptivity
gap for stochastic set cover (a special case of submodular cover),
resolving an open question from~\cite{GV06} up to an $O(\log Q)$
factor, where $Q$ is the number of objects in set cover.

\medskip
Our second set of results is for the case where items have correlated
distributions. Let $s$ denote the support size of the joint
distribution $\D$, i.e., the number of \emph{scenarios}.
% In the case of scenario submodular cover:
\begin{theorem}[Correlated Items]
  \label{thm:scn-main}
For any integer $r \geq 1$, there is an $r$-round adaptive algorithm for scenario submodular cover  with cost $O\left(s^{1/r} (\log s + r\log Q)\right)$ times the cost of an optimal adaptive algorithm.
\end{theorem}
We also obtain a $2r$-round adaptive algorithm with an better cost
guarantee of $O\left(s^{1/r} \log (s Q)\right)$ times the cost of an
optimal adaptive algorithm (see Corollary~\ref{cor:ssp}). Combined with the conversion to a set-based solution (\Cref{thm:set-large-r}) and
setting $r=\log s$, we can then infer:
\begin{corollary}[Correlated Items: Set-Based Model]
  \label{cor:scn} There is an $O(\log s)$ round  algorithm
  for  scenario submodular cover in the set-based model with cost
  $O(\log (s Q))$ times the optimal (fully) adaptive cost.
\end{corollary}
The above approximation guarantee is nearly the best possible, even with an arbitrary number of adaptive rounds: there is an $\Omega(\log s)$-factor hardness of approximation~\cite{CPRAM11}.  
Scenario submodular cover generalizes the classic optimal decision tree problem
\cite{GG74,HR76,L85,KPB99,D01,AH12,GNR17}. A  fully-adaptive $O(\log (sQ))$-approximation for scenario submodular cover  was obtained in \cite{GHKL16}; see also \cite{NavidiKN20} for a more general result.
In terms of rounds-of-adaptivity, an $O(\log (mQ\frac{c_{\max}}{p_{\min}}))$-approximation in $O(\log m \, \log(Qm \frac{c_{\max}}{p_{\min}}))$ set-based rounds follows from \cite{GHKL16,EsfandiariKM19}. Here $p_{\min}\le \frac 1s$ is the minimum probability of any
scenario.  
\ifICMLVersion
Scenario submodular cover is  not ``adaptive submodular'', and so results from \cite{EsfandiariKM19} cannot be used directly. Still, \cite{GHKL16} showed an equivalent goal function for scenario submodular cover that is indeed 
adaptive-submodular, but with a larger ``$Q$ value'' of $\frac{Q}{p_{min}}$. Then,  the algorithm from \cite{EsfandiariKM19} can be applied to this new goal function. \fi
We note that when the number of rounds is less than logarithmic, our result provides the first  approximation guarantee  even in the well-studied special case of optimal decision tree. 

The results in \Cref{thm:main-ssc}
and \Cref{thm:scn-main} are incomparable: while the independent case
has more structure in the distribution $\D$, its support size is
exponential.  Finally,
the dependence on the
support size $s$ is necessary in the correlated setting, as our next
result shows.
\begin{theorem}[Lower Bound] 
  \label{thm:scn-lb}
For any integer $r \geq 1$, there is an instance of scenario submodular cover with $Q=1$ where the cost of any $r$-round adaptive solution is $\Omega(\frac{1}{r^2\log s} \cdot s^{1/r})$ times the optimal adaptive cost.
\end{theorem}
This lower
bound is information-theoretic and does not depend on computational
assumptions, whereas the upper bound of \Cref{thm:scn-main} is given by a polynomial algorithm. 

Finally, we note that our algorithms are also easy to implement. We tested both algorithms on synthetic and real datasets that validate the practical performance of our algorithms. Specifically, we test our algorithm for the independent case (Theorem~\ref{thm:main-ssc})  on instances of stochastic set cover, and our algorithm for the correlated case (Theorem~\ref{thm:scn-main}) on instances of  optimal decision tree. For stochastic set cover, we use real-world datasets to generate instances with $\approx 1200$ items. We observe a sharp improvement in performance within a few rounds of adaptivity, and that $6$-$7$ rounds of adaptivity are nearly as good as fully adaptive solutions. For  optimal decision tree, we use both real-world data and synthetic data. The real-world data has $\approx 400$ scenarios and the synthetic data has $10,000$ scenarios. Again, we find that about 6 rounds of adaptivity suffice to obtain solutions as good as fully adaptive ones. We also compared our algorithms' cost to  information-theoretic lower bounds for both applications: our costs are typically within 50\% of these lower bounds. 

\ifICMLVersion
\subsection{Techniques} 
\else
\subsubsection{Techniques}
\fi
The algorithms for the independent and correlated cases are similar,
with some crucial differences. In each round of both algorithms, we
iteratively compute a ``score'' for each item and greedily select the
item of maximum score. This results in a non-adaptive list of all
remaining items, and the items are {\em probed} in this order until a
stopping rule. The \textsc{ParCA} stopping rule in the independent case corresponds
to reducing the remaining target (on the function value) by a factor
of $Q^{1/r}$, whereas the \textsc{SParCA} rule involves reducing the number of ``compatible
scenarios'' by an $s^{1/r}$ factor in the correlated case.

The analysis for both Theorems~\ref{thm:main-ssc}
and \ref{thm:scn-main} follows parallel lines at the beginning.
For each $i\ge 0$, we relate the ``non-completion'' probabilities of
the algorithm after cost $\alpha\cdot 2^i$ to the optimal adaptive
solution after cost $2^i$. The ``stretch'' factor $\alpha$ corresponds
to the approximation ratio, which is different for the independent and
correlated cases. In order to relate these non-completion
probabilities, we consider the total score $G$ of items selected by
the algorithm between cost $\alpha 2^{i-1}$ and  $\alpha 2^i$. The
crux of the analysis lies in giving lower and upper bounds on the
total score $G$; the arguments here are different for the independent
and correlated settings.

In the independent case, the score of any item $\rv{X}_e$ is an estimate of its relative marginal gain, where we take an expectation over all previous items as well as $\rv{X}_e$. We also normalize this gain by the item's cost. See Equation~\eqref{eq:greedy-choice} for the definition. In lower bounding the total score $G$, we use a variant of a sampling lemma from \cite{AAK19} as well as the constant-factor adaptivity gap for submodular maximization~\cite{BSZ19}. We also need to partition the outcome space (of all previous realizations) into ``good'' and bad outcomes: conditional on a good outcome, $\OPT$ has a high probability of completing before cost $2^i$. Good outcomes are necessary in our proof of the sampling lemma, but luckily the total probability of bad outcomes is small (and they can be effectively ignored). In upper bounding $G$, we consider the total score as a sum over decision/sample paths and use the fact that the sum of relative gains corresponds to a harmonic series. 

In the correlated case, the score of any item $\rv{X}_e$ is the sum of two terms (i) its expected relative marginal gain as in the independent case, and (ii) an estimate of the probability on ``eliminated'' scenarios. Both terms are needed because the algorithm needs to balance (i) covering the function and (ii) identifying the realized scenario (after which it is trivial to cover $f$). Again, we normalize by the item's cost. See Equation~\eqref{eq:scn-greedy-choice}. 
In lower bounding the total score $G$, we partition the outcome space into good/okay/bad outcomes that correspond to a high conditional probability of $\OPT$ (i) covering function $f$ by cost $2^i$, (ii) eliminating a constant fraction of scenarios by cost $2^i$, or (iii) neither of the two cases. Further, by restricting to outcomes that have a ``large'' number of compatible scenarios (else, the algorithm's stopping rule would apply), we can bound the number of ``relevant'' outcomes by $s^{1/r}$. Then we consider $\OPT$ (up to cost $2^i$) conditional on all good/okay outcomes and show that one of these items has a high score. To upper bound $G$, we again  consider the total score as a sum over decision paths.

\ifICMLVersion
\subsection{Other related work} 
The   role of adaptivity has been extensively studied for various stochastic ``maximization'' problems such as knapsack~\cite{DGV08,BGK11}, matching~\cite{BGLMNR12,BehnezhadDH20}, matroid-constrained probing~\cite{GN13} and submodular-maximization~\cite{AsadpourN16,GuptaNS17,BSZ19}. In all these cases,  constant-factor adaptivity gaps are known. 

Recently, there have been several results examining the role of adaptivity in (deterministic) submodular optimization~\cite{BalkanskiS18,BalkanskiBS18,BalkanskiS18b,BalkanskiRS19,ChekuriQ19}. The motivation in these works was  parallelizing function queries that are often expensive. In many settings, there are now algorithms using (poly)logarithmic number of rounds that nearly match the best sequential (or fully adaptive) approximation algorithms. While our motivation is similar (in parallelizing the expensive probing steps), the techniques used are completely different.
\fi

\subsection{Related Work} 

A $(1+\ln Q)$-approximation algorithm for the basic submodular cover problem was obtained in \cite{W82}. This is also the best possible (unless P=NP) as the set cover problem is a special case~\cite{F98}. In the past several years, there have been many papers on {\em stochastic} variants of submodular cover~\cite{GK11,INZ12,DHK16,GHKL16,NavidiKN20} as this framework captures many different applications.

The stochastic set cover problem was first studied in \cite{GV06}, which showed that the adaptivity gap lies between $\Omega(d)$ and $O(d^2)$ where $d$ is the number of objects to be covered. Recently, \cite{AAK19} improved the upper bound to $O(d\, \log d\, \log(mc_{max}))$. As a corollary of Theorem~\ref{thm:main-ssc}, we obtain a tighter $O(d\log d)$ adaptivity gap. Importantly, we eliminate the dependence of the items $m$ and maximum-cost, which may even be exponential in $d$. 

An $O(\log Q)$ adaptive approximation algorithm for stochastic submodular cover (independent items) follows from \cite{INZ12}. Related results for special cases or with weaker bounds were obtained in \cite{LiuPRY08,GK11,GolovinK-arxiv,DHK16}. As mentioned earlier, upper/lower bounds for algorithms with limited rounds-of-adaptivity were obtained in \cite{AAK19}.  Theorem~\ref{thm:main-ssc} improves on these upper bounds by an $O(r\cdot \log (mc_{max}))$ factor. Moreover, our algorithm is greedy-style and much simpler to implement. We bypass computationally expensive steps in prior work such as solving several instances of stochastic submodular maximization. Our analysis (outlined above) is also very different. We use a variant of a key sampling lemma from \cite{AAK19}, but it is applied in different manner and only affects the analysis.  Our high-level approach of lower/upper bounding the total score is similar to the analysis in \cite{INZ12} for the fully adaptive problem. However, the details are very different because we need to handle  issues of $r$-round-adaptivity gaps. 

The scenario submodular cover problem was introduced in \cite{GHKL16} as a common generalization of several problems including optimal decision tree \cite{GG74,GNR17}, equivalence class determination \cite{CLS14} and decision region determination \cite{JavdCKKBS14}. An $O(\log (sQ))$ fully adaptive algorithm was obtained in \cite{GHKL16}. The same approximation ratio (in a more general setting) was also obtained in \cite{NavidiKN20}. In the correlated setting we are  not aware of any prior work on limited rounds of adaptivity (when the number of rounds $r<\log s$) . Some aspects of our analysis (e.g., good/okay/bad outcomes) are similar to \cite{NavidiKN20}, but  more work is needed as we have to bound the $r$-round-adaptivity gap.

The framework of ``adaptive submodularity'', introduced by \cite{GolovinK-arxiv}, models correlations in stochastic submodular cover in  a different way. Adaptive submodularity is a combined condition on the goal function $f$ (that needs to be covered)  and the distribution $\D$ on items. While  stochastic submodular cover with independent items  satisfies adaptive-submodularity, scenario submodular cover does not.  \cite{GolovinK-arxiv} gave a fully-adaptive $O(\log^2(mQ))$-approximation algorithm for adaptive-submodular cover (\asc) when all costs are unit. Recently,   \cite{EsfandiariKM19} improved this to  an $O(\log (mQ))$-approximation in  $O(\log m \, \log(Qm))$  rounds of adaptivity, still assuming  unit costs.  When costs are arbitrary, the  result in \cite{EsfandiariKM19} implies  an $O(\log (mQc_{\max}))$-approximation in $O(\log m \, \log(Qmc_{\max}))$ rounds: this   result also applies to  stochastic submodular cover and  can be compared to Corollary~\ref{cor:indep} that we obtain.  Although  scenario submodular cover is {\em not} a special case of \asc,   \cite{GHKL16} showed that scenario submodular cover can be re-formulated as \asc with a different goal function that is  
adaptive-submodular. However, this new goal function  has a larger ``$Q$ value'' of $\frac{Q}{p_{min}}$. So, when the algorithm from \cite{EsfandiariKM19} is applied  to this new goal function, it only implies  an  $O(\log (mQ\frac{c_{\max}}{p_{\min}}))$-approximation in $O(\log m \, \log(Qm \frac{c_{\max}}{p_{\min}}))$  rounds (this can be compared to Corollary~\ref{cor:scn}). To the best of our knowledge, there are no algorithms  for \asc using fewer than squared-logarithmic rounds of adaptivity.

The role of adaptivity has been extensively studied for stochastic submodular {\em maximization}. A constant adaptivity gap under matroid constraints (on the probed items) was obtained in \cite{AsadpourN16}. Later, \cite{GuptaNS17} obtained a constant adaptivity gap for a very large class of ``prefix closed'' constraints; the constant factor was subsequently improved to $2$ which is also the best possible~\cite{BSZ19}. We make use of this result in our analysis. More  generally, the role of adaptivity has been extensively studied for various stochastic ``maximization'' problems such as knapsack~\cite{DGV08,BGK11}, matching~\cite{BGLMNR12,BehnezhadDH20}, probing~\cite{GN13} and orienteering~\cite{GuhaM09,GuptaKNR15,BansalN15}.

Recently, there have been several results examining the role of adaptivity in (deterministic) submodular optimization~\cite{BalkanskiS18,BalkanskiBS18,BalkanskiS18b,BalkanskiRS19,ChekuriQ19}. The motivation here was in parallelizing function queries (that are often expensive). In many settings, there are algorithms using (poly)logarithmic number of rounds that nearly match the best sequential (or fully adaptive) approximation algorithms. While our motivation is similar (in parallelizing the expensive probing steps), the techniques used are completely different.

\section{Definitions}\label{sec:prelim}

In the stochastic submodular cover problem, the input is a collection
of $m$ random variables (or \emph{items})
$\rv{X} = \{\rv{X}_1, \ldots, \rv{X}_m\}$. Each item $\rv{X}_i$ has a
cost $c_i\in \R_+$, and realizes to a random element of groundset
${U}$. Let the joint distribution of $\rv{X}$ be denoted by $\D$. The
random variables $\rv{X}_i$ may or may not be independent; we discuss
this issue in more detail below.  The realization of item $\rv{X}_i$
is denoted by $X_i\in U$; this realization is only known when
$\rv{X}_i$ is {\em probed} at a cost of $c_i$. Extending this
notation, given a subset of items $\rv{S}\sse \rv{X}$, its realization
is denoted $S=\{X_i \,:\, \rv{X}_i\in \rv{S}\}\sse U$.

In addition, we are given an integer-valued monotone submodular
function $f:2^U \rightarrow \mathbb{Z}_+$ % on groundset $U$
with $f(U)=Q$.
A  realization $S$ of items $\rv{S} \sse \rv{X}$ is \emph{feasible} if
and only if $f(S)=Q$ the maximal  value; in this case, we also say that $\rv{S}$ \emph{covers} $f$. 
The goal is to probe (possibly adaptively) a subset $\rv{S} \sse
\rv{X}$ of items that gets realized to a feasible set. We use the
shorthand $c(\rv{S}) := \sum_{i: \rv{X}_i \in \rv{S}} c_i$ to denote the total cost of items in  $\rv{S}\sse \rv{X}$. The objective is to minimize the expected cost of probed items, where the expectation is taken over the randomness in $\rv{X}$. 
We consider the following types of solutions. 
\begin{definition}\label{def:r-round}
  For an integer $r\ge 1$, an {\bf $r$-round-adaptive} solution 
  proceeds in $r$ rounds of adaptivity. In each round
  $k\in \{1,\ldots, r\}$, the solution specifies an ordering of all
  remaining items and probes them in this order until some stopping
  rule. The decisions in round $k$ can depend on the realizations seen
  in all previous rounds $1,\ldots, k-1$.
\end{definition}

Setting $r=1$ gives us a \emph{non-adaptive} solution, and setting
$r=m$ gives us a \emph{(fully) adaptive} solution. 
Having more rounds leads to a smaller objective value, so  adaptive
solutions have the least objective value. Our performance guarantees
are relative to an optimal fully adaptive solution; 
let $\OPT$ denote this solution and its cost. 
The \emph{$r$-round-adaptivity gap} is defined as follows:
$$\sup_{\text{instance } I} \frac{\E[\text{cost of best $r$-round adaptive solution on }I]}{\E[\text{cost of best adaptive solution on }I]}$$ 
Setting $r=1$ gives the \emph{adaptivity gap}.

\paragraph{Independent and Correlated Distributions}
We first study the case where the random variables $\rv{X}$ are \emph{independent}. 
In keeping with existing literature~\cite{INZ12,DHK16,AAK19}, we refer
to this problem simply as the \emph{stochastic submodular cover}
problem.
We then consider the case when  the random variables $\rv{X}$ are correlated with a joint distribution of polynomial size,
and refer to it as the \emph{scenario submodular cover} problem~\cite{GHKL16}. 

\ifICMLVersion
\vspace{-0.1in}
\else\fi

%%%%%%%%%%%%%%%%%%%%%%%%%%%%

%%%%%%%%%%%%%% independent case %%%%%%%%%%%%%%%%%%%%%

\section{Stochastic Submodular Cover}\label{sec:submod-cover-r-adaptive}

We now consider the (independent) 
stochastic submodular cover problem and prove Theorem~\ref{thm:main-ssc}. 
For simplicity, we assume that costs $c_i$ are integers. Our results also hold for arbitrary costs (by replacing certain summations in the analysis by integrals).

We find it convenient to solve a \emph{partial cover} version of the
stochastic submodular cover problem. In this partial cover version, we
are given a parameter $\delta \in [0,1]$ and the goal is to probe
some $\rv{R}$ that realizes to a set $R$ achieving value
$f(R) > Q(1-\delta)$. We are interested in a {\em non adaptive} algorithm for this problem. Clearly, if $\delta = 1/Q$, the integrality of
the function $f$ means that $f(R) = Q$, and we solve the original (full coverage)
problem. Moreover, we can use this algorithm with different thresholds
to also get the $r$-round version of the problem. The main result of
this section is:

\begin{theorem}\label{thm:partial-cover}
  There is a non-adaptive algorithm for the partial cover version of
  stochastic submodular cover with cost
  $O\big(\frac{\ln 1/\delta}{\delta}\big)$ times the 
  optimal adaptive cost for the (full) submodular cover.
\end{theorem}

The algorithm first creates an ordering/list $L$ of the
items non-adaptively; that is, without knowing the realizations of the
items. To do so, at each step we pick a new item that maximizes a
carefully-defined score function (Equation~\eqref{eq:greedy-choice}).
The score of an item cannot depend on the realizations of 
previous items on the list (since we are non-adaptive). However,  it depends on  the \emph{marginal relative increase} for a random draw
from the previously chosen items. Once this ordering $L$ is
specified, the algorithm starts probing and realizing the items in this order, and does so until the realized value exceeds $(1-\delta)Q$.

\ifICMLVersion
\begin{algorithm}[h]
\caption{PARtial Covering Algorithm \textsc{ParCA}$(\rv{X}, f, Q, \delta)$}  \label{alg:parca}
\begin{algorithmic}[1]
\State $\rv{S} \leftarrow \emptyset$, list $L \leftarrow \langle
  \rangle$, $\tau \gets Q(1-\delta)$
\While{$\rv{S}\ne \rv{X}$} \Comment{Building the list non-adaptively}
\State select an item $\rv{X}_e \in \rv{X} \setminus \rv{S}$ that maximizes:
\begin{align}\label{eq:greedy-choice}
\text{score}(\rv{X}_e) :=  \frac1{c_e} \cdot &\sum_{S \sim \rv{S} : f(S)\le \tau } \pr(\rv{S}=S)\cdot \notag \\
    &\E_{X_e \sim \rv{X}_e} \left[  \frac{f(S \cup X_e) - f(S)}{Q - f(S)} 
 \right]
\end{align}
\State $\rv{S} \leftarrow \rv{S} \cup \{\rv{X}_e\}$ and list $L\gets L \circ \rv{X}_e$
\EndWhile
\State  $\rv{R} \gets \emptyset$, $R \gets \emptyset$
\While{$f(R) \leq \tau$} \Comment{Probing items on the list}
\State $\rv{X}_e \gets$ first r.v.\ in list $L$ not in $\rv{R}$
\State $X_e \in U$ be the realization of $\rv{X}_e$
\State $R \gets R \cup \{ X_e \}, \rv{R} \gets \rv{R} \cup \{ \rv{X}_e
\}$
\EndWhile
\State return probed items $\rv{R}$ and their realizations $R$.
\end{algorithmic}
\end{algorithm}

\else
\begin{algorithm}[h]
\caption{PARtial Covering Algorithm \textsc{ParCA}$(\rv{X}, f, Q, \delta)$} \label{alg:parca}
\begin{algorithmic}[1]
  \State $\rv{S} \leftarrow \emptyset$, list $L \leftarrow \langle
  \rangle$, $\tau \gets Q(1-\delta)$
\While{$\rv{S}\ne \rv{X}$} \Comment{Building the list non-adaptively}
\State \label{step:parca-greedy} select an item $\rv{X}_e \in \rv{X} \setminus \rv{S}$ that maximizes:
\begin{equation}\label{eq:greedy-choice}
\text{score}(\rv{X}_e) :=  \frac1{c_e} \cdot \sum_{S \sim \rv{S} : f(S)\le \tau } \pr(\rv{S}=S)\cdot \E_{X_e \sim \rv{X}_e} \left[  \frac{f(S \cup X_e) - f(S)}{Q - f(S)} 
 \right]
\end{equation}
\State $\rv{S} \leftarrow \rv{S} \cup \{\rv{X}_e\}$ and list $L\gets L \circ \rv{X}_e$
\EndWhile
\State  $\rv{R} \gets \emptyset$, $R \gets \emptyset$
\While{$f(R) \leq \tau$} \Comment{Probing items on the list}
\State $\rv{X}_e \gets$ first r.v.\ in list $L$ not in $\rv{R}$, and
let $X_e \in U$ be its realization.
\State $R \gets R \cup \{ X_e \}, \rv{R} \gets \rv{R} \cup \{ \rv{X}_e
\}$
\EndWhile
\State return probed items $\rv{R}$ and their realizations $R$.
\end{algorithmic}
\end{algorithm}
\fi

Given this partial covering algorithm we immediately get an algorithm
for the $r$-round version of the problem, where we are allowed to make
$r$ rounds of adaptive decisions. Indeed, we can first set
$\delta = Q^{-1/r}$ and solve the partial covering problem with this
value of $\delta$. Suppose we probe  variables $\rv{R}$ and their
realizations are given by the set $R \sse U$. Then we can condition on
these values to get the marginal value function $f_R$ (which is submodular), and
inductively get an $r-1$-round solution for this problem. The following algorithm formalizes this.
\begin{algorithm}[h]
\caption{$r$-round adaptive algorithm for stochastic submodular cover \textsc{SSC}$(r, \rv{X}, f)$} \label{alg:r-round-ind}
\begin{algorithmic}[1]
\State run  \textsc{ParCA} $(\rv{X}, f, Q, Q^{-1/r})$ for round \#1. \label{step:round-1}
\State let $\rv{R}$ (resp.\ $R$) denote the probed items (resp.\ their realizations) in \textsc{ParCA}.  
\State define residual submodular function $\fh := f_R$.
\State recursively solve \textsc{SSC}$(r-1, \rv{X}\setminus \rv{R}, \fh)$. \label{step:round-rec} 
\end{algorithmic}
\end{algorithm}

\begin{theorem}\label{thm:ssc-alg}
Algorithm~\ref{alg:r-round-ind} is an $r$-round adaptive algorithm for stochastic submodular
  cover with cost $O(Q^{1/r} \log Q)$ times the  optimal
  adaptive cost.
\end{theorem}
\ifICMLVersion
The proofs of \Cref{thm:partial-cover} and \Cref{thm:ssc-alg} can be found in the full version of the paper.
\else
\begin{proof} We proceed by induction on the number of rounds $r$. Let $\OPT$ denote the cost of an optimal adaptive solution. 
The base case is $r=1$, in which case $\delta=Q^{-1/r} = \frac1Q$. By Theorem~\ref{thm:partial-cover}, the partial cover algorithm \textsc{ParCA}$(\rv{X}, f, Q, Q^{-1/r})$ obtains a realization $R$ with $f(R)>(1-\delta)Q = Q-1$. As $f$ is integer-valued, we must have $f(R)=Q$, which means the function is fully covered. So the algorithm's expected cost is $O(Q\log Q)\cdot \OPT$.

We now consider $r>1$ and assume (inductively) that Algorithm~\ref{alg:r-round-ind} finds an $r-1$-round $O(Q^{\frac1{r-1}} \log Q)$-approximation algorithm for any instance of stochastic submodular cover. Let $\delta = Q^{-1/r}$. By Theorem~\ref{thm:partial-cover}, the expected cost in round 1 (step~\ref{step:round-1} in Algorithm~\ref{alg:r-round-ind}) is $O(\frac1r Q^{1/r} \log Q)$.  Let $\Qh:= Q-f(R) = \fh(U)$ denote the maximal value of the residual submodular function $\fh=f_R$. Note that  $\Qh<\delta Q = Q^{(r-1)/r}$, by the definition of the partial covering
  algorithm.  The optimal
  solution $\OPT$ conditioned on the variables in $\rv{R}$ realizing to $R$ gives a feasible adaptive solution to the residual
  problem of covering $\fh$; we denote this conditional solution by
  $\widehat{\OPT}$. We inductively get that the cost of our
  $r-1$-round solution on $\fh$ is at most 
  $$O(\Qh^{\frac1{r-1}} \log \Qh)\cdot  \widehat{\OPT} \le O\left(\frac{r-1}r \; Q^{1/r} \log Q\right)\cdot \widehat{\OPT},$$ where we used  $\Qh < Q^{(r-1)/r}$.  As this holds for every realization $R$,  we
  can take expectations over $R$ to get that the (unconditional)
  expected cost of the last
  $r-1$ rounds is  
$O\big(\frac{r-1}r \; Q^{1/r} \log Q\big) \cdot \OPT$. 
  Adding to this the cost of the first round, which is $ O\big(\frac1r Q^{1/r} \log Q\big)$, completes the proof.
\end{proof}
\fi
\ifICMLVersion
\vspace{-0.1in}
\else\fi
\paragraph{Remark:} Assuming that the scores~\eqref{eq:greedy-choice} can be computed in polynomial time, it is clear that our entire algorithm can be implemented in polynomial time. In particular, if $T$ denotes the time taken to calculate the score of one item, then the overall algorithm runs in time $poly(m,T)$ where $m$ is the number of items. We are not aware of a closed-form expression for the scores (for arbitrary submodular functions $f$). 
\ifICMLVersion
However, as discussed in the full version, we can use sampling to estimate these scores to within a constant factor. 
\else
However, as discussed in \S\ref{app:sampling}, we can use sampling to estimate these scores to within a constant factor. 
\fi
Moreover,  our analysis  works even if we only choose an  approximate maximizer for \eqref{eq:greedy-choice} at each step. It turns out that  $T=poly(m, c_{max})$ many samples suffice to estimate these scores. So,  the final runtime is $poly(m, c_{max})$; note that this does not depend on the number $|U|$ of elements. We note that even the previous algorithms \cite{AAK19,EsfandiariKM19} have a polynomial dependence on $c_{max}$ in their runtime. 
\ifICMLVersion
In the following analysis we will assume that the scores \eqref{eq:greedy-choice} are computed exactly. 
\else
In the following analysis we will assume that the scores \eqref{eq:greedy-choice} are computed exactly; see \S\ref{app:sampling} for the sampling details. 
\fi

\subsection{Analysis for a Call to \textsc{ParCA}} \label{subsec:analysis-PARCA}
We now prove Theorem~\ref{thm:partial-cover}. 

We denote by $\OPT$ an optimal adaptive solution for the covering problem on $f$.  
Now we analyze the cost incurred by following the
non-adaptive strategy (which we call $\NA$): probe variables
according to the order given by the list $L$ generated by
\textsc{ParCA}, and stop when the realized coverage exceeds $\tau := Q(1 - \delta)$ (see \Cref{alg:parca} for details). We consider the expected cost of this strategy, and relate it to the cost
of ${\OPT}$.

We refer to the cumulative cost incurred (either by ${\OPT}$
or by $\NA$) until any point in a solution as
\emph{time} elapsing. We say that ${\OPT}$ is in phase $i$ when it is
in the time interval $[2^i, 2^{i+1})$ for $i \geq 0$. Let
$\alpha := O\left(\frac{\ln 1/\delta}{\delta}\right)$. We say that $\NA$ is in
\emph{phase} $i$ when it is in time interval
$[\alpha\cdot 2^{i-1}, \alpha\cdot 2^{i})$ for $i \geq 1$; phase $0$
refers to the interval $[1, \alpha)$. Define
\ifICMLVersion
\vspace{-0.1in}
\else\fi
\begin{itemize}
    \item $u_{i}$: probability that $\NA$ goes beyond phase $i$; that is, has cost at least $\alpha \cdot 2^i$.
    \item $u_{i}^*$: probability that $\OPT$ goes beyond phase $i-1$; that is, costs at least
      $2^{i}$.
\end{itemize}
Since all costs are integers, $u_0^* = 1$. For ease of notation, we also use $\OPT$ and $\NA$ to denote the \emph{random} cost incurred by $\OPT$ and $\NA$ respectively. The following lemma forms
the crux of the analysis.
\ifICMLVersion
\begin{lemma}\label{lem:key2}
For any phase $i \geq 1$,   $ u_{i} \leq \frac{u_{i-1}}{4} + \frac65 u_{i}^*. $
\end{lemma}
\else
\begin{lemma}\label{lem:key2}
For any phase $i \geq 1$, we have $ u_{i} \leq \frac{u_{i-1}}{4} + \frac65 u_{i}^*. $
\end{lemma}
\fi
\ifICMLVersion
\else
Before we prove this lemma, we use it to prove \Cref{thm:partial-cover}.
\begin{proof}[Proof of \Cref{thm:partial-cover}]
  With probability $(u_{i-1} - u_{i})$, $\NA$ ends in phase $i$
  with cost at most $\alpha \cdot 2^{i}$. As a consequence of this, we
  have
  \ifICMLVersion
  \begin{align}
    \E[\NA] \ &\leq \  \alpha \cdot (1 - u_{0}) +  \sum_{i \geq 1} \alpha \cdot 2^{i} \cdot (u_{i-1} - u_{i}) \notag \\
                      =& \ \alpha + \alpha \cdot \sum_{i \geq 0} 2^i u_{i}. \label{eq:alg-ub}
  \end{align}
  \else
  \begin{align}
    \E[\NA] \ &\leq \  \alpha \cdot (1 - u_{0}) +  \sum_{i \geq 1} \alpha \cdot 2^{i} \cdot (u_{i-1} - u_{i}) 
                      = \ \alpha + \alpha \cdot \sum_{i \geq 0} 2^i u_{i}. \label{eq:alg-ub}
  \end{align}
  \fi
  Similarly, we can bound the cost of the optimal adaptive algorithm
  as
  \ifICMLVersion
  \begin{align}
    \E[{\OPT}] &\geq \sum_{i \geq 0} 2^{i} (u_i^* - u_{i+1}^*) \notag \\
    &\geq u_0^* + \frac{1}{2}\cdot \sum_{i \geq 1} 2^i u_i^* = 1 + \frac{1}{2}\cdot \sum_{i \geq 1} 2^i u_i^*, \label{eq:opt-lb}
  \end{align}
  \else
  \begin{equation}\label{eq:opt-lb}
    \E[{\OPT}] \geq \sum_{i \geq 0} 2^{i} (u_i^* - u_{i+1}^*) \geq u_0^* + \frac{1}{2}\cdot \sum_{i \geq 1} 2^i u_i^* = 1 + \frac{1}{2}\cdot \sum_{i \geq 1} 2^i u_i^*,
  \end{equation}
  \fi
  where the final
  equality uses the fact that $u_{0}^* = 1$. Define $\Gamma := \sum_{i \geq 0} 2^i u_{i}$. We have
  \begin{align*}
    \Gamma := \sum_{i \geq 0} 2^i u_{i}
    &\leq u_{0} + \frac{1}{4}  \sum_{i \geq 1} 2^i \cdot u_{i-1} +  \frac65 \sum_{i \geq 1}2^i \cdot u_i^* \\
    &\leq u_{0} +  \frac{1}{4} \sum_{i \geq 1} 2^i u_{i-1} + \frac{12}5 \cdot (    \E[{\OPT}]  - 1)\\
    &= u_{0} + \frac{1}{2}\Big(\sum_{i\geq 0}2^i u_{i}\Big) +
      \frac{12}5 \cdot (
      \E[{\OPT}]  - 1)\\ 
      &\leq\,\, \frac{1}{2}\Gamma +
      \frac{12}5\, \E[{\OPT}] - 1,
  \end{align*}
  where the first inequality follows from \Cref{lem:key2}, the
  second inequality from \eqref{eq:opt-lb}, and the last
  inequality from the fact that $u_{0} \leq 1$. Thus,
  $\Gamma \leq \frac{24}5\cdot \E[{\OPT}] - 2$. From
  \eqref{eq:alg-ub}, we conclude
  $\E[\NA] \leq \frac{24}5\alpha \cdot \E[{\OPT}]$. Setting $\alpha =O\Big(\frac{\ln(1/\delta)}{\delta}\Big)$ completes the proof.
\end{proof}
\fi

\ifICMLVersion
This lemma implies \Cref{thm:partial-cover}. We include the proofs of \Cref{thm:partial-cover} and \Cref{lem:key2} in the full version of the paper.
\else
\section{Proof of the Key Lemma (\Cref{lem:key2})}

Recall the setting of \Cref{lem:key2} and fix the phase $i \geq 1$. Consider the list $L$ generated
by \textsc{ParCA}$(\rv{X}, f, Q, \delta)$. Let $\NA$ denote both the non-adaptive
strategy of \Cref{alg:parca}, as well as its cost. Note that $\NA$ probes
items in the order given by $L$ until a coverage value greater than
$\tau:=Q(1-\delta)$ is achieved. 

We will assume (without loss of generality) that $\delta$ is a power-of-two, i.e., $\delta=2^{-z}$ for some integer $z\ge 0$. Indeed, if this is not the case, we can use a power-of-two value $\delta'$ where $\frac{\delta}{2}\le \delta'\le \delta$: this only increases the approximation ratio $O(\frac{1}{\delta}\ln\frac{1}{\delta})$ in Theorem~\ref{thm:partial-cover} by a constant factor.

For each time $t\ge 0$,
let $\rv{X}_{e(t)}$ denote the item that  would be selected at time $t$. In other words, this is the item  which causes the cumulative cost of $L$ to exceed $t$ for the first time.   
We define the {\em total gain} as the random variable 
\begin{gather}
  G \,\,:=\,\, \sum_{t = \alpha 2^{i-1}}^{ \alpha 2^{i}}
  \text{score}(\rv{X}_{e(t)}), \label{eq:total-gain}
\end{gather}
which corresponds to the sum of scores over the time interval $[\alpha
\cdot 2^{i-1}, \alpha \cdot 2^{i})$.
The proof of \Cref{lem:key2} will be completed by giving upper- and
lower-bounds on $G$, which we do next. The lower bound views $G$ as a sum over time steps, whereas the upper bound  views $G$ as a sum over decision paths.

\subsection{A Lower Bound for $G$} 
Fix some time $t\in [\alpha 2^{i-1}, \alpha 2^i)$ in phase $i$ of our algorithm.
Let $\rv{S}$ be the random variable denoting the set of chosen items (added to list $L$)
until some time $t$, and let $S$ denote its realization.  We also need the following definitions:
\begin{enumerate}
    \item For any power-of-two $\theta \in \{\frac{1}{2^j} : 0\le j\le \log Q\}$, we say that a realization $S$ of $\rv{S}$ belongs to \emph{scale} $\theta$ iff $ \frac{\theta}{2} <  \frac{{Q} - f(S)}{{Q}} \leq \theta$.
    We use $\calE_{\theta}$ to denote all outcomes of $\rv{S}$ that belong to scale $\theta$. We also use $S\sim \calE_{\theta}$ to denote the conditional distribution of $\rv{S}$ corresponding to $\calE_{\theta}$. 
    
    \item For any scale $\theta$, let  $r_{i\theta}^* := \pr(\text{$\OPT$ covers $f$ with cost at most $2^i$ \bf{and} $S \in \calE_{\theta}$})$.
\item 
Scale $\theta$ is called {\em good} if $\frac{r_{i\theta}^*}{\pr(\calE_{\theta})} \geq \frac{1}{6}$ where $\pr(\calE_{\theta}):=\pr_{S\sim \rv{S}}(S\in \cal\calE_{\theta})$. 
\end{enumerate}

\begin{lemma}\label{lem:stop-scale}
  $\NA$  terminates by  time $t$ if and only if the realization of $\rv{S}$ is in a scale $\theta\le \delta$.
\end{lemma}
\begin{proof}
Note that if the realization of $\rv{S}$ is in some scale $\theta\ge 2\delta$ then 
$$f(S)<Q-\frac{\theta Q}{2} \le Q-\delta  Q =\tau,$$
and $\NA$ would not   terminate  before time $t$. On the other hand, if the realization of $\rv{S}$ is in a scale $\theta \le \delta$ (note that all scales are powers-of-two) then $f(S)\ge Q-\theta Q \ge \tau,$  
and $\NA$ would   terminate before time $t$. 
\end{proof}

We now define a function
\begin{equation}
\label{eq:exp-func-g}
g({T}) :=  \sum_{\theta > \delta} \pr(\calE_{\theta})\cdot \E_{S\sim \calE_{\theta}} \left[ \frac{f(S \cup T) - f(S)}{{Q} - f(S)} \right]
= \sum_{S \sim \rv{S} : f(S)< \tau } \pr(\rv{S}=S)\cdot \frac{f_{S}(T)}{{Q} - f(S)} 
, \quad \forall \, T\sse U.\end{equation} 
The second equality is by Lemma~\ref{lem:stop-scale}. 
The function $g$ is monotone and submodular, since $f_S$ is  
monotone and submodular for each  $S \sse U$, and $g$ is a non-negative linear combination of such functions. 
Moreover, for any item $\rv{X}_e$, we have $\text{score}(\rv{X}_e) =
\frac1{c_e}\cdot \E_{\rv{X}_e}[ g(X_e)]$.

\paragraph{Constrained stochastic submodular maximization.} Our analysis makes use of some known results for stochastic submodular {\em maximization}. Here, we are given as input a non-negative monotone submodular function $g: 2^U \to \mathbb{R}_{\geq 0}$, independent stochastic items $\rv{X} = \{\rv{X}_1, \ldots, \rv{X}_m\}$ such that each $\rv{X}_i$ realizes to some element of the groundset $U$. There is a cost $c_i$ associated with each item $\rv{X}_i$, and a budget $B$ on the total cost. The goal is to select $\rv{S} \sse \rv{X}$ such that the total cost of $\rv{S}$ is at most $B$ and it maximizes the expected value $\E_{S \sim \rv{S}}[g(S)]$. The adaptivity gap for stochastic submodular maximization is at most $2$; see Theorem~1 in~\cite{BSZ19}.

\begin{lemma}
  \label{lem:key-key}
  For any good scale $\theta$, there exists a subset
  $\rv{T_\theta}\sse \rv{X} \setminus \rv{S}$ with
  $c(\rv{T_\theta})\le \frac{144}{\theta}\cdot 2^i$ such that:
  \begin{equation}\label{eq:key-lem}
     \E_{S \sim \calE_{\theta}} \E_{T_{\theta}\sim \rv{T_\theta}} \Big[f(S \cup  T_{\theta}) - f(S)\Big] \geq \frac{\theta {Q}}{6} \cdot \frac{r^*_{i\theta}}{\pr(\calE_{\theta})}.
  \end{equation} 
\end{lemma}

\begin{proof}
  We construct set $\rv{T_\theta}$ as follows.  Initially,
  $\rv{T}_{\theta} \leftarrow \emptyset$.  For each
  $k = 1, 2, \ldots, \frac{144}\theta$:
  \begin{enumerate}
  \item Sample realization $S$ of $\rv{S}$ from $\calE_{\theta}$.
  \item Let $\rv{T}_k\sse \rv{X}\setminus \rv{S}$ be an optimal {\em
      non-adaptive} solution to the stochastic submodular maximization
    instance with function $f_{S}$ and cost budget $2^i$.
  \item Set $\rv{T}_{\theta} \leftarrow \rv{T}_{\theta} \cup \rv{T}_k$.
  \end{enumerate}

  By construction, 
  $c(\rv{T_\theta})\le \frac{144}{\theta}\cdot 2^i$. So we focus on
  the expected function value.  Consider any 
  $S\in \calE_{\theta}$ as the realization of $\rv{S}$. Let
  $\rv{T}_{S}$ denote the non-adaptive solution obtained in step~2
  above for realization
  $S$. 
  Define $w^*_{i,S}$ as the probability that $\OPT$ covers $f$ with
  cost at most $2^i$ given that $\rv{S}$ realizes to $S$, i.e.,
$$ w^*_{i,S} := \pr(\text{ $\OPT$ covers $f$ with cost at most $2^i$ $\mid$ $\rv{S} = S$  }). $$ Note that $\sum_{S \in \calE_{\theta}} w^*_{i,S} \cdot \pr(\rv{S}=S ) = r^*_{i\theta}$. Let $\AD$ denote $\OPT$ until time $2^i$ and restricted to items $\rv{X}\setminus \rv{S}$. Note that $\AD$ is a  feasible adaptive solution to the stochastic submodular maximization instance with function $f_{S}$: every decision path has total cost at most $2^i$.  The expected value of $\AD$ is at least $({Q} -f(S))\cdot w^*_{i,S}$ by definition of $w^*_{i,S}$. As $S \in \calE_{\theta}$, we have ${Q} - f(S) \geq \frac{\theta {Q}}{2}$, which implies $\AD$ has value at least $\frac{\theta {Q}}{2}\cdot w^*_{i,S}$. Now, using the factor $2$ adaptivity gap for stochastic submodular {\em maximization} (discussed above), it follows that the non-adaptive solution $\rv{T}_S$ has expected value: 
\begin{equation}\notag
  \E_{T_{S}\sim \rv{T}_S}[{f({S} \cup T_{{S}}) - f({S})} ] \geq \frac{{\theta {Q}}}{4} \cdot  w^*_{i,S},\qquad \forall S\in \calE_{\theta}.
\end{equation}
Taking expectations,
\begin{align*}
  \E_{{S}} \E_{T_{S}}[{f({S} \cup T_{{S}}) - f({S})} \mid S \in \calE_{\theta}] &\geq \sum_{S \in \calE_{\theta}} \frac{{\theta {Q}}}{4} \cdot  w^*_{i,S} \cdot \pr(\rv{S} = S \mid S \in \calE_{\theta}) \\
                                                                                &= \frac{{\theta {Q}}}{4} \cdot \sum_{S \in \calE_{\theta}} w^*_{i,S} \cdot \frac{\pr(\rv{S} =S)}{\pr(\calE_{\theta})} 
                                                                                  = \frac{{\theta {Q}}}{4} \cdot \frac{r^*_{i \theta}}{\pr(\calE_{\theta})}.
\end{align*}
The left-hand-side of the above relation can be rewritten to give
\begin{equation}\label{eq:key-lemma-cont2}
  \E_{S \sim \calE_{\theta}} \E_{T_{S}}[{f({S} \cup T_{{S}}) - f({S})}] \geq \frac{{\theta {Q}}}{4} \cdot \frac{r^*_{i \theta}}{\pr(\calE_{\theta})}. 
\end{equation}
For a contradiction to \eqref{eq:key-lem}, suppose that
\begin{equation}\label{eq:key-lemma-cont}
  \E_{S \sim \calE_{\theta}} \E_{T_{\theta}} [f(S \cup  T_{\theta}) - f(S)] < \frac{\theta {Q}}{6} \cdot \frac{r^*_{i\theta}}{\pr(\calE_{\theta})}. 
\end{equation}
Subtracting Equation \eqref{eq:key-lemma-cont} from Equation
\eqref{eq:key-lemma-cont2} gives the following:
\begin{align}
  \frac{\theta {Q}}{12} \cdot \frac{r^*_{i \theta}}{\pr(\calE_{\theta})} &< \E_{S \sim \calE_{\theta}} \E_{T_{S}} \E_{T_{\theta}} [f(S \cup T_{S}) - f(S \cup T_\theta)]\notag \\
                                                                         &\leq \E_{S \sim \calE_{\theta}} \E_{T_{S}} \E_{T_{\theta}} [f(S \cup T_{\theta} \cup T_{{S}}) - f(S \cup T_\theta)]\notag \\
                                                                         &\leq \E_{S \sim \calE_{\theta}}\E_{T_{{S}}} \E_{T_{\theta}} [f(T_{\theta} \cup T_{{S}}) - f(T_\theta)] \label{eq:key-lemma-cont3}
\end{align}
where the second inequality uses monotonicity of $f$, and the
last one its submodularity.

Let $\rv{T}^{(k)} = \rv{T}_1 \cup \rv{T}_2 \ldots \cup \rv{T}_k$
denote the items selected until iteration $k$.
We write a telescoping sum as follows
$$ \E_{T_{\theta}}[f(T_{\theta})] = \sum_{k} \E_{T^{(k)}} \Big[f(T^{(k)})
- f(T^{(k-1)})\Big] $$ where we define $f(T^{(0)}) := 0$. Note that in
each iteration $k$, the sample $S$ is drawn independently and
identically from $\calE_{\theta}$, and items $\rv{T}_k =\rv{T}_{S}$
are added.
\begin{align*}
  \E_{T^{(k)}} \left[f(T^{(k)}) - f(T^{(k-1)})\right]   &=   \E_{S \sim \calE_{\theta}}\E_{T_{{S}}} \E_{T^{k-1}} \left[f(T^{(k-1)} \cup T_{{S}}) - f(T^{(k-1)})\right] \notag \\
                                                        &\geq    \E_{S \sim \calE_{\theta}}\E_{T_{{S}}} \E_{T_{\theta}} \left[f(T_{\theta} \cup T_{{S}}) - f(T_\theta)\right] \,\,>\,\,      \frac{\theta {Q}}{12} \cdot \frac{r^*_{i \theta}}{\pr(\calE_{\theta})} .
\end{align*}
The first inequality above uses $\rv{T}^{(k-1)}\sse \rv{T}_\theta$ and
submodularity, and the second inequality uses
\eqref{eq:key-lemma-cont3}.  Adding over all iterations $k$,
$$     \E_{T_{\theta}}[f(T_{\theta})] \,\,>\,\,  \sum_{k} \frac{\theta {Q}}{12} \cdot \frac{r^*_{i \theta}}{\pr(\calE_{\theta})} \,\, = \,\, \frac{144}{\theta} \cdot \frac{\theta {Q}}{12} \cdot \frac{r^*_{i \theta}}{\pr(\calE_{\theta})} \,\,\geq \,\,2{Q},$$
where the last inequality uses the fact that $\theta$ is a good
scale. This is a contradiction since the maximum function value is
${Q}$. This completes the proof of \eqref{eq:key-lem}.
\end{proof}

\begin{lemma}\label{lem:key-LB}
For any $\rv{S} \sse \rv{X}$, there exists a subset $\rv{T} \sse \rv{X} \setminus \rv{S}$  of total cost at most $\frac{144}\delta  \cdot  2^i$ with 
$$\E_{T\sim \rv{T}} [ g(T)] \quad \ge \quad \frac{1}{6} \cdot \sum_{\theta > \delta,\, \text{good}} r^*_{i\theta}.$$
\end{lemma}

\begin{proof}
Let ${\cal B}$ denote the set of {\em good} scales $\theta$ with
$\theta > \delta$. Note that $\sum_{\theta\in {\cal B}} \frac{1}{\theta}
\le \frac1\delta$ as the scales are powers of
two. From \Cref{lem:key-key}, let $\rv{T_\theta}$ denote the items
satisfying \eqref{eq:key-lem} for each scale $\theta\in {\cal
  B}$. Define $\rv{T}=\cup_{\theta\in {\cal B}} \rv{T_\theta}$. As
claimed, the
total cost is
\[ c(\rv{T})\le \sum_{\theta\in {\cal B}} c(\rv{T_\theta})\le 2^i \sum_{\theta\in {\cal B}} \frac{144}{\theta}\le \frac{144}\delta \cdot 2^i.\]
Next, we bound the expected increase in the value of $f$. By \eqref{eq:key-lem} and $\rv{T}\supseteq \rv{T_\theta}$, it follows that
$$\E_{S \sim \calE_{\theta}} \E_{T\sim \rv{T}}  [f(S \cup T) - f(S)] \geq \frac{\theta {Q}}{6} \cdot \frac{r^*_{i\theta}}{\pr(\calE_{\theta})}$$
for each $\theta\in {\cal B}$.
Using the fact that $\frac{\theta {Q}}2 < {Q}-f(S)\le \theta {Q}$ for all $S\in \calE_{\theta}$, we get:
$$\E_{S \sim \calE_{\theta}} \E_{T\sim \rv{T}}  \left[\frac{f(S \cup T) - f(S)}{{Q}-f(S)}\right] \geq \frac{1}{6} \cdot \frac{r^*_{i\theta}}{\pr(\calE_{\theta})}, \qquad \forall \theta\in {\cal B}.$$
By definition of function $g$ (see \eqref{eq:exp-func-g}),  
\begin{eqnarray*}
\E_{T\sim \rv{T}}[g(T)] &= &\sum_{\theta > \delta} \pr(\calE_{\theta})\cdot \E_{S\sim \calE_{\theta}} \E_{T\sim \rv{T}} \left[ \frac{f(S \cup T) - f(S)}{{Q} - f(S)} \right] \\
&\ge &\sum_{\theta\in {\cal B}} \pr(\calE_{\theta}) \cdot \E_{S \sim \calE_{\theta}} \, \E_{T\sim \rv{T}} \left[ \frac{f(S \cup T) - f(S)}{{Q}-f(S)}\right]  \, \geq \,  \frac{1}{6} \cdot \sum_{\theta \in {\cal B}} r^*_{i\theta},
\end{eqnarray*}
which completes the proof of the lemma.
\end{proof}

Using \Cref{lem:key-LB} and averaging,
\begin{eqnarray}
\max_{\rv{X}_e \in \rv{X}\setminus \rv{S}} \text{score}(\rv{X}_e) &\ge& \max_{\rv{X}_e \in \rv{T}} \text{score}(\rv{X}_e) \ge \frac{\sum_{\rv{X}_e\in \rv{T}} \E [g(X_e)]}{c(\rv{T})} = \frac{\E_{T\sim \rv{T}}[\sum_{X_e\in T} g(X_e)]}{c(\rv{T})}\notag \\
&\ge& \frac{\E_{T\sim \rv{T}}[g(T)]}{c(\rv{T})} \quad \ge \quad \frac{\delta}{\beta\cdot 2^i}\cdot \sum_{\theta > \delta,\, \text{good}} r^*_{i\theta},\label{eq:score:LB}
\end{eqnarray}
where $\beta=O(1)$.

\begin{lemma}\label{lem:score-per-step} We have $\sum_{\theta > \delta,\, \text{good}} r^*_{i\theta}\ge u_{i} - \frac{6u^*_i}{5}$.
\end{lemma}
\begin{proof}
First, we upper bound $\sum_{\theta \text{ not good}} r^*_{i\theta}$.  
Consider any scale $\theta$ that is not good. Then, $\frac{{r}^*_{i \theta}}{\pr(\calE_{\theta})} < \frac{1}{6}$, i.e.,  $\pr(\text{$\OPT < 2^i$} \mid S \sim \calE_{\theta}) < 1/6$, which implies  $\pr(\text{$\OPT  \geq 2^i$} \mid  S \sim \calE_{\theta}) > \frac56 > 5\frac{ r_{i\theta}^*}{\pr(\calE_{\theta})}$. So, 
\begin{align} \label{eq:bad-theta-prob}
   \sum_{\theta \text{ not good}} r^*_{i\theta} &< \frac{1}{5} \sum_{\theta\text{ not good}} \pr(\calE_{\theta})\cdot \pr(\OPT \geq 2^i \mid  S \sim \calE_{\theta})  \leq \frac{1}{5} \sum_{\theta } \pr(\OPT \geq 2^i \text{ \emph{and} } S \sim \calE_{\theta}) \leq \frac{u_i^*}{5} 
\end{align}
where the  last inequality uses  the fact that $\sum_{\theta }\pr(\OPT \geq 2^i \text{ \emph{and} } S \sim \calE_{\theta}) = u_i^*$.

We now upper bound $\sum_{\theta \le \delta} r^*_{i\theta}$. By Lemma~\ref{lem:stop-scale}, if the realization $S$ is in scale $\theta \le \delta$   then $\NA$ ends before time $t\le \alpha 2^i$, i.e., it does not go beyond phase $i$. Hence, 
\begin{equation}\label{eq:low-theta-prob}
\sum_{\theta \le \delta} r^*_{i\theta}\le \sum_{\theta \le \delta} \pr(S\sim \calE_{\theta})\le 1-u_{i}.
\end{equation}

We now use the fact that $1-u^*_i = \sum_\theta r^*_{i \theta}$ where we sum over all scales. So,
$$\sum_{\theta > \delta,\, \text{good}} r^*_{i\theta} \,\, \ge \,\,  \sum_{\theta} r^*_{i\theta} - \sum_{\theta \text{ not good}} r^*_{i\theta} -\sum_{\theta \le \delta} r^*_{i\theta}  \,\, \ge  \,\, (1-u^*_i) - \frac{u^*_i}{5} - (1-u_{i})  \,\, = \,\,  u_{i} - \frac{6u^*_i}{5},$$
where we used \eqref{eq:bad-theta-prob} and \eqref{eq:low-theta-prob}.
\end{proof}

Combining \eqref{eq:score:LB} and \Cref{lem:score-per-step}, and using the greedy choice in step~\ref{step:parca-greedy} (Algorithm~\ref{alg:parca}), the score at time $t$,
$$\text{score}(\rv{X}_{e(t)}) \quad \ge\quad  \frac{\delta}{\beta 2^i}\left(u_i - \frac65 u^*_i\right).$$
We note that this inequality continues to hold  (with a larger constant $\beta$) even if we choose an item that only maximizes the score~\eqref{eq:greedy-choice} within a constant factor.  

Using the above inequality  for each time $t$ during phase $i$, we have
\begin{equation}\label{eq:G-LB}
G\quad \ge\quad  \alpha 2^{i-1}\cdot \frac{\delta}{\beta 2^i}\left(u_i - \frac65 u^*_i\right)\quad =\quad \frac{\alpha \delta}{2\beta}\cdot \left(u_i - \frac65 u^*_i\right).
\end{equation}
We will use this lower bound for $G$ in conjunction with an upper
bound, which we prove next.

\subsection{An Upper Bound for $G$}
We now consider the implementation of the non-adaptive list $L$ and
calculate $G$ as a sum of contributions over the observed decision
path. Let $\Pi$ denote the (random) decision path followed by the non-adaptive strategy $\NA$: this consists of a prefix of $L$ along with their realizations. Denote by
$\langle X_1, X_2,\ldots,  \rangle$ the sequence of realizations (each
in $U$) observed on $\Pi$. So item $\rv{X}_j$ is selected between time
$\sum_{\ell=1}^{j-1} c_\ell$ and $\sum_{\ell=1}^{j} c_\ell$. Let $h$
(resp.\ $p$) index the first (resp.\ last) item in $\Pi$ (if any) that
is selected (even partially) during phase $i$, i.e., between time
$\alpha 2^{i-1}$ and $\alpha 2^{i}$. For each index $h\le j\le p$, let
$t_j$ denote the duration of time that item $\rv{X}_j$ is selected
during in phase $i$; so $t_j$ is the width of interval
$[\sum_{\ell=1}^{j-1} c_\ell,\, \sum_{\ell=1}^{j} c_\ell] \bigcap
[\alpha 2^{i-1},\, \alpha 2^{i}]$. It follows that $t_j \leq c_j$.

Define $G(\Pi):=0$ if index $h$ is undefined (i.e., $\Pi$ terminates before phase $i$), and otherwise:
\begin{equation}\label{eq:G-path-ub}
G(\Pi) :=\sum_{j=h}^p \frac{t_j}{c_j} \cdot \frac{f(\{X_1,\ldots,  X_j\}) - f(\{X_1,\ldots,  X_{j-1}\})}{{Q} - f(\{X_1,\ldots,  X_{j-1}\})}\le \sum_{j=h}^p  \frac{f(\{X_1,\ldots,  X_j\}) - f(\{X_1,\ldots,  X_{j-1}\})}{{Q} - f(\{X_1,\ldots,  X_{j-1}\})} .
\end{equation}
By the stopping criterion for $L$, the $f$ value {\em before} the end of $\Pi$ remains at most  $\tau={Q}(1-\delta)$. So the denominator above, i.e., ${Q} - f(\{X_1,\ldots,  X_{j-1}\})$ is at least $\delta Q$ for all $j$.

\begin{lemma}
  \label{lem:per-path}
For any decision path $\Pi$, we have $G(\Pi)\le 1+   \ln (1/\delta)$.
\end{lemma}
\begin{proof}
For each $h\le j\le p$, let $V_j:=f(\{X_1,\ldots,  X_j\})$; also let $V_0=0$.  For $j\le p-1$, as noted above, $V_j\le  \tau$; as $f$ is  integer-valued,    $V_j\in\{0,1,\cdots, \lfloor \tau\rfloor\}$. We have:
$$\sum_{j=h}^{p-1} \frac{V_j-V_{j-1}}{{Q}-V_{j-1}}\,\,\le\,\,\sum_{j=h}^{p-1}\,\,\, \sum_{y=0}^{V_j-V_{j-1}-1}\frac{1}{{Q}-V_{j-1} - y} \le \sum_{\ell=\delta Q+1}^{{Q}} \frac{1}{\ell} \le \ln\left(\frac{{Q}}{\delta {Q}}\right)= \ln (1/\delta).$$
The second inequality above uses $Q-V_{j-1}-y\ge Q-V_j+1\ge Q-\tau+1=\delta Q+1$. 
The right-hand-side of \eqref{eq:G-path-ub} is then:
$$\sum_{j=h}^{p-1} \frac{V_j-V_{j-1}}{{Q}-V_{j-1}}  \,\,+\,\,   \frac{V_p - V_{p-1}}{{Q} - V_{p-1}} \quad \le\quad 1+ \ln (1/\delta),$$
where we used $V_p\le {Q}$. This  completes the proof.
\end{proof}

Taking expectations over the various decision paths means
$G=\E_\Pi[G(\Pi)]$, so~\Cref{lem:per-path} gives
\begin{equation}
\label{eq:G-UB} G\,\, \le \,\,(1+\ln (1/\delta))\cdot \pr(\Pi \text{ doesn't terminate before phase }i) \,\,= \,\, \left(1+ \ln (1/\delta) \right)\cdot u_{i-1} .
\end{equation}

\subsection{Wrapping Up}

To complete the proof of \Cref{lem:key2}, we set
$\alpha := \frac{8\beta}{\delta}(1+\ln (1/\delta)) = O \big( \frac{\ln
  (1/\delta)}{\delta}\big)$, and combine \eqref{eq:G-LB} and
\eqref{eq:G-UB} to get
$$u_{i-1} \ge \frac{G}{1+\ln (1/\delta)} \ge \frac{\alpha \delta}{2 \beta (1+\ln (1/\delta))} \cdot \left(u_{i} - \frac65 u^*_i\right) = 4 \left(u_{i} - \frac65 u^*_i\right).$$
This completes the proof of \Cref{lem:key2} and hence Theorem~\ref{thm:main-ssc}. 
\fi

%%%%%%%%%%%%%% correlated case %%%%%%%%%%%%%%%%%%%%%

\newcommand{\scn}{\mathcal{\omega}}

\section{Scenario Submodular Cover}\label{sec:scn-submod-cover-r-adaptive}
In this section, we describe an $r$-round adaptive algorithm for the \emph{scenario submodular cover} problem. 
\ifICMLVersion
\else
As before, we have  a collection of $m$ stochastic items $\rv{X} = \{\rv{X}_1, ..., \rv{X}_m\}$ with costs $c_i$. 
\fi
In contrast to the independent case, the stochastic items here are correlated, and their joint distribution $\D$ is given as input. 
\ifICMLVersion
\else
The goal is to minimize the expected cost of a solution $\rv{S} \sse \rv{X}$ that realizes to a feasible set (i.e., $f(S)=Q$). 
\fi

The joint distribution $\D$ specifies the (joint) probability that $\rv{X}$ realizes to any outcome $X\in U^m$. We refer to the realizations $X\in U^m$ that have a non-zero probability of occurrence as \emph{scenarios}. Let $s = |\D|$ denote the number of scenarios in $\D$. 
The set of scenarios is denoted $M=\{1,\cdots, s\}$ and $p_\scn$ denotes the probability of each scenario $\scn\in M$. Note that $\sum_{\scn=1}^s p_\scn=1$. For each scenario $\scn\in M$ and item $\rv{X}_e$, we denote by $\rl{X_e}{\scn}\in U$ the realization of $\rv{X}_e$ in scenario $\scn$. The distribution $\D$ can be viewed as selecting a random {\em realized scenario} $\scn^*\in M$ according to the probabilities $\{p_\scn\}$, after which the item realizations are deterministically set to $\langle \rl{X_1}{\scn^*} , \cdots, \rl{X_m}{\scn^*}\rangle$.  However, an algorithm does not know the realized scenario $\scn^*$: it only knows the realizations of the probed items (using which it can infer a posterior distribution for $\scn^*$). 
As stated in \S\ref{sec:prelim}, our performance guarantee in this case depends on the support-size $s$. We will also show  that such a dependence is necessary (even when $Q$ is small). 

For any subset $\rv{S} \sse \rv{X}$ of items, we denote by $\rl{S}{\scn}$, the realizations for items in $\rv{S}$ under scenario $\scn$. We say that scenario $\scn$ is \emph{compatible} with a realization  of $\rv{S} \sse \rv{X}$ if and only if, $\rv{X}_e$ realizes to $\rl{{X}_e}{\scn}$ for all items $\rv{X}_e \in \rv{S}$.

\subsection{The Algorithm}\label{subsec:overview-scn-submod-cover}
\def\wOPT{\ensuremath{{\widehat{\OPT}}}\xspace}
Similar to the algorithm for the independent case, it is convenient to solve a \emph{partial cover} version of the
scenario submodular cover problem. However, the notion of partial progress is different: we will use the {\em number of compatible scenarios} instead of function value. Formally, in the partial version, we
are given a parameter $\delta \in [0,1]$ and the goal is to probe
some items $\rv{R}$ that realize to a set $R$ such that either (i) the number of compatible scenarios is less than $\delta s$ or (ii) the function $f$ is fully covered. 
 Clearly, if $\delta = 1/s$ then case (i) cannot happen (it corresponds to zero compatible scenarios), so the function $f$ must be fully covered. We will use this algorithm with different parameters $\delta$
to solve the $r$-round version of the problem. The main result of
this section is:

\begin{theorem}\label{thm:scn-partial-cover}
  There is a non-adaptive algorithm for the partial cover version of
scenario submodular cover with cost
  $O\left(\frac{1}{\delta} (\ln \frac1{\delta} + \log Q)\right)$ times the 
  optimal adaptive cost  for the (full) submodular cover.
\end{theorem}

The algorithm first creates an ordering/list $L$ of the
items non-adaptively; that is, without knowing the realizations of the
items. To do so, at each step we pick a new item that maximizes a
carefully-defined score function (Equation~\eqref{eq:scn-greedy-choice}).
The score of an item depends
on an estimate of progress towards (i) eliminating scenarios and (ii) covering function ${f}$. Before we state this score formally, we need some definitions. 

\def\H{{\cal H}}
\begin{definition} \label{defn:partn-scn}
For any $\rv{S}\sse \rv{X}$ let $\H(\rv{S})$ denote the partition $\{ Y_1, \cdots, Y_{\ell}\}$ of the scenarios $M$ where
all scenarios in a part have the same realization for items in $\rv{S}$. Let $\mathcal{Z} := \{ Y \in \H(\rv{S}) : |Y| \geq \delta s\}$ 
be the set of ``large'' parts having size at least $\delta s$.  
\end{definition}
In other words,  scenarios $\omega$ and $\sigma$ lie in the same part of $\H(\rv{S})$ if and only if $S(\omega)=S(\sigma)$. 
Note that partition $\H(\rv{S})$ does not depend on the realization of $\rv{S}$.  Moreover, after probing and realizing  items $\rv{S}$,  the set of compatible scenarios  must be one of the parts in $\H(\rv{S})$.  
Also, the number of ``large'' parts $|{\cal Z}|\le  \frac{s}{\delta s}= \frac{1}{\delta }$ as the number of scenarios $|M|=s$. See Figure~\ref{fig:parts} for an example.

\ifICMLVersion
\begin{figure}
     \centering
     \begin{subfigure}[b]{0.225\textwidth}
         \centering
         \includegraphics[width=\textwidth]{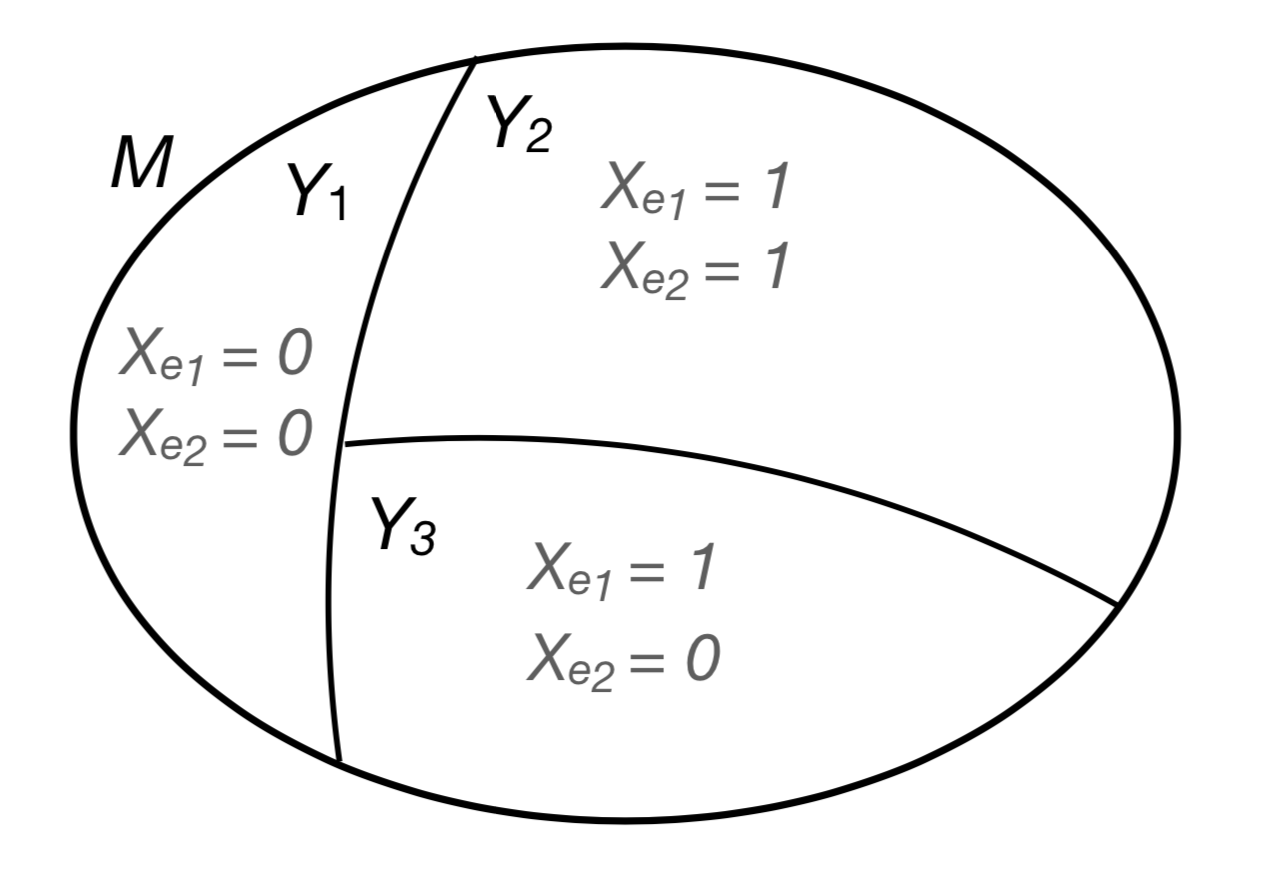}
         \caption{We have $\rv{S} = \{\rv{X}_{e_1}, \rv{X}_{e_2}\}$, and we partition the set of scenarios $M$ based on outcomes of $\rv{S}$ to get $\mathcal{H}(S) = \{Y_1, Y_2, Y_3\}$. }
         \label{fig:parts}
     \end{subfigure}
     \quad
     \begin{subfigure}[b]{0.225\textwidth}
         \centering
         \includegraphics[width=\textwidth]{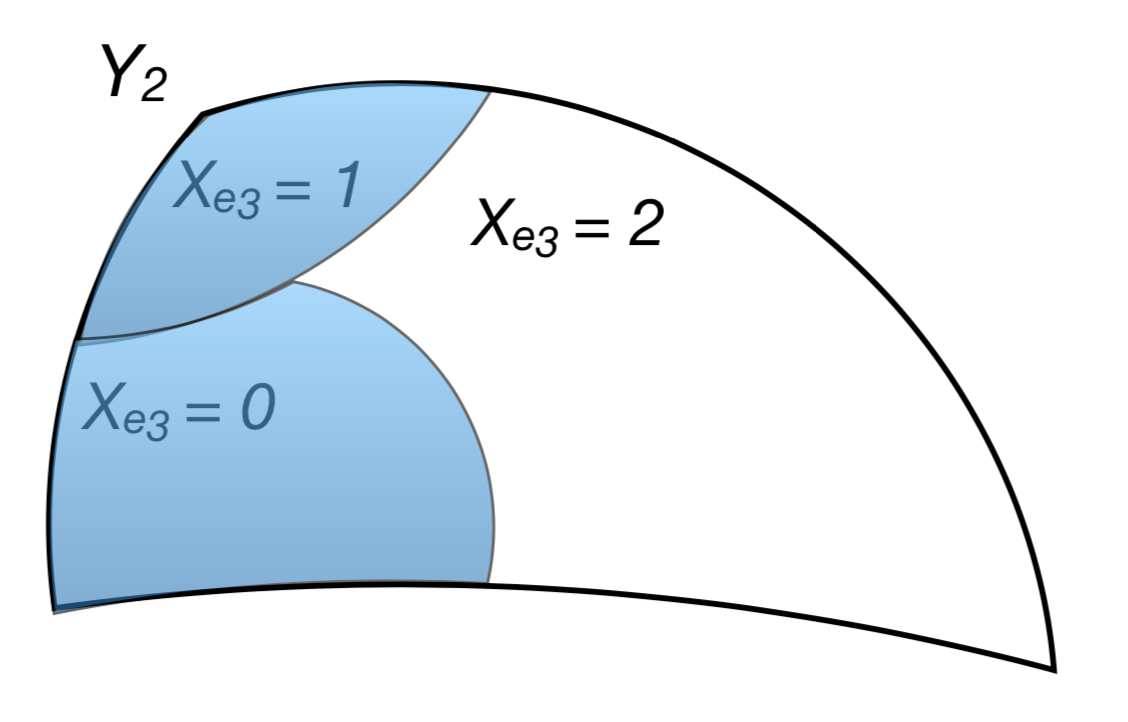}
         \caption{We further partition scenarios $Y_2$ based on realizations of $\rv{X}_{e_3}$.The part of $Y_2$ compatible with outcome $X_{e_3} = 2$ is the largest cardinality part, that is, $B_{e_3}(Y_2)$. The shaded region represents $L_{e_3}(Y_2)$.}
         \label{fig:L_e}
     \end{subfigure}
     \caption{Illustrations of Key Definitions}
        \label{fig:partition}
\end{figure}

\else
\begin{figure}
     \centering
     \begin{subfigure}[b]{0.45\textwidth}
         \centering
         \includegraphics[width=\textwidth]{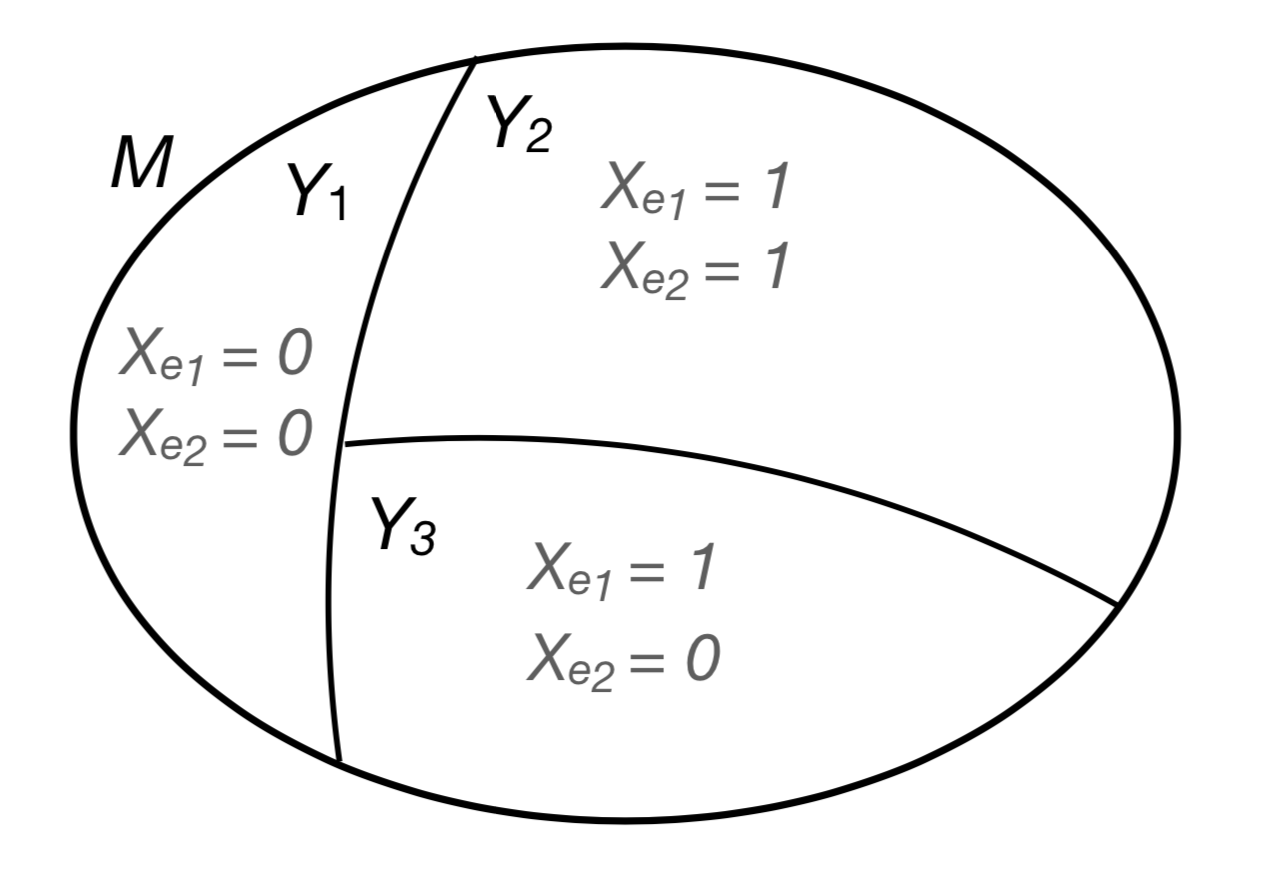}
         \caption{We have $\rv{S} = \{\rv{X}_{e_1}, \rv{X}_{e_2}\}$, and we partition the set of scenarios $M$ based on outcomes of $\rv{S}$ to get $\mathcal{H}(S) = \{Y_1, Y_2, Y_3\}$. }
         \label{fig:parts}
     \end{subfigure}
     \quad
     \begin{subfigure}[b]{0.45\textwidth}
         \centering
         \includegraphics[width=\textwidth]{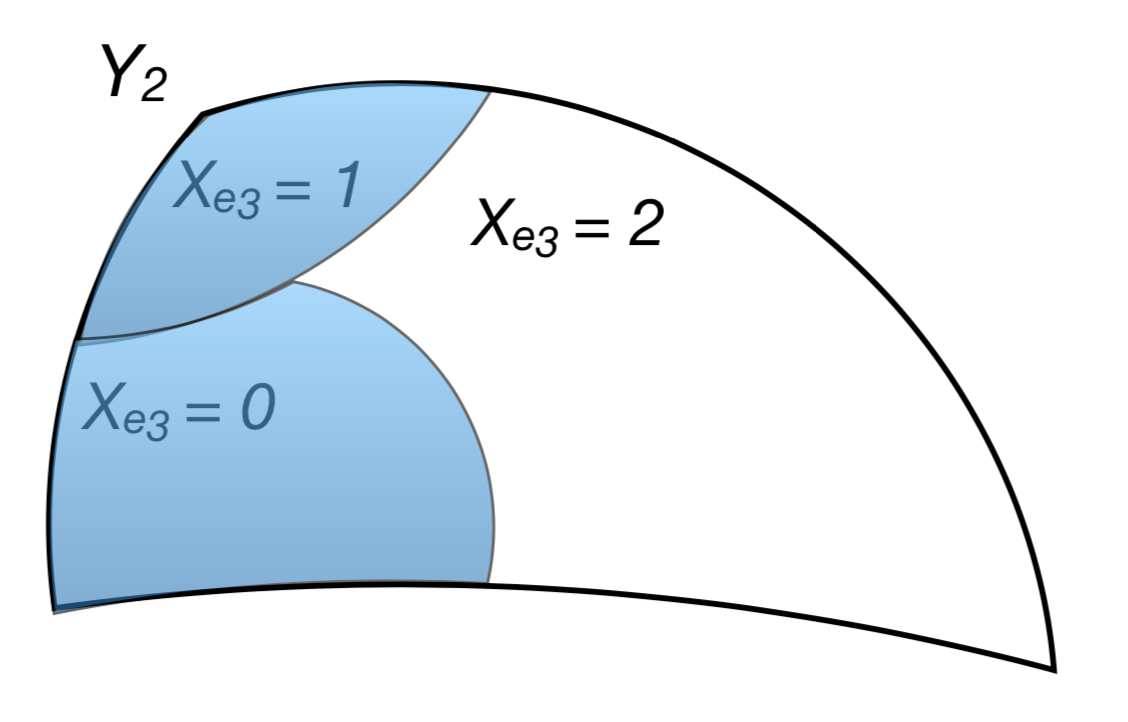}
         \caption{We further partition scenarios $Y_2$ based on realizations of $\rv{X}_{e_3}$.The part of $Y_2$ compatible with outcome $X_{e_3} = 2$ is the largest cardinality part, that is, $B_{e_3}(Y_2)$. The shaded region represents $L_{e_3}(Y_2)$.}
         \label{fig:L_e}
     \end{subfigure}
     \caption{Illustrations of Key Definitions}
        \label{fig:partition}
\end{figure}
\fi

\begin{definition}\label{defn:Le-scn}
For any $\rv{X}_e\in {\rv{X}}$ and subset $Z\sse M$ of scenarios, consider the partition of $Z$ based on the realization of $\rv{X}_e$. Let $B_e(Z)\sse Z$ be the largest cardinality part, and define $L_e(Z) := Z \setminus B_e(Z)$. 
\end{definition}
Note that $L_e(Z)$ is comprised of several parts of the above partition of $Z$. If the realized scenario  $\scn^*\in L_e(Z)$ and $\rv{X}_e$ is selected then, at least half the scenarios in $Z$ will be eliminated (as being incompatible with $X_e$).
Figure~\ref{fig:L_e} illustrates these definitions. 

For any $Z\in \H(\rv{S})$, note that the realizations $S(\scn)$ are identical for all $\scn\in Z$: we use $\rl{S}{Z}\sse U$ to denote the realization of $\rv{S}$ under each scenario in $Z$. 

If $\rv{S}$ denotes the previously added  items in list $L$, the score \eqref{eq:scn-greedy-choice} involves a term for each part $Z \in \mathcal{Z}$, which itself comes from two sources:
\ifICMLVersion
\vspace{-0.1in}
\else\fi
\begin{itemize}
    \item \emph{Information gain} $\sum_{\scn\in L_e(Z)} p_\scn$, the total probability of scenarios in $L_e(Z)$.
    
\ifICMLVersion
\item \emph{Relative function gain} 
$$\sum_{\scn \in Z} p_\scn \cdot \frac{{f}(\rl{S}{Z} \cup \rl{X_e}{\scn}) -  {f}(\rl{S}{Z})}{{Q} - {f}(\rl{S}{Z})},$$ the expected relative gain obtained by including element $X_e$, where the expectation is over   scenarios   $Z$.
\else
\item \emph{Relative function gain} $\sum_{\scn \in Z} p_\scn \cdot \frac{{f}(\rl{S}{Z} \cup \rl{X_e}{\scn}) -  {f}(\rl{S}{Z})}{{Q} - {f}(\rl{S}{Z})} $, the expected relative gain obtained by including element $X_e$, where the expectation is over the scenarios in part $Z$.
\fi
\end{itemize}
The overall score of item $\rv{X}_e$ is the sum of these terms (over all parts in $\mathcal{Z}$) normalized by the cost $c_e$ of item $\rv{X}_e$. In defining the score, we only focus on the ``large'' parts ${\cal Z}$. If the realization of $\rv{S}$ corresponds to any other part then the number of compatible scenarios would be less than $\delta s$ (and the partial-cover algorithm would have  terminated).

Once the list $L$ is
specified, the algorithm starts probing and realizing the items in this order, and
does so until either (i) the number of compatible scenarios drops below $\delta s$, or (ii) the realized function value equals $Q$. Note that in case (ii), the function is fully covered. See Algorithm~\ref{alg:scn-parca} for a formal description of the non-adaptive partial-cover algorithm.

\begin{algorithm}[h]
\caption{Scenario PARtial Covering Algorithm \textsc{SParCA}$(\rv{X}, M, f, Q, \delta)$} \label{alg:scn-parca}
\begin{algorithmic}[1]
  \State $\rv{S} \leftarrow \emptyset$ and list $L \leftarrow \langle
  \rangle$.
\While{$\rv{S}\ne \rv{X}$} \Comment{Building the list non-adaptively}
\State define ${\cal Z}$ and $L_e(Z)$ as in Definitions~\ref{defn:partn-scn} and \ref{defn:Le-scn}.
\State select an item $\rv{X}_e \in \rv{X} \setminus \rv{S}$ that maximizes:

\ifICMLVersion
\begin{align}
\text{score}&(\rv{X}_e) =  \frac1{c_e} \cdot \sum_{Z \in \mathcal{Z}} \Bigg( \sum_{\scn \in L_e(Z)} {p}_{\scn} + \notag \\ 
        &\sum_{\scn \in Z} {p}_\scn \cdot \frac{{f}(\rl{S}{Z} \cup \rl{X_e}{\scn}) -  {f}(\rl{S}{Z})}{{Q} - {f}(\rl{S}{Z})}\Bigg) \label{eq:scn-greedy-choice}
\end{align}

\else
\begin{equation}\label{eq:scn-greedy-choice}
\text{score}(\rv{X}_e) =  \frac1{c_e} \cdot \sum_{Z \in \mathcal{Z}} \Bigg( \sum_{\scn \in L_e(Z)} {p}_{\scn} + \sum_{\scn \in Z} {p}_\scn \cdot \frac{{f}(\rl{S}{Z} \cup \rl{X_e}{\scn}) -  {f}(\rl{S}{Z})}{{Q} - {f}(\rl{S}{Z})}\Bigg)  
\end{equation}
\fi

\State $\rv{S} \leftarrow \rv{S} \cup \{\rv{X}_e\}$ and list $L\gets L \circ \rv{X}_e$
\EndWhile
\State  $\rv{R} \gets \emptyset$, $R \gets \emptyset$, $H\gets M$.
\While{$|H|\ge \delta |M|$ and $f(R)<Q$} \label{step:sparca-probe} \ifICMLVersion \else\Comment{Probing items on the list}\fi
\State $\rv{X}_e \gets$ first r.v.\ in list $L$ not in $\rv{R}$
\State $X_e = v\in U$ be the realization of $\rv{X}_e$.
\State $R \gets R \cup \{ v\}, \rv{R} \gets \rv{R} \cup \{ \rv{X}_e
\}$
\State $H\gets \{\scn\in H : X_e(\scn) = v\}$
\EndWhile
\State return probed items $\rv{R}$, realizations $R$ and compatible scenarios $H$.
\end{algorithmic}
\end{algorithm}

Given this partial covering algorithm we immediately get an algorithm
for the $r$-round version of the problem, where we are allowed to make
$r$ rounds of adaptive decisions. Indeed, we can first set
$\delta = s^{-1/r}$ and solve the partial covering problem. Suppose we probe the items $\rv{R}$ (with 
realizations $R \sse U$) and are left with compatible scenarios $H\sse M$. Then we can {\em condition} on
scenarios $H$ and the marginal value function $f_R$ (which is submodular), and
inductively get an $r-1$-round solution for this problem. The following algorithm and result formalizes this.
\begin{algorithm}[h]
\caption{$r$-round adaptive algorithm for scenario submodular cover \textsc{NSC}$(r, \rv{X}, M, f)$} \label{alg:r-round-scn}
\begin{algorithmic}[1]
  \State\label{step:scn-round-1} run  \textsc{SParCA}$(\rv{X}, M, f, Q, |M|^{-1/r})$ for round one.
  Let $\rv{R}$ denote the probed items, $R$ their realizations, and $H\sse M$ the compatible scenarios returned by \textsc{SParCA}.  
  \State define residual submodular function $\fh := f_R$.
  \State \label{step:scn-round-rec}recursively solve \textsc{NSC}$(r-1, \rv{X}\setminus \rv{R}, H, \fh)$.
\end{algorithmic}
\end{algorithm}

\begin{theorem}\label{thm:r-round-scn}
Algorithm~\ref{alg:r-round-scn} is an $r$-round adaptive algorithm for scenario submodular
  cover with cost $O\left(s^{1/r} (\log s + r\log Q)\right)$ times the optimal
  adaptive cost, where $s$ is the number of scenarios.
\end{theorem}
\ifICMLVersion
\else
\begin{proof} We proceed by induction on the number of rounds $r$. Let $\OPT$ denote the cost of an optimal fully adaptive solution. 
The base case is $r=1$, in which case $\delta=s^{-1/r} = \frac1s$. By Theorem~\ref{thm:scn-partial-cover}, the partial cover algorithm \textsc{SParCA}$(\rv{X}, M, f, Q, s^{-1/r})$ obtains a realization $R$ and compatible scenarios $H\sse M$ where either (i) $|H|<\delta s = 1$  or (ii) $f(R)=Q$. Clearly, we cannot have case (i) as there is always at least one compatible scenario. So $f$ is fully covered.
Moreover, the algorithm's expected cost is $O(s(\log s + \log Q))\cdot \OPT$, as claimed.

\def\sh{\widehat{s}}
We now consider $r>1$ and assume (inductively) that Algorithm~\ref{alg:r-round-scn} finds an $r-1$-round $O\left(s^{1/r} (\log s + r\log Q)\right)$-approximation algorithm for any instance of scenario submodular cover. Let $\delta = s^{-1/r}$. By Theorem~\ref{thm:scn-partial-cover}, the expected cost in round 1 (step~\ref{step:scn-round-1} in Algorithm~\ref{alg:r-round-scn}) is $O(s^{1/r}(\frac1r \log s + \log Q))\cdot \OPT$.  Let $\sh=|H|$ denote the number of scenarios in the residual instance. Let
 $\Qh:= Q-f(R) = \fh(U)$ denote the maximal value of the residual submodular function $\fh=f_R$. By definition of the partial covering problem,  we have either $\sh <\delta s$ or $\Qh= 0$. If $\Qh=0$ then our algorithm incurs no further cost and the inductive statement follows. We now assume $\sh<\delta s = s^{\frac{r-1}{r}}$.  The optimal
  solution $\OPT$ conditioned on the scenarios $H$ gives a feasible adaptive solution to the residual instance  of covering $\fh$; we denote this conditional solution by
  $\wOPT$. We inductively get that the cost of our
  $r-1$-round solution on $\fh$ is at most 
  \ifICMLVersion
  \begin{align*}
      O&(\sh^{1/(r-1)}  (\log \sh + (r-1)\log \Qh))\cdot  \widehat{\OPT} \\ &\le O\left(s^{1/r} \left(\frac{r-1}r \log s + (r-1) \log Q\right)\right)\cdot \widehat{\OPT},
  \end{align*}
  \else
  $$O(\sh^{\frac1{r-1}} (\log \sh + (r-1)\log \Qh))\cdot  \widehat{\OPT} \le O\left(s^{1/r} \left(\frac{r-1}r \log s + (r-1) \log Q\right)\right)\cdot \widehat{\OPT},$$
  \fi
  where we used  $\sh < s^{(r-1)/r}$ and $\Qh \le Q$.  As this holds for every subset of scenarios $H$,  we
  can take expectations over $H$ to get that the (unconditional)
  expected cost of the last
  $r-1$ rounds is  
$O\left(s^{1/r} \left(\frac{r-1}r \log s + (r-1) \log Q\right)\right)\cdot \OPT$. 
  Adding to this the cost of the first round, which is $O(s^{1/r}(\frac1r \log s + \log Q))\cdot \OPT$, completes the proof.
\end{proof}
\fi

%%%%%%%%%%%%%%%%%%%%%%%%%%%%%%%%%%%%%%%%%%%%%%%%%%%%%%%%%%%%%%%%%%%%
\ifICMLVersion
We defer the proofs of Theorems~\ref{thm:scn-partial-cover} and \ref{thm:r-round-scn} to the full version of the paper.
\else
\subsection{Analysis for the partial covering algorithm}
We now prove Theorem~\ref{thm:scn-partial-cover}. 
Consider any call to \textsc{SParCA}. Let $s=|M|$ denote the number of scenarios in the instance. Recall that the goal is to probe items $\rv{R} \sse \rv{X}$ with some realization $R$ and compatible scenarios $H\sse M$ such that (i) $|H|<\delta s$ or (ii) $f(R) = Q$. We denote by $\OPT$ an optimal fully adaptive solution for the covering problem on $f$.  Now we analyze the cost incurred by our algorithm's
non-adaptive strategy (which we call $\NA$). Note that $\NA$ probes items in  the order given by the list $L$ (generated by
\textsc{SParCA}) and stops when either condition (i) or (ii) above occurs. 
We consider the expected cost of this strategy, and relate it to the cost
of ${\OPT}$.
  The high-level approach is similar to that for the independent case: but the details are quite different.

We refer to the cumulative cost incurred until any point in a solution as \emph{time}. We say that $\OPT$ is in phase $i$ in the time interval $[2^i, 2^{i+1})$ for $i \geq 0$. We say that $\NA $ is in phase $i$ in the time interval $[\beta\cdot 2^{i-1}, \beta\cdot 2^{i})$ for $i \geq 1$. We use phase $0$ to refer to the interval $[1, \beta)$. We set $\beta := \frac{16}{\delta}\log(Q/\delta)$; this choice will become clear in the proof of Lemma~\ref{lem:scn-submod-key}. The following notation is associated with any phase $i \geq 0$:
\begin{itemize}
    \item $u_{i}$: probability that $\NA $ goes beyond phase $i$, i.e., costs at least $\beta\cdot 2^i$.
    \item $u_{i}^*$: probability that $\OPT$ goes beyond phase $i-1$, i.e., costs at least $2^{i}$.
\end{itemize}
Since all costs are integers, $u_0^* = 1$. For ease of notation, we also use $\OPT$ and $\NA$ to denote the \emph{random} cost incurred by $\OPT$ and $\NA$ respectively. 
The main result here is that $\E[\NA ]\le  O(\beta)\cdot \E[\OPT]$. As in the independent case, it suffices to show:
\begin{lemma}\label{lem:scn-submod-key}
For any phase $i \geq 1$, we have $ u_{i} \leq \frac{u_{i-1}}{4} + 2u_{i}^*. $
\end{lemma}
Using this lemma, we can immediately prove Theorem~\ref{thm:scn-partial-cover}. This part of the analysis is identical to the one presented in \S\ref{sec:submod-cover-r-adaptive} for Theorem~\ref{thm:partial-cover}.

\section{Proof of the key lemma for Scenario Submodular Cover }
We now prove \Cref{lem:scn-submod-key}. 
Let $L$ be the list returned by \textsc{SParCA}. Recall that $\NA $ is the cost incurred by non-adaptively realizing items in the order of $L$ until (i) the number of compatible scenarios $|H| < \delta s$ or (ii) $f$ gets fully covered. For each time $t\ge 0$, let $\rv{X}_{e(t)}$ denote the item that  would be selected at time $t$. In other words, this is the item  which causes the cumulative cost of $L$ to exceed $t$ for the first time. We define the {\em total gain} as the random variable 
$$ G \,\,:=\,\, \sum_{t = \beta 2^{i-1}}^{ \beta 2^{i}} \text{score}(\rv{X}_{e(t)}), $$ 
which corresponds to the sum of scores over the time interval $[\beta \cdot 2^{i-1}, \beta \cdot 2^{i})$.  The proof will be completed by upper and lower bounding $G$, which we do next. The lower bound views $G$ as a sum over time steps, whereas the upper bound  views $G$ as a sum over decision paths.

\subsection{Lower bounding $G$}
For the analysis, it will be convenient to view $\OPT$ as a decision tree where nodes  correspond to probed items and branches correspond to item realizations. 
Fix some time $t \in [\beta\cdot 2^{i-1}, \beta\cdot 2^i)$ in phase $i$. Let $\rv{S}$ denote the  items that have already been added to the list $L$: so ${\rv{X}} \setminus \rv{S}$ are the available items. The partition $\H(\rv{S})$ of scenarios $M$ and the large parts ${\cal Z}$ are as in Definition~\ref{defn:partn-scn}.  For any $Z\in {\cal Z}$, we define:
\begin{itemize}
\item $\rl{S}{Z}\sse U$ is the realization from items $\rv{S}$ compatible with all scenarios in $Z$. 
\item ${Q}_{Z} := {Q} - {f}(\rl{S}{Z})$ is the residual target (after selecting $\rv{S}$) under scenarios $Z$.
\item residual submodular function ${f}_{Z} := {f}_{\rl{S}{Z}}$ under scenarios $Z$.
\item for any item $\rv{X}_e\in \rv{X}\setminus \rv{S}$, scenarios $L_e(Z)\sse Z$ are as in Definition~\ref{defn:Le-scn}. This is the complement of the largest part of $Z$ (based on the realization of $\rv{X}_e$).
\item $\OPT_{Z}$ is the subtree of $\OPT$ until time $2^i$, restricted to paths traced by scenarios in $Z$; this only includes items that are \emph{completely} selected by time $2^i$. 
\item $\Stem_Z$ is the path in $\OPT_{Z}$ that at each node $\rv{X}_e$ follows the branch corresponding to the realization compatible with scenarios $B_e(Z)=Z\setminus L_e(Z)$. We use $\Stem_Z$  to also denote the set of items on this path.
\end{itemize}
See Figure~\ref{fig:opt-stem} for an illustration of $\OPT_Z$ and $\Stem_Z$.

\begin{figure}
    \centering
    \includegraphics[width=6in]{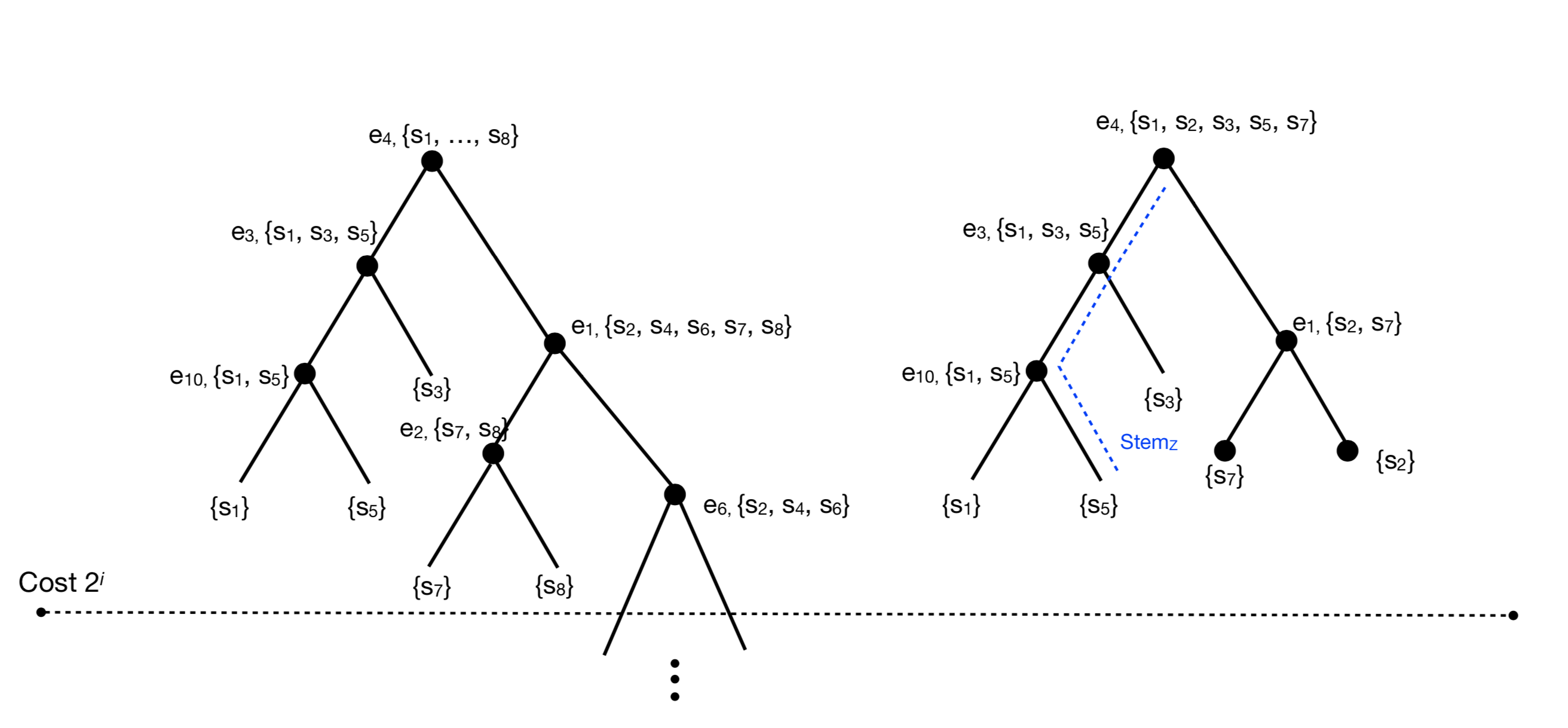}
    \caption{The   tree on the left represents the decision tree  $\OPT$ for scenarios $\{s_1, ..., s_8\}$. Each node represents the element chosen at the node, and the set of scenarios compatible until the node (before element at node is chosen). We redraw the decision tree for $Z = \{s_1, s_2, s_3, s_5, s_7\}$ and restrict attention to scenarios in $Z$. The decision tree on the right represents $\OPT_Z$. Note that both trees are restricted to time $2^i$. Furthermore, we add $\Stem_Z$ (using dotted lines) in $\OPT_Z$. Note that scenario $s_5$ is a \emph{good} scenario since it is covered by the end of $\Stem_Z$, and scenarios $s_1, s_2, s_3, s_7$ are all \emph{okay} scenarios. Thus $Z_{\text{good}} = \{s_5\}$ and $Z_{\text{okay}} = \{s_1, s_2, s_3, s_7\}$. If $s_5$ was \emph{not} covered by time $2^i$, it would be a \emph{bad} scenario.}
    \label{fig:opt-stem}
\end{figure}

For each $Z\in {\cal Z}$, we partition scenarios $Z$ based on the realizations on $\Stem_Z$:
\begin{itemize}
    \item $Z_{\text{good}} = \{ \scn \in Z : \scn \text{ compatible at the end of } \Stem_Z \text{ and } {f} \text{ covered}\}$,
    \item $Z_{\text{bad}} = \{ \scn \in Z : \scn \text{ compatible at the end of } \Stem_Z \text{ and } {f} \text{ uncovered}\}$,
    \item $Z_{\text{ok}} = \{ \scn \in Z : \scn \in L_e(Z) \text{ for some } \rv{X}_e \in \Stem_Z \}$; these are the scenarios that diverge from $\Stem_Z$ at some item $\rv{X}_e$. 
\end{itemize}
See Figure~\ref{fig:opt-stem} for an illustration.
 
 \def\wp{p}
\begin{definition}
We define part $Z\in {\cal Z}$ as \textbf{good}  if $\wp(Z_{\text{good}}) \geq \wp(Z)/2$, \textbf{bad}  if $\wp(Z_{\text{bad}}) \geq \wp(Z)/2$, or \textbf{okay}  if $\wp(Z_{\text{okay}}) \geq \wp(Z)/2$. \end{definition}
If there are ties, they can be broken arbitrarily. 
\begin{lemma} Each $Z\in {\cal Z}$ is either good, bad or okay.
\end{lemma}
\begin{proof} Observe that one of $Z_{\text{good}}$ or $Z_{\text{bad}}$ is always empty and the other equals all scenarios compatible at the end of $\Stem_Z$. Moreover,  $Z_{\text{okay}}$ consists of all scenarios that diverge from $\Stem_Z$. So $Z_{\text{good}} \cup Z_{\text{bad}} \cup Z_{\text{okay}} = Z$. 
It follows that $\max\{ \wp(Z_{\text{good}}) , \wp(Z_{\text{bad}}), \wp(Z_{\text{okay}}) \} \ge \wp(Z)/2$, which proves the lemma.
\end{proof}

\def\nz{\rv{T}_Z}
For each $Z\in {\cal Z}$, let $\nz =\Stem_Z\setminus \rv{S}\sse {\rv{X}}\setminus \rv{S}$ denote the {\em available} items on $\Stem_Z$. 
In the next two lemmas, we lower bound the total score from items in $\nz$. We consider separately the terms corresponding   to function gain and information gain. 
For any scenario $\scn\in Z$, we use $\rl{Stem_Z}{\scn}\sse U$ denotes the realization of  items $\Stem_Z$ under scenario $\scn$.

\begin{lemma}\label{lem:scn-good}
If $Z$ is good, then $\sum_{\scn \in Z} {p}_\scn \cdot \frac{{f}_{Z}(\rl{T_Z}{\scn})}{{Q}_{Z}} \geq \frac{\wp(Z)}{2}$.
\end{lemma}
\begin{proof}
Consider any scenario $\scn \in Z_{\text{good}}$. If $\scn$ is realized then ${f}$ is covered by the realization $\rl{Stem_Z}{\scn}$ of $\Stem_Z$. Thus, ${f}(\rl{Stem_Z}{\scn}) = {Q}$. We have
$${f}_{Z}(T_Z(\scn))={f}(T_Z(\scn) \cup S(Z)) - {f}( S(Z)) ={f}(Stem_Z(\scn) \cup S(Z)) - {f}( S(Z)) = {Q}- {f}( S(Z)) ={Q}_Z,$$
where the second equality uses that $T_Z(\scn) = Stem_Z(\scn) \setminus S(Z)$. 
As $Z$ is good, $\wp(Z_{\text{good}}) \geq \wp(Z)/2$, and the lemma follows by summing over $\scn\in Z_{\text{good}}$. 
\end{proof}

\begin{lemma}\label{lem:scn-okay}
If $Z$ is okay, then $\wp\left(\bigcup_{\rv{X}_e \in \nz} L_e(Z)\right) \geq \frac{\wp(Z)}{2}$.  
\end{lemma}
\begin{proof}
Note that $Z_{\text{okay}} = \bigcup_{\rv{X}_e \in \Stem_Z} L_e(Z)$. As $Z$ is okay, we have $\wp(Z_{\text{okay}}) \geq \wp(Z)/2$. Recall that $Z$ is a part in $\H(\rv{S})$, which implies that realizations $S(\scn)$ of $\rv{S}$ are identical under all scenarios $\scn\in Z$. Hence, $B_e(Z)=Z$ for each $\rv{X}_e\in \rv{S}$, i.e.,  $\cup_{\rv{X}_e \in \rv{S}} L_e(Z)=\emptyset$. So,
$$\bigcup_{\rv{X}_e \in \nz} L_e(Z) \quad \supseteq \quad  \left(\bigcup_{\rv{X}_e \in \Stem_Z} L_e(Z) \right) \setminus \left(\bigcup_{\rv{X}_e \in \rv{S}} L_e(Z) \right) \quad = \quad \bigcup_{\rv{X}_e \in \Stem_Z} L_e(Z) \quad = \quad Z_{\text{okay}}. $$
Hence, $\wp\left(\bigcup_{\rv{X}_e \in \nz} L_e(Z)\right) \geq \frac{\wp(Z)}{2}$, which completes the proof.
\end{proof}

We now prove a bound on the \emph{overall} probability of parts that are either good or okay.

\begin{lemma}\label{lem:scn-bad}
We have $\sum_{Z: \text{okay or good}} \wp(Z) \geq (u_{i} - 2u_i^*)$.
\end{lemma}
\begin{proof}
We first show that $\sum_{Z: \text{ bad }} \wp(Z)\le 2 \cdot u^*_i$. For any $Z\in {\cal Z}$, 
if scenario $\scn \in Z_{\text{bad}}$ is realized then we know that $\OPT$ does not cover ${f}$ at the end of $\Stem_Z$, which implies that $\OPT$ costs at least $2^i$. Thus $u_i^* \geq \sum_{Z \in \mathcal{Z}}\sum_{\scn \in Z_{\text{bad}}} {p}_\scn$. Moreover, $\wp(Z_{\text{bad}})\ge\wp(Z)/2$ if $Z$ is bad. So, 
\begin{equation}\label{eq:scn-bad-prob}
 u_i^* \geq \sum_{Z \in \mathcal{Z}}\sum_{\scn \in Z_{ \text{bad}}} {p}_\scn \geq \sum_{Z: \text{bad}} \,\,\sum_{\scn \in Z_{\text{bad}}} {p}_\scn \geq \sum_{Z: \text{ bad}} \wp(Z)/2.    
\end{equation}
Thus, $\sum_{Z: \text{bad}} \wp(Z) \leq 2u_i^*$ as claimed. 

Consider any  part $Y\in \H(\rv{S})\setminus {\cal Z}$; recall Definition~\ref{defn:partn-scn}. Note that $|Y|<\delta s$. So, if $\rv{S}$ realizes to the set $\rl{S}{Y}$  then $\NA $ would   terminate  by time $t\le \beta 2^i$, which implies that $\NA $ does not go beyond phase $i$. Thus, 
\begin{equation}\label{eq:scn-low-prob}
\sum_{Y\in \H(\rv{S})\setminus {\cal Z} } \wp(Y) \leq \Pr[\NA \mbox{ terminates by phase }i] = 1 - u_{i}.    
\end{equation}
 Finally, we have 
$$\sum_{Z: \text{okay or good}} \wp(Z)  = 1 -  \sum_{Z: \text{ bad }} \wp(Z) - \sum_{Y\in \H(\rv{S})\setminus {\cal Z} } \wp(Y) \geq 1 - 2u_i^* - (1 - u_{i})= u_{i} - 2u_i^*,$$
where we used \eqref{eq:scn-bad-prob} and \eqref{eq:scn-low-prob}. This  completes the proof.
\end{proof}
 
 We are now ready to prove the lower bound on $\text{score}(\rv{X}_{e(t)})$. 
\begin{lemma}\label{lem:scn-item-lb}
For any time $t$ in phase $i$, we have \emph{score}$(\rv{X}_{e(t)}) \geq \delta\cdot\frac{1}{2^{i+1}} \cdot (u_i - 2u_i^*)$.
\end{lemma}
\begin{proof} Recall the subsets $\nz \sse {\rv{X}}\setminus \rv{S}$ for $Z\in {\cal Z}$. Let $\rv{T} = \bigsqcup_{Z\in {\cal Z}} \nz$ denote the {\em multiset} that contains each item as many times as it occurs. By definition of $\Stem_Z$, the cost $c(\nz)\le 2^i$ for each $Z$. So, $c(\rv{T}) = \sum_{Z\in {\cal Z}} c(\nz) \le |{\cal Z}| \cdot 2^i \le \frac1\delta\cdot 2^i$. Moreover, $\rv{T}\sse \rv{X}\setminus \rv{S}$: so  each item in $\rv{T}$ is available to be added to  $L$ at time $t$.   We can lower bound $\text{score}(\rv{X}_{e(t)})$ by averaging over the multiset $\rv{T}$: 
\begin{align}
\text{score}(\rv{X}_{e(t)}) &\ge  \max_{\rv{X}_e \in \rv{T}} \text{score}(\rv{X}_{e}) \ge \frac{1}{c(\rv{T})}\cdot  \sum_{Z \in \mathcal{Z}}\sum_{\rv{X}_e \in \nz} \sum_{Y \in \mathcal{Z}} \Bigg( \wp(L_e(Y)) + \sum_{\scn \in Y} {p}_\scn \cdot \frac{{f}_{Y}(\rl{X_e}{\scn})}{{Q}_{Y}}\Bigg)\notag \\
&\ge \frac{1}{c(\rv{T})}\cdot  \sum_{Z \in \mathcal{Z}}\sum_{\rv{X}_e \in \nz} \Bigg( \wp(L_e(Z)) + \sum_{\scn \in Z} {p}_\scn \cdot \frac{{f}_{Z}(\rl{X_e}{\scn})}{{Q}_{Z}}\Bigg).\label{eq:scn-LB-score} 
\end{align} 
We now consider the two terms above (information and function gain) separately.  
 
 \paragraph{Bounding the information gain.} We only consider \emph{okay} sets $Z$.
$$ \sum_{Z \in \mathcal{Z}} \sum_{\rv{X}_e \in \nz} \wp(L_e(Z)) \ge \sum_{Z \in \mathcal{Z}} \wp\left(\bigcup_{\rv{X}_e \in \nz} L_e(Z)\right) \ge \sum_{Z : \text{ okay}}\frac{\wp(Z)}{2},$$
where the first inequality is by submodularity of the $\wp$-weighted coverage function and the second inequality is by Lemma~\ref{lem:scn-okay}.

 \paragraph{Bounding the function gain.} Consider any \emph{good} set $Z$.
$$ \sum_{\rv{X}_e \in \nz} \, \sum_{\scn \in Z} {p}_\scn \cdot \frac{{f}_{Z}(\rl{X_e}{\scn})}{{Q}_{Z}} = \sum_{\scn \in Z} {p}_\scn \cdot \frac{\sum_{\rv{X}_e \in \nz}{f}_{Z}(\rl{X_e}{\scn})}{{Q}_{Z}}   \ge \sum_{\scn \in Z} {p}_\scn \cdot \frac{{f}_{Z}(\rl{T_Z}{\scn})}{{Q}_{Z}} \ge \frac{\wp(Z)}{2},$$
where the first inequality is by submodularity of ${f}_Z$ and the second inequality is by Lemma~\ref{lem:scn-good}. Adding over all good sets $Z$,
$$\sum_{Z\in {\cal Z}} \sum_{\rv{X}_e \in \nz} \, \sum_{\scn \in Z} {p}_\scn \cdot \frac{{f}_{Z}(\rl{X_e}{\scn})}{{Q}_{Z}} \ge \sum_{Z : \text{ good}}\frac{\wp(Z)}{2}.$$

Combining the two bounds above,
$$  \sum_{Z \in \mathcal{Z}}\sum_{\rv{X}_e \in \nz} \Bigg( \wp(L_e(Z)) + \sum_{\scn \in Z} {p}_\scn \cdot \frac{{f}_{Z}(\rl{X_e}{\scn})}{{Q}_{Z}}\Bigg) \ge \sum_{Z : \text{okay or good}}\frac{\wp(Z)}{2}\ge \frac{u_i-2u_i^*}{2},$$
where the last inequality is by Lemma~\ref{lem:scn-bad}. Combined with \eqref{eq:scn-LB-score} and $c(\rv{T})\le  \frac1\delta\cdot 2^i$, this completes the proof.
\end{proof}

As Lemma~\ref{lem:scn-item-lb} holds for each time $t$ in phase $i$,
\begin{equation}
\label{eq:scn-LB}
G = \sum_{t=\beta 2^{i-1}}^{\beta 2^i}  \text{score}(\rv{X}_{e(t)}) \ge  \beta  2^{i-1} \cdot \delta\cdot  \frac{1}{2^{i+1}} \cdot (u_{i} - 2u_i^*) \geq \frac{\beta\delta }{4} \cdot (u_{i} - 2u_i^*)\end{equation}

\subsection{Upper bounding $G$}
We now consider the implementation of the non-adaptive list $L$ and calculate $G$ as a sum of contributions over the decision path in the non-adaptive solution $\NA$. Let ${\scn} \in M$ denote the realized scenario, and consider the decision path under scenario $\scn$. Let $\langle X_1, X_2,\ldots \rangle$ be the sequence of realizations (each in $U$) of items in $L$ under scenario ${\scn}$. So item $\rv{X}_j$ is selected between time $\sum_{\ell=1}^{j-1} c_\ell$ and $\sum_{\ell=1}^{j} c_\ell$. 
Let $h$ (resp. $p$) index the first (resp. last) item  (if any) that is selected (even partially) during phase $i$, i.e., between time $\beta 2^{i-1}$ and $\beta 2^{i}$. For each index $h\le j\le p$, let $t_j$ denote the duration of time that item $\rv{X}_j$ is selected during in phase $i$; so $t_j$ is the width of interval $[\sum_{\ell=1}^{j-1} c_\ell,\, \sum_{\ell=1}^{j} c_\ell] \bigcap  [\beta \cdot 2^{i-1},\, \beta \cdot 2^{i}]$. 
Furthermore, for each index $ j$, let $Z_j\sse M$ be the set  of scenarios compatible with outcomes $\langle X_1,\cdots, X_{j-1}\rangle$. Note that  $M = Z_1 \supseteq Z_2 \supseteq \cdots \supseteq Z_p\supseteq Z_{p+1}$. 
Define $G({\scn}):=0$ if index $h$ is undefined (i.e., $\NA $ ends before phase $i$ under scenario $\scn$) and otherwise:
\begin{align}
   G({\scn}) := &\sum_{j=h}^{p} \frac{t_j}{c_j} \cdot \Bigg( \mathbbm{1}[\scn  \in L_j(Z_j)] +  \frac{{f}(\{X_1, ..., X_j\}) -  {f}(\{X_1, ..., X_{j-1}\})}{{Q} - {f}(\{X_1, ..., X_{j-1}\})}\Bigg)  \notag \\
   &\leq \sum_{j=h}^{p}  \Bigg( \mathbbm{1}[\scn  \in L_j(Z_j)] +  \frac{{f}(\{X_1, ..., X_j\}) -  {f}(\{X_1, ..., X_{j-1}\})}{{Q} - {f}(\{X_1, ..., X_{j-1}\})}\Bigg)\label{eq:scn-path-UB}
\end{align}
By the stopping criterion for $L$, {\em before} the end of  $\NA $ we must have (i) the number of compatible scenarios remains at least $\delta s$ and (ii) the ${f}$ value remains less than  ${Q}$. In particular, we must have $|Z_{p}|\ge \delta s$.

We analyze the two quantities in the above expression separately. We begin by analyzing $ \sum_{j=h}^{p}  \mathbbm{1}[\scn  \in L_j(Z_j)]$.  Fix any $h\le j \le p$. If the realized scenario $\scn  \in L_j(Z_j)$, then the number of compatible scenarios drops by a factor of at least $2$ upon observing $X_j$; that is, $|Z_{j+1}| \leq \frac{1}{2} |Z_{j}|$. Using the fact that $|Z_1|=|M|={s}$ and $|Z_{p}|\ge \delta s$, it follows that   
$$ \sum_{j=h}^{p}  \mathbbm{1}[\scn  \in L_j(Z_j)] \leq \log_2 \frac{s}{\delta s} = \log_2(1/\delta).$$ 

Next, we analyze the second term in \eqref{eq:scn-path-UB}. Using the fact that ${f}$ is integral, monotone and takes values in the range $[0,{Q}]$, we have:
$$ \sum_{j=h}^{p} \frac{{f}(\{X_1, ..., X_j\}) -  {f}(\{X_1, ..., X_{j-1}\})}{{Q} - {f}(\{X_1, ..., X_{j-1}\})} \ \leq \ \sum_{\ell=1}^{{Q}} \frac{1}{\ell} \ \leq \ \ln Q.$$ 
Thus, we have $ G(\scn ) \ \leq \  \log \frac1\delta + \ln Q \ \leq  \ \log (Q/\delta)$.

Note that $G=\E_{\scn  \sim M}[G(\scn )]$. It now follows that 
\begin{equation}
 G\,\, \le \,\,\log (Q/\delta)\cdot \pr(\NA  \text{ doesn't terminate before phase }i) \,\,
 = \,\, \log (Q/\delta)\cdot u_{i-1}. \label{eq:scn-G-ub}
\end{equation}

\subsection{Completing proof of \Cref{lem:scn-submod-key}} Using \eqref{eq:scn-LB}, \eqref{eq:scn-G-ub} and setting $\beta = \frac{16}\delta\log(Q/\delta)$, we get 
$$ \frac{16\log(Q/\delta)  }{4} \cdot (u_i - 2u_i^*) \ \leq \ \log(Q/\delta) \cdot u_{i-1}$$ which on rearranging gives $ u_i \ \leq \ \frac{u_{i-1}}{4} + 2u_i^*$, as desired.

This completes the proof of Theorem~\ref{thm:scn-main}.

\subsection{Tight  approximation using more rounds}  

We now show that a better (and tight) approximation is achievable if
we use $2r$ (instead of $r$) adaptive rounds. The main idea is to use
the following variant of the partial covering problem, called
\emph{scenario submodular partial cover} (\ssp). The input is the same as scenario submodular cover: items $\rv{X}$, scenarios $M$ with $|M|=s$ and submodular function $f$ with maximal value $Q$. Given parameters $\delta, \varepsilon \in [0,1]$, the goal in \ssp is to  probe
some items $\rv{R}$ that realize to set $R\sse U$ such that either (i) the number of compatible scenarios is less than $\delta s$ or (ii) the function value $f(R)>Q(1-\varepsilon)$. Unlike the partial version studied previously, we do not require $f$ to be fully covered in case (ii). Note that setting $\varepsilon=\frac1Q$, we recover  the previous partial covering problem; so \ssp is more general.  
\begin{corollary}\label{cor:ssp}
  There is a non-adaptive algorithm for \ssp  with cost
  $O\left(\frac{1}{\delta} (\ln \frac1{\delta} + \ln \frac1\varepsilon )\right)$ times the 
  optimal adaptive cost  for the (full) submodular cover.
\end{corollary}
\begin{proof}
The algorithm is nearly identical to \textsc{SParCA} (Algorithm~\ref{alg:scn-parca}). In Definition~\ref{defn:partn-scn}, we change 
$$\mathcal{Z} := \{ Y \in \H(\rv{S}) : |Y| \geq \delta s \,\,\mathbf{ and } \,\,f(S(Y))\le Q(1-\varepsilon)\}.$$
In other words $\cal Z$ contains parts having size at least $\delta s$
and for which the realized function value is at most the target
$Q(1-\varepsilon)$.  Note that scenarios in ${\cal Z}$ correspond to those under which the \ssp stopping rule does not already apply. After this,  we use the same steps to produce list $L$ in  Algorithm~\ref{alg:scn-parca}. The only difference is in the stopping rule (when probing items from $L$) which reflects the  definition of \ssp. In particular, the condition in the while-loop (step~\ref{step:sparca-probe} of Algorithm~\ref{alg:scn-parca}) is now replaced by:
\begin{quote}
    While $|H|\ge \delta |M|$ and $f(R)\le Q(1-\varepsilon)$.
\end{quote}
The analysis is also nearly identical: we prove Lemma~\ref{lem:scn-submod-key} by   lower/upper bounding the total gain $G$. The lower bound~\eqref{eq:scn-LB} remains the same. For the upper bound, the analysis for the first term in \eqref{eq:scn-path-UB} remains the same, but its second term is now bounded as:
$$ \sum_{j=h}^{p} \frac{V_j -  V_{j-1}}{{Q} - V_{j-1}} \ \leq 1 + \sum_{j=h}^{p-1} \frac{V_j -  V_{j-1}}{{Q} - V_{j-1}} \le 1+ \sum_{\ell=\varepsilon Q}^{{Q}} \frac{1}{\ell} \ \leq \ 1+\ln \frac{1}{\varepsilon},$$ 
 where $V_j:=f(\{X_1, ..., X_j\})$. The second inequality uses $V_1\le V_2\le \cdots V_{p-1}\le Q(1-\varepsilon)$. This implies (just as before) that 
 $$
 G\,\, \le \, \left(1+  \log \frac1\delta + \log \frac1\varepsilon\right)\cdot u_{i-1}.$$ 
 Finally, choosing $\beta= \frac{16}\delta\left(1+ \log\frac1\delta+\log\frac1\varepsilon\right)$ and simplifying, we get   $ u_i \ \leq \ \frac{u_{i-1}}{4} + 2u_i^*$ (Lemma~\ref{lem:scn-submod-key}). This completes the proof. 
\end{proof}

Next, we show how \ssp can be used iteratively to obtain a $2r$-round algorithm.
\begin{theorem}\label{thm:scn-improved-apx}
There is a $2r$-round adaptive algorithm for   scenario submodular
  cover with cost $O\left(s^{1/r} \log (s Q)\right)$ times the optimal
  adaptive cost, where $s$ is the number of scenarios.
\end{theorem}
\begin{proof} We use \ssp with $\delta=s^{-1/r}$ and $\varepsilon=Q^{-1/r}$ repeatedly for $2r$ adaptive rounds. Let $\rho = O(\frac1r s^{1/r}\log(sQ))$ denote the approximation ratio from Corollary~\ref{cor:ssp} for these $\delta,\varepsilon$ values. 

The {\em state} at any round  consists of previously probed items and their realizations. For the analysis, we view the iterative algorithm  as  a $2r$ depth tree, where the nodes at depth $i$ are the   states   in round $i$ and the branches out of each node  represent the observed realizations in that round. For any $i\in [2r]$, let $\Omega_i$ denote all the states in round $i$: note that these form a partition of the outcome space. 

For any state $\omega\in \Omega_i$, let $\OPT(\omega)$ denote the expected cost of the optimal adaptive policy conditioned on the realizations in $\omega$.
By Corollary~\ref{cor:ssp}, conditioned on $\omega$, the expected cost in round $i$
 is at most $\rho\cdot \OPT(\omega)$.
Hence, the (unconditional) expected cost in round $i$ is at most $\rho \sum_{\omega\in \Omega_i} \Pr[\omega]\cdot\OPT(\omega) =\rho\cdot \OPT$ 
where we used that $\Omega_i$ is a partition of all outcomes. So  the  expected cost of the $2r$-round algorithm is at most $2r\cdot \rho\cdot \OPT$ as claimed. 

It remains to show that $f$ is fully covered after $2r$ rounds. Note that after each round, we have one of the following (i)  the number of compatible scenarios drops by factor $s^{1/r}$, or (ii) the residual target drops by factor $Q^{1/r}$.  Clearly, case (i) can happen at most $r$ times (as there is always some compatible scenario). So case (ii) occurs at least $r$ times, which implies that the residual target is less than $1$, i.e., $f$ is fully covered. 
\end{proof}

Setting $r=\log s$,  we   achieve a tight $O(\log(sQ))$ approximation in $O(\log s)$ rounds. Combined with the conversion to set-based rounds (Theorem~\ref{thm:set-large-r}), this proves Corollary~\ref{cor:scn}.

\fi

\ifICMLVersion
\else
\section{Lower Bound for Scenario Submodular Cover}
Recall that we obtained an $r$-round $O(s^{1/r} \cdot (\log s + r\log Q))$-approximation algorithm for scenario submodular cover. It is natural to ask if the dependence on the support-size $s$ is necessary, given that the bound for the independent case only depends on $Q$. The main result of this section is Theorem~\ref{thm:scn-lb}, which is a nearly matching lower bound on the $r$-round-adaptivity gap. In particular, it shows that the $s^{1/r}$ dependence is necessary even when $Q=1$.

\subsection{Hard Instances for Scenario Submodular Cover}\label{subsec:scn-lb-overview}
 
As a warm-up we first consider the case when $r=1$. We design an instance to prove a $1$-round adaptivity gap. The instances that we use to prove the adaptivity gap for $r > 1$ will build on this. The instance has groundset $U=\{0,1,\sym,\nil\}$ and function $f(S)=|S\cap \{\sym\}|$. Clearly, $f$ is a $0-1$ valued monotone submodular function; and $Q=1$. Let $\ell$ be an integer parameter and $N=2^\ell$. We have two types of items: $\rv{Y}_0, ..., \rv{Y}_{\ell-1}$ each of which realizes to either $0$ or $1$, and 
$N$ items $\rv{Z}_0, ..., \rv{Z}_{N-1}$ each of which realizes a symbol $\sym$ or $\nil$. All items have  cost $1$. The joint distribution $\D$ of these items has $s=N$ scenarios, each of uniform probability. For each $B\in \{0,1,\cdots, N-1\}$, scenario $B$ has the following realizations (below, $b_0 b_1 \cdots b_{\ell-1}$ is the binary representation of $B$).
$${Y}_i(B) = b_i,\,\, \forall i=0,\cdots, \ell-1 \quad \mbox{ and } \quad  {Z}_j(B) = \nil, \,\, \forall j \in \{0,\cdots, N-1\}\setminus \{B\} \mbox{ and }   {Z}_B(B) = \sym.$$
 We say that item $\rv{Z}_k$ is \emph{active} if, and only if, $\rv{Z}_k$ realizes to $\sym$. Note that there is a unique active item in each scenario, and this item needs to be probed to cover $f$. 

An adaptive solution is as follows. First, probe items $\rv{Y}_0, ..., \rv{Y}_{\ell-1}$, and based on their realizations calculate the number $B$ represented by $Y_0, \cdots, Y_{\ell-1}$. Then, probe item $\rv{Z}_B$ (the active one) to cover $f$. This policy has cost $\ell+1$. On the other hand, since each $\rv{Z}_{k}$ is equally likely to be active, any non-adaptive policy needs to pick $\Omega(2^{\ell})$ items, in expectation, to cover $f$. In conclusion, this proves an $\Omega(2^{\ell}/\ell) = \Omega(\frac{s}{\log s})$ lower bound on the $1$-round adaptivity gap for scenario submodular cover. 

\paragraph{The $r$-round instance.}
This is based on recursively applying the above instance structure. Recall that $\ell$ is some parameter  and $N = 2^{\ell}$. Moreover, the realizations of $\rv{Y}$-items above point to the unique $\rv{Z}$-item that is active. To extend this idea, we consider an $N$-ary tree $\T$ of depth $r$. For each internal node,   we index its $N$ children by $\{0,\cdots, N-1\}$ in an arbitrary order. We define a collection of items, each of unit cost, using $\T$ as follows.
\begin{itemize}
    \item For each depth $i = 0, \cdots, r-1$ and node $v\in \T$ at depth $i$, there are $\ell$ items $\rv{Y}_0(v), \cdots, \rv{Y}_{\ell-1}(v)$, each of which realizes $0$ or $1$.
    
    \item For each {\em leaf} node $w$ of $\T$ (i.e., at depth $r$),  there is an item $\rv{Z}_w$ which realizes to one of the symbols $\sym$ or $\nil$.
\end{itemize}
For any non-leaf node $v$, we use the shorthand $\rv{Y}(v) = \langle\rv{Y}_0(v), \cdots, \rv{Y}_{\ell-1}(v)\rangle$ to denote all its $\rv{Y}$-items. The function $f$ is as defined earlier: $f(S)=|S\cap \{\sym\}|$. So $f(S) = 1$ if $\sym \in S$ (else $f(S) = 0$). 
The distribution $\D$ involves $s=2^{r\ell}$ scenarios with uniform probability. Each scenario is associated with a unique leaf node (i.e. $\rv{Z}$-item); note that the number of leaves is exactly $s$. 

For each leaf $w$, we define its corresponding scenario as follows. Let $v_0, \cdots, v_{r-1}, v_r =w$ denote the path from the root of $\T$ to $w$. For each depth $i=0,\cdots, r-1$, let node $v_{i+1}$ be the $B_i^{th}$ child of node $v_{i}$. Then, scenario $w$'s realizations are:
\begin{itemize}
\item The realization $Y_0(v_i), \cdots, Y_{\ell-1}(v_i)$ of $\rv{Y}(v_i)$ equals the binary representation of $B_i$, for each depth $0\le i\le r-1$. The  realization of all other $\rv{Y}$-items is $0$. 
\item The realization $Z_w=\sym$ and $Z_u=\nil$ for all leaf $u\ne w$.
\end{itemize}
In this way, the realizations at the depth $i$ node $v_i$ point to the ``correct'' depth $i+1$ node $v_{i+1}$ that should be probed next. This  completes the description of the instance.

\paragraph{Adaptive solution.} We now describe a good  $(r+1)$-round adaptive solution $\OPT$ of cost $r\ell + 1$. It is a natural top-down strategy. It starts by probing the $\ell$ items $\rv{Y}(v_0)$ at the root $v_0$, and based on their realization calculates the number $B_0$ represented by $Y_0(v_0),\cdots, Y_{\ell-1}(v_0)$; then it probes the  $\rv{Y}$-items at the $B_0^{th}$ child of $v_0$; and so on. After $r$ rounds, it is able to discern the unique active item $\rv{Z}_w$ and probes it (in the final round) to cover $f$. 

\subsection{Lower bound proof}

Consider any $r$-round solution $\A$ for the above instance. By Yao's minimax principle, we can assume that $\A$ is deterministic. We denote by random variables $\rv{L}_0, ..., \rv{L}_{r-1}$, the ordered lists used by $\A$ in each of its $r$ rounds. Moreover, 
$\rv{S}_0, ..., \rv{S}_{r-1}$ denote the sets of items probed by $\A$ in each round; note that each $\rv{S}_i$ is a prefix of $\rv{L}_i$. Also, $\rv{L}_i$ is deterministic conditioned on the realizations of $\rv{S}_0,\cdots, \rv{S}_{i-1}$.
The following proof strategy is similar to that in \cite{AAK19} for the independent setting; however, note that the instance itself is very different.

We use $w^*$ to denote the realized scenario (leaf node) and $\rv{u}_0, \cdots,  \rv{u}_r=w^*$ to denote the (random) nodes of $\T$ on the path from the root to $w^*$. 
For any depth $0\le i\le r$,  random variable $\rv{T}_i$ denotes the set of all items at all descendant nodes of $\rv{u}_i$ (including itself).  
Note that conditioned on $w^*$, all other random variables ($\rv{L}_i$, $\rv{S}_i$, $\rv{u}_i$, $\rv{T}_i$ etc.) are deterministic. For this reason, each random  event can be seen as a subset of the scenarios.

Roughly speaking, we show that one of the following events occurs: (i) either the size of probed set $\rv{S}_i$ is large for some round $i$, or (ii) if sets $\rv{S}_0, ..., \rv{S}_{i-1}$ are small, then with high probability $\rv{S}_i$ does not contain any item from $\rv{T}_{i+1}$.  
In the first case, we are immediately done. If the second case keeps occuring then the problem eventually reduces to an instance with  $r=1$ rounds (for which we already know a lower bound). 
For a contradiction, assume:
\begin{assumption}\label{asm:LB}
The expected cost of $\A$ is at most $\frac{N}{(4r)^2}$.
\end{assumption}

\def\F{{\cal F}}
\def\G{{\cal G}}
\def\cE{{\cal E}}

For each  $i = 0, ..., r$, define the following events.
\begin{itemize}
\item  $\F_i$ is the event that 
$$\left(\rv{S}_0\cup \cdots \rv{S}_{i-1}\right) \bigcap \rv{T}_{i}=\emptyset ,$$ i.e.,  no item from $\rv{T}_{i}$ is probed {\em before} round $i$ (i.e. in rounds $0,\cdots, i-1$). 
\item $\cE_i$ is the event that
$$|\rv{S}_j|\le \frac{N}{4r},\quad \mbox{ for all }j=0,\cdots, i-1,$$
i.e., at most $\frac{N}{4r}$ items are probed in each round $0,\cdots, i-1$.
\end{itemize}
Note that $|\F_0|=|\cE_0|=N^r$ as they include  all scenarios. 

\begin{lemma}
For any $i \in \{0, \cdots, r-1\}$, we have  
${|\cE_{i+1}\cap \F_{i+1}|} \ge {|\cE_{i }\cap \F_{i }|} - \frac{N^r}{2r}.$
\end{lemma}
\begin{proof}
Note that $\E[|\rv{S}_{i }|]\le N/(4r)^2$ by Assumption~\ref{asm:LB}. So,
$$\Pr[\cE_{i}\setminus \cE_{i+1}] \le \Pr\left[|\rv{S}_{i }| >  \frac{N}{4r}\right]\le \frac{1}{4r},$$  
where the first inequality is by definition of events $\cE_i$s and the second is Markov's inequality. As all scenarios have uniform probability, $|\cE_{i}\setminus \cE_{i+1}| = N^r\cdot \Pr[\cE_{i}\setminus \cE_{i+1}] \le   N^r/(4r)$.

Let $\G:=\cE_{i+1}\cap \F_{i }$. Note that 
\begin{equation}
\label{eq:LB:scn-step}
|\G|\ge |\cE_{i }\cap \F_{i }| - |\cE_i \setminus \cE_{i+1}| \ge |\cE_{i }\cap \F_{i }| -  N^r/(4r).
\end{equation} 
We now partition scenarios in $\G$ according to the depth $i$ nodes in tree $\T$. For any node $u$ at depth $i$, let $\G(u)\sse \G$ denote the scenarios that appear in the subtree below $u$. Note $|\G| = \sum_u |\G(u)|$.

Fix any depth $i$ node $u$. Note that under any scenario of $ \G(u)$, set $\rv{T}_i$ consists of all items $\T(u)$ in the subtree below $u$.  
 We first claim that 
\begin{quote}
 The list $\rv{L}_i$ in round $i$ is the {\em same} for all scenarios $w \in \G(u)$. 
\end{quote} 
  As $w\in \F_{i}$, we know that no item in $\rv{T}_{i}=\T(u)$ is probed  in rounds $0,\cdots, i-1$  (under scenario $w$). Moreover, by construction of our instance, each item $\rv{I} \in \rv{X}\setminus \rv{T}_i=\rv{X}\setminus \T(u)$ has the same realization $I$ under all scenarios in $\G(u)$. So the realizations seen before round $i$ are identical under all scenarios in $\G(u)$, which implies that the list $\rv{L}_i$ in round $i$ is also the same (for all scenarios in $\G(u)$). 
  
  Let $\rv{P}_u$ denote the first $N/(4r)$ items in list $\rv{L}_i$ under scenarios $\G(u)$. Note that $\rv{S}_i\sse \rv{P}_i$ under all scenarios in $\G(u)\sse \cE_{i+1}$: this follows from the definition of $\cE_{i+1}$. 

Let $\G'(u) = \G(u)\cap \F_{i+1}$. We now claim that 
\begin{equation}
\label{eq:LB-scn-count}
|\G(u)\setminus \G'(u)| \le   N^{r-i}/(4r).
\end{equation}
 To see this, let $N_u$ denote all children nodes of $u$ and $C_u\sse N_u$ denote those children  that do not contain any $\rv{P}_u$-item in in their subtree. As $u$ has $N$ children, $|C_u|\ge N - |\rv{P}_u|$. Now consider any depth $i+1$ node $v\in C_u$ and scenario $w\in \G(u)$ that lies in the subtree below $v$. Note that, under scenario $w$, $\rv{T}_{i+1}=\T(v)$ the items in subtree below $v$. We now claim that  
$w\in \F_{i+1}$ (and hence $w\in \G'(u)$) because:
\begin{itemize}
\item no item probed in rounds before $i$ is in  $\rv{T}_i \supseteq \rv{T}_{i+1}$ as $w\in \G(u)\sse \F_i$.
\item no item probed in round $i$ is in $\rv{T}_{i+1}=\T(v)$; recall that  $\rv{S}_i\sse \rv{P}(u)$ (as $w\in \G(u)\sse \cE_{i+1}$) and $\rv{P}(u)\cap \T(v)=\emptyset$.
\end{itemize}
It follows that every scenario in $\G(u)\setminus \G'(u)$ appears in a subtree below a node of $N_u\setminus C_u$. Hence, 
 $|\G(u)\setminus \G'(u)| \le |N_u\setminus C_u| \cdot N^{r-i-1} \le |\rv{P}(u)|\cdot N^{r-i-1}$ as the number of scenarios under each such subtree is $N^{r-i-1}$. Equation \eqref{eq:LB-scn-count} follows by using $|\rv{P}(u)|\le N/(4r)$. 

Finally, let $\G' = \cup_u \G'(u)$ over all depth-$i$ nodes $u$.  Note that $\G'=\G\cap \F_{i+1}\sse \cE_{i+1}\cap \F_{i+1}$.
We have 
$$|\G'| = \sum_u |\G'(u)| \ge  \sum_u |\G(u)|  - N^i \frac{N^{r-i}}{4r} = |\G| - \frac{N^r}{4r} \ge |\cE_i\cap \F_i| - \frac{N^r}{2r}.$$
Above, the first inequality uses \eqref{eq:LB-scn-count} and the fact that the number of depth-$i$ nodes is $N^i$. The last inequality uses \eqref{eq:LB:scn-step}. This completes the proof.
\end{proof}

Applying this lemma repeatedly, we obtain $|\cE_{r }\cap \F_{r }|\ge |\cE_{0}\cap \F_{0}|-\frac{rN^r}{2r}\ge N^r/2$. Consider any scenario $w\in \cE_r\cap \F_r$. Note that $w\in \F_{r}$ means that, under scenario $w$, no $\rv{T}_r$ item is probed in rounds $0,\cdots, r-1$, i.e., all $r$ rounds. However, we need to probe  item $\rv{Z}_w$ in $\rv{T}_r$ in order to cover $f$ (under scenario $w$).  This is a contradiction to solution $\A$ being feasible and hence Assumption~\ref{asm:LB}. 

Thus, the cost of any $r$-round-adaptive solution must be at least $\frac{N}{(4r)^2}=\Omega(2^\ell / r^2)$. Using the fact that there is an adaptive solution of cost at most $r\ell+1$, this implies an adaptivity gap of $\Omega(\frac{1}{r^3\ell}2^\ell)=\Omega(\frac{s^{1/r}}{r^2\log s})$ as 
$s = 2^{r\ell}$. This completes the proof of  Theorem~\ref{thm:scn-lb}. Note that the above instance can also be viewed as  one for optimal decision tree: by just ignoring the $\rv{Z}$-items, and with the goal of identifying the realized scenario.  
\fi

%%%%%%%%%%%%%% computations %%%%%%%%%%%%%%%%%%%%%

\section{Applications}

\subsection{Stochastic Set Cover} 
The stochastic set cover problem is a special case of   stochastic
submodular cover. It abstracts the unreliable sensor deployment
problem mentioned in \S\ref{sec:intro}. The input is a universe $E$ of
$d$ objects and a collection $\{\rv{X}_1 ,\ldots, \rv{X}_m \}$ of $m$
items. Each item $\rv{X}_i$ has a cost $c_i \in \R_+$ and corresponds
to a random  subset of objects (with a known explicit
distribution). Different items are independent of each other. The goal
is to select a set of items such that the realized subsets cover $E$
and the expected cost is minimized. This problem was first studied by
\cite{GV06}, where it was shown that the {\em adaptivity gap}
(relative gap between the best non-adaptive and adaptive solutions) is
between $\Omega(d)$ and $O(d^2)$. The correct adaptivity gap for
stochastic set cover was posed as an open question by
\cite{GV06}. Subsequently, \cite{AAK19} made significant progress, by
showing that the adaptivity gap is $O(d\log d \cdot \log
(mc_{max}))$. However, as a function of the natural parameter $d$
(number of objects), the best adaptivity gap  remained $O(d^2)$
because the number of stochastic sets $m$ and maximum cost $c_{max}$
may be arbitrary larger than $d$.

As a corollary of our result for  (independent) stochastic submodular cover, we obtain a nearly-tight $O(d\log d)$ adaptivity gap for this problem. In fact, for any $r\ge 1$, we obtain an $r$-round adaptive algorithm that costs at most $O(d^{1/r} \log d)$ times the optimal (fully) adaptive cost. This nearly matches the $\Omega(\frac{1}{r^3} d^{1/r})$ gap for $r$-round adaptive algorithms shown in \cite{AAK19}. We note that when $r=\log d$, we obtain an $O(\log d)$-approximation algorithm with a logarithmic number of adaptive rounds;
this approximation ratio is the best possible even for deterministic set cover. 

\subsection{Optimal Decision Tree} 
The optimal decision tree problem captures problems from a variety of fields including, but not limited to, learning theory, medical diagnosis, pattern recognition, and boolean logic (see \cite{Moret82} and \cite{Murthy97} for a comprehensive survey on the use of decision trees). In an instance of optimal decision tree ($\ODT$) we are given $s$ hypotheses from which an unknown random hypothesis $y^*$ is drawn. There is a collection of $m$ tests, where test $e$ costs $c_e\in \R_+$ and returns a positive result if $y^*$ lies in some subset $T_e\sse [s]$ (and a negative result if $y^*\not\in T_e$). The goal is to identify $y^*$ using tests of minimum expected cost. 

Note that $\ODT$ is a special case of the scenario submodular cover problem. We can transform an instance of $\ODT$ into an instance for scenario submodular cover. We associate scenarios with hypotheses. For each test $e$, we have an item $\rv{X}_e$ of cost $c_e$. Let $X_e(y)$ denote the realization of $\rv{X}_e$ under scenario $y$. We set $$ {X}_e(y) = \begin{cases} [s] \setminus T_e, \qquad \text{ if } y \in T_e \\ T_e, \qquad \qquad \text{ otherwise.} \end{cases} $$ The joint distribution $\D$ realizes to $\langle X_1(y), ..., X_m(y) \rangle$ with probability $p_y$. We define the submodular function $g(S) = |\bigcup_{X_e \in S} X_e |$ where $X_e$ corresponds to the realization of item $\rv{X}_e$. Observe that $g(S)$ computes the number of eliminated scenarios. Finally, we define $f(S) = \min\left( g(S), s-1 \right)$; that is, we truncate full coverage at $s-1$. This is because we can identify the realized scenario once $s-1$ scenarios have been eliminated. Noting the $f$ is an integer-valued monotone submodular function completes the description of the instance for scenario submodular cover. 
For any $r\ge 1$, we obtain an $r$-round adaptive algorithm that costs at most $O(r s^{1/r}\log s)$ times the optimal (fully) adaptive cost. We also obtain a nearly matching lower bound: any $r$-round adaptive algorithm must incur approximation ratio $\Omega(\frac{1}{r^3} s^{1/r})$.

\subsection{Shared Filter Evaluation}
This problem was introduced by \cite{MunagalaSW07} in the context of executing multiple queries on a dataset with shared (boolean) filters. 
There are $n$ independent   ``filters'' $\rv{X}_1,\cdots, \rv{X}_n$, where  each filter $i$  has cost $c_i$ and evaluates to {\em true} with probability $p_i$ (and {\em false} otherwise). There are also $m$ conjunctive queries, where each query $j\in [m]$ is the ``and'' of some subset $Q_i\sse [n]$ of the filters. So query $j$ is true iff $X_i=true$ for all $i\in Q_j$. The goal is to evaluate all queries at the minimum expected cost. This can be modeled as an instance of stochastic submodular cover. The items correspond to filters. The groundset $U=\{T_i,F_i \}_{i=1}^n$ where $T_i$ (resp. $F_i$) corresponds to filter $i$ evaluating to {\em true} (resp. {\em false}). The submodular function is:
$$f(S) :=  \sum_{j=1}^m \min\left\{ |Q_j|\, ,\, |Q_j|\cdot |S\cap \{F_i : i\in Q_j\}| +  |S\cap \{T_i : i\in Q_j\}|\right\}.$$
Note that the term for query $Q_j$  is $|Q_j|$ iff the query's value is determined: it is {\em false} when a single false filter is seen and it's {\em true} when all its filters are true.  The maximal value of the function is $Q=\sum_{j=1}^m |Q_j| \le mn$. So, for any $r\ge 1$, we obtain an $r$-round adaptive algorithm with cost at most $O((mn)^{1/r}\cdot \log(mn))$ times the optimal adaptive cost. 

\subsection{Correlated Knapsack Cover}
There are $n$ items, each of cost $c_i$ and random (integer) reward $\rv{X}_i$. The rewards may be correlated and are given by a joint distribution $\D$ with $s$ scenarios. The exact reward of an item is only known when it is probed. Given a target $Q$, the goal is to probe some items so that the total realized reward is at least $Q$ and we  minimize the expected cost. This is a special case of scenario submodular cover. The groundset  $U=\{(i,v) \,:\, i\in[n], 0\le v\le Q\}$ where element $(i,v)$ represents item $\rv{X}_i$ realizing to $v$. Any realization of value at least $Q$ is treated as equal to $Q$: this is because  the target is itself $Q$. For each element $e=(i,v)\in U$ let $a_e=v$ denote its value. Then, the submodular function is $f(S)=\min\{\sum_{e\in S} a_e\, ,\, Q\}$ for $S\sse U$. For any $r\ge 1$, we obtain an $r$-round adaptive algorithm with cost at most $O(s^{1/r}(\log s+r \log Q))$ times the optimal adaptive cost.

\section{Computational Results} \label{sec:comp-results}

We provide a summary of computational results of our $r$-round adaptive algorithms for the stochastic set cover and optimal decision tree problems. We conducted all of our computational experiments using Python $3.8$ and Gurobi $8.1$ with a $2.3$ Ghz Intel Core $i5$ processor and $16$ GB $2133$ MHz LPDDR3 memory.

\subsection{Stochastic Set Cover}

\paragraph{Instances.} We use the Epinions network\footnote{http://snap.stanford.edu/} and a Facebook messages dataset described in \cite{nr} to generate instances of the stochastic set cover problem. The Epinions network consists of $75 \ 879$ nodes and $508 \ 837$ directed edges. We compute the subgraph induced by the top $1000$ nodes with the highest out-degree (this subgraph has $57 \ 092$ directed edges) and use this to generate the stochastic set cover instance. The Facebook messages dataset consists $1 \ 266$ nodes and $6 \ 451$ directed edges. Given an underlying directed graph, we generate an instance of the stochastic set cover problem as follows. We associate the ground set $U$ with the set of nodes of the underlying graph. We associate an item with each node. Let $\Gamma(u)$ denote the out-neighbors of $u$. We sample a subset of $\Gamma(u)$ using the binomial distribution with $p = 0.1$ for $500$ times. Let $S \sse \Gamma(u)$ be sampled $\alpha_S$ times: then $\rv{X}_u$ realizes to $S \cup \{u\}$ with probability $\alpha_S / 500$. This ensures that $\rv{X}_u$ always covers $u$. We set $f$ to be the coverage function and set $Q = \delta n$ where $n$ represents the number of nodes in the underlying network. We use $\delta= 0.5$ for the Epinions network. However, since the Facebook messages network is sparse, we set $\delta = 0.25$ in the second instance.

\ifICMLVersion
\begin{figure}
     \centering
     \begin{subfigure}[b]{0.4\textwidth}
         \centering
         \includegraphics[width=\textwidth]{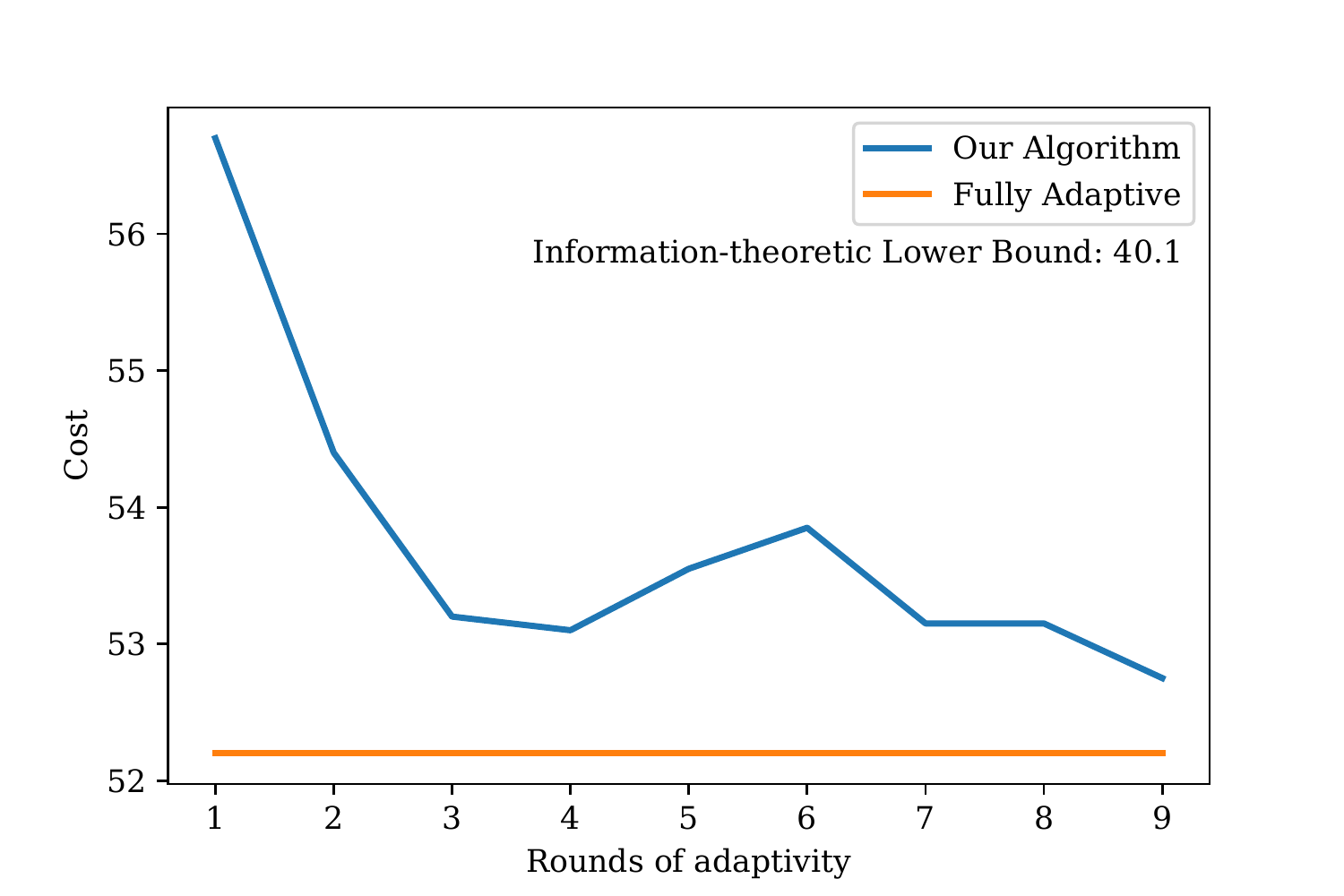}
         \caption{Epinions network}
         \label{fig:comp-ssc-epinions-rounds}
     \end{subfigure}
     \qquad 
     \begin{subfigure}[b]{0.4\textwidth}
         \centering
         \includegraphics[width=\textwidth]{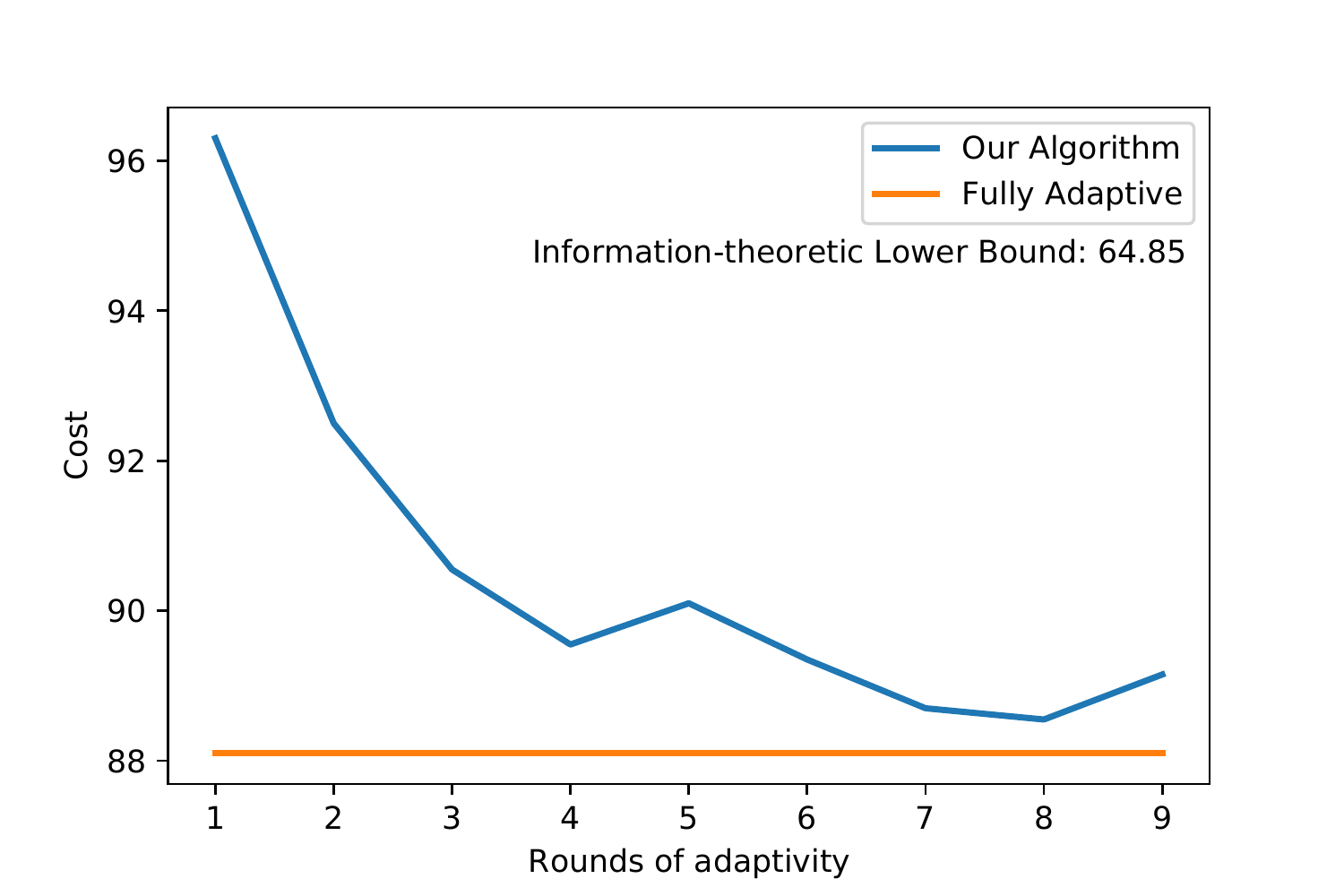}
         \caption{Facebook messages network}
         \label{fig:comp-odt-fb-rounds}
     \end{subfigure}
        \caption{Computational results for Stochastic Set Cover}
        \label{fig:comp-ssc-rounds}
\end{figure}
\else
\begin{figure}[h!]
     \centering
     \begin{subfigure}[b]{0.45\textwidth}
         \centering
         \includegraphics[width=\textwidth]{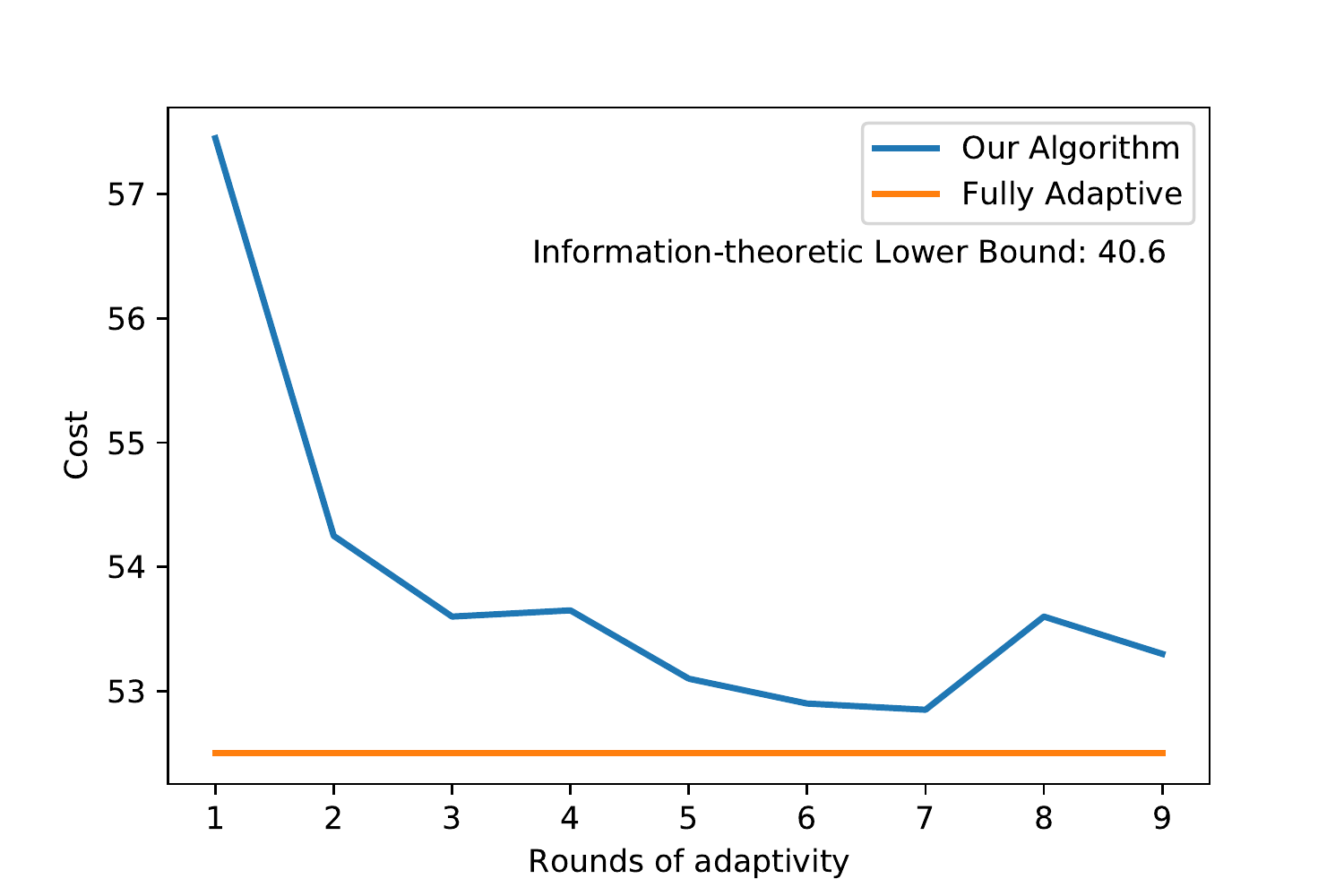}
         \caption{Epinions network}
         \label{fig:comp-ssc-epinions-rounds}
     \end{subfigure}
     \qquad 
     \begin{subfigure}[b]{0.45\textwidth}
         \centering
         \includegraphics[width=\textwidth]{figures_ssc/ia-fb-messages-final.pdf}
         \caption{Facebook messages network}
         \label{fig:comp-odt-fb-rounds}
     \end{subfigure}
        \caption{Computational results for Stochastic Set Cover}
        \label{fig:comp-ssc-rounds}
\end{figure}
\fi
\ifICMLVersion
\vspace{-0.2in}
\fi
\paragraph{Results.} We test our $r$-round adaptive algorithm on the two kinds of instances described above. We vary $r$ over integers in the interval $[1, \log n]$, where $n \approx 1000$. To compute an estimate of  the expected cost, we sample new realizations $20$ times and take the average cost incurred by the algorithm over these trials. In each trial, we also solve an  integer linear program (using the Gurobi solver) to determine the optimal \emph{offline} cost to cover $Q$ elements for the given realization: we use the average cost over the trials as an estimate on an information-theoretic lower bound: no adaptive policy can do better than this lower bound. In fact, the gap between this information-theoretic lower bound and  an optimal adaptive solution may be as large as $\Omega(Q)$ on worst-case instances. We observe that in both cases, the expected cost of our solution after only a few rounds of adaptivity is within $50\%$ of the information-theoretic lower bound. Moreover, in $6-7$ rounds of adaptivity we notice a decrease of $\approx 8\%$ in the expected  cost and these solutions are nearly as good as fully adaptive solutions. We plot this trend in Figure~\ref{fig:comp-ssc-rounds}. Finally, note that the increase in expected cost with rounds of adaptivity (see \Cref{fig:comp-ssc-epinions-rounds}) can be attributed to the the probabilistic nature of our algorithm (and the experimental setup). We also notice this in the next batch of experiments.

\subsection{Optimal Decision Tree}

\ifICMLVersion
\begin{figure*}
    \begin{subfigure}[b]{0.35\textwidth}
        \centering
        \includegraphics[width=\textwidth]{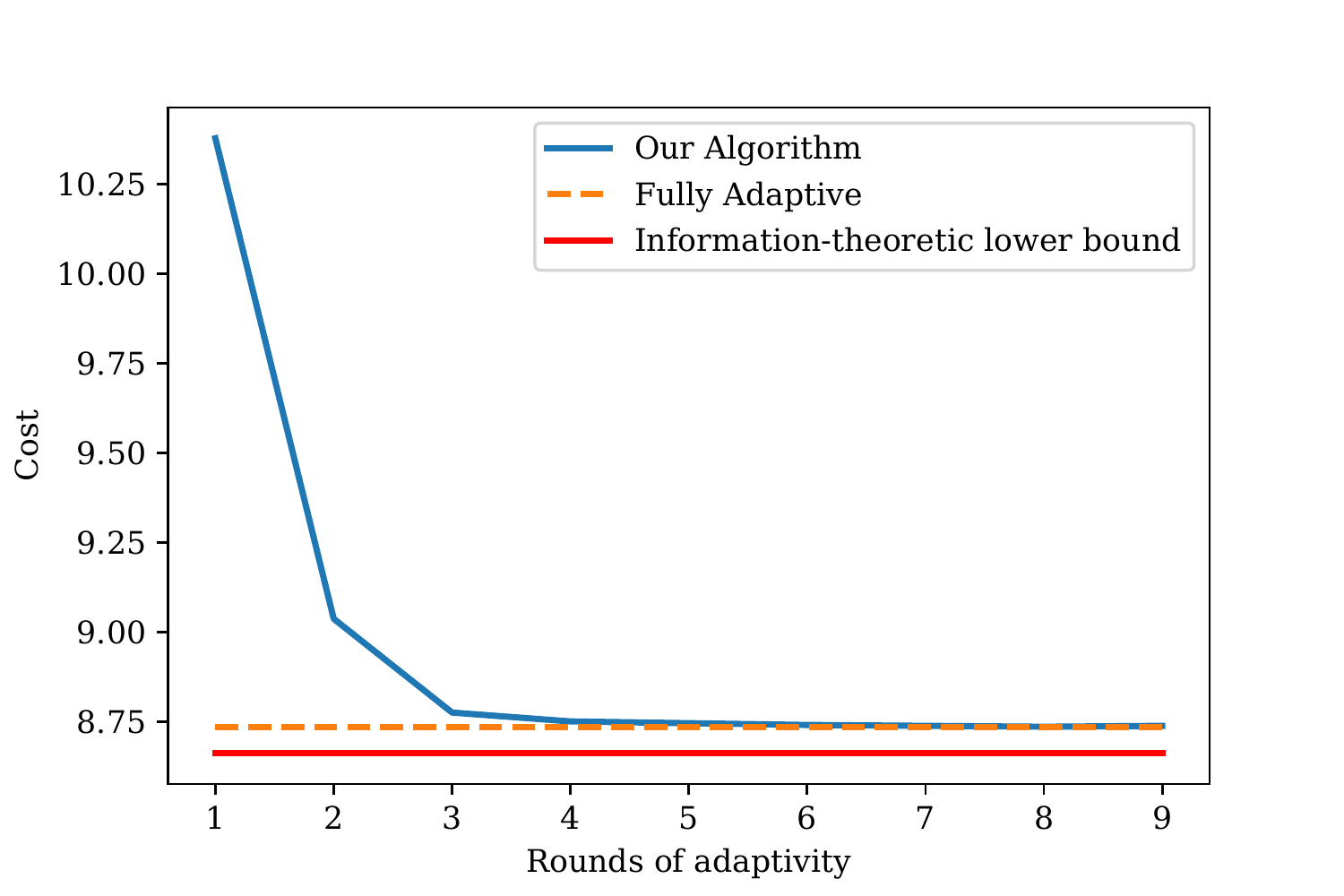}
        \caption{$\WISERu$}
        \label{fig:comp-odt-unit}
    \end{subfigure}
    \begin{subfigure}[b]{0.35\textwidth}
        \centering
        \includegraphics[width=\textwidth]{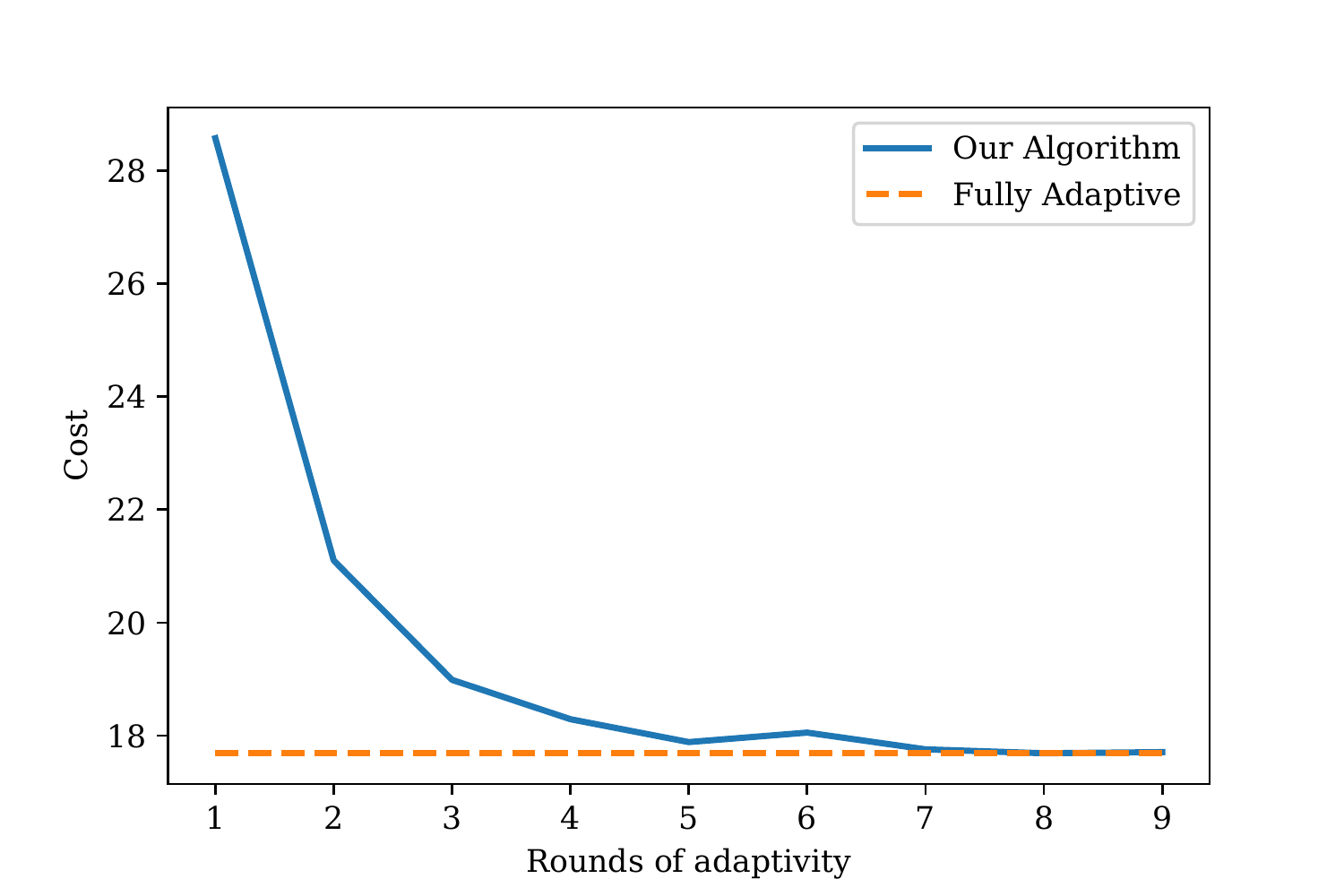}
        \caption{$\WISERr$}
        \label{fig:comp-odt-random}
    \end{subfigure}
    \begin{subfigure}[b]{0.35\textwidth}
         \centering
         \includegraphics[width=\textwidth]{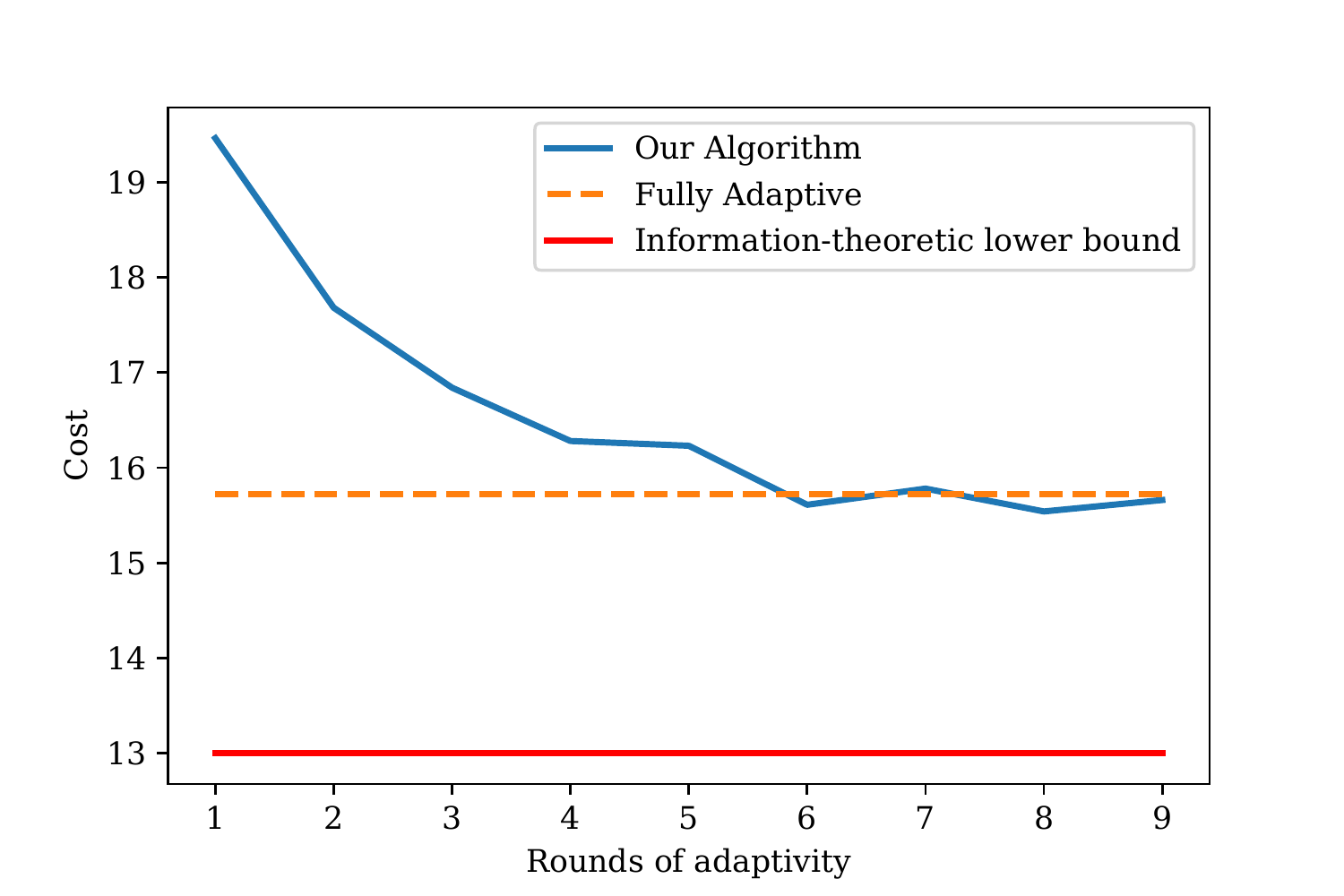}
         \caption{$\SYNu-0.2$}
         \label{fig:comp-odt-syn-1}
     \end{subfigure}
     \begin{subfigure}[b]{0.35\textwidth}
         \centering
         \includegraphics[width=\textwidth]{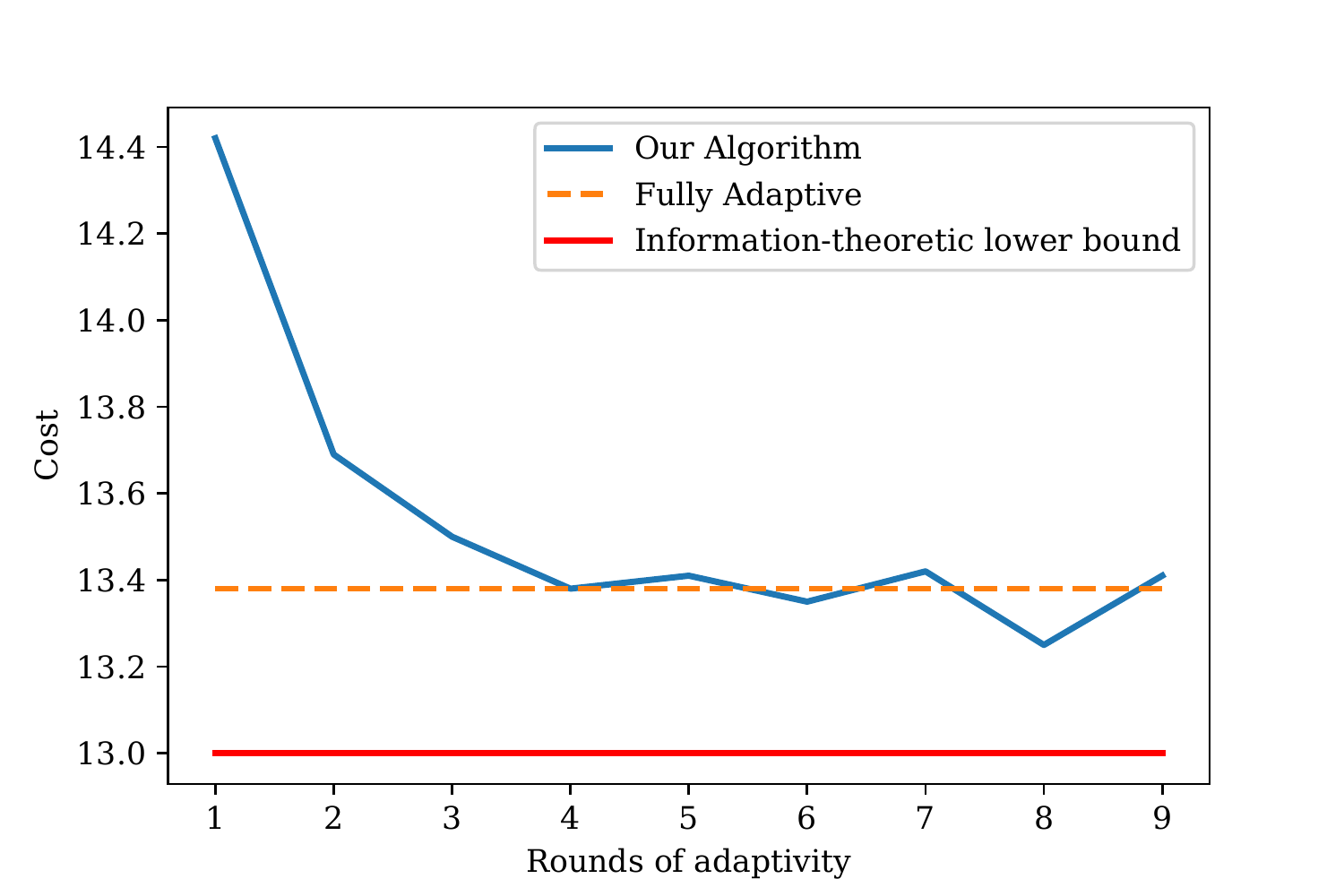}
         \caption{$\SYNu-0.5$}
     \end{subfigure}
    \begin{subfigure}[b]{0.35\textwidth}
         \centering
         \includegraphics[width=\textwidth]{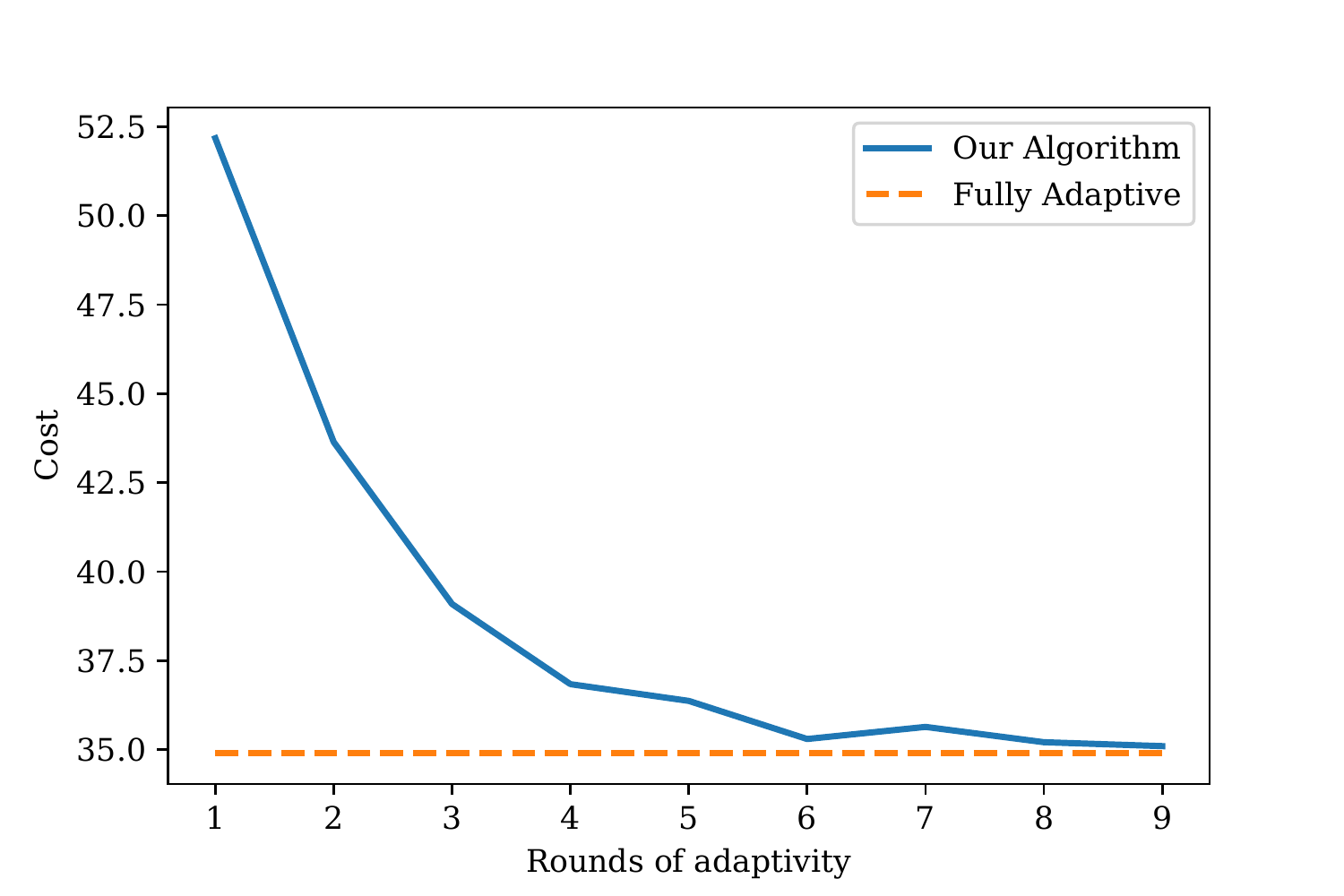}
         \caption{$\SYNr-0.2$}
     \end{subfigure}
     \begin{subfigure}[b]{0.35\textwidth}
         \centering
         \includegraphics[width=\textwidth]{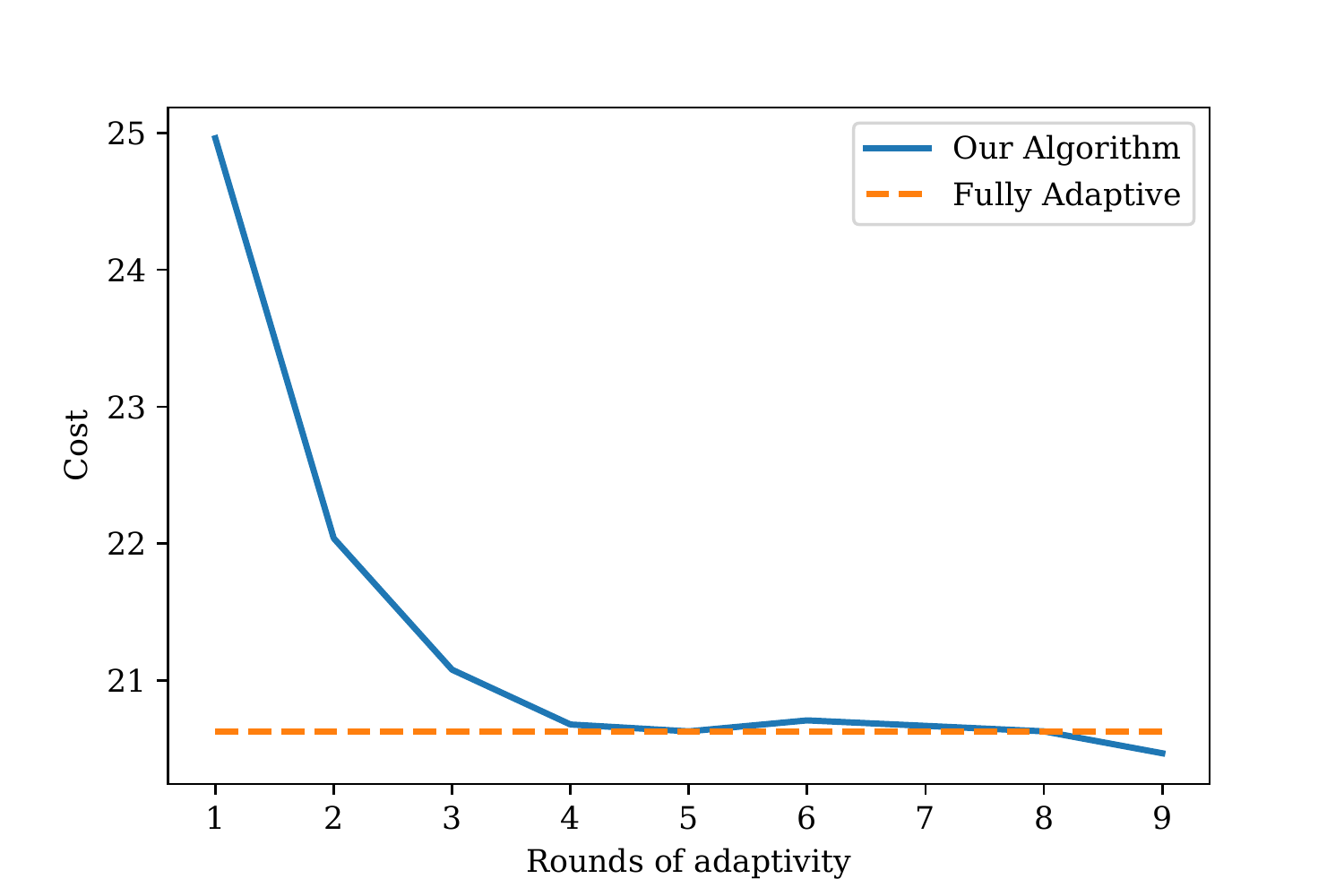}
         \caption{$\SYNr-0.5$}
         \label{fig:comp-odt-syn-4}
     \end{subfigure}
        \caption{Performance of our algorithm for the Optimal Decision Tree problem on the WISER and synthetic datasets.}
        \label{fig:comp-odt}
\end{figure*}
\else
\fi

\paragraph{Instances.} We use a real-world dataset called WISER \footnote{http://wiser.nlm.nih.gov/} for our experiments. The WISER dataset describes symptoms that one may suffer from after being exposed to certain chemicals. It contains data corresponding to $415$ chemicals (scenarios for $\ODT$) and $79$ symptoms (elements with binary outcomes). Given a patient exhibiting certain symptoms, the goal is to identify the chemical that the patient has been exposed to (by testing  as few symptoms as possible).  This dataset has been used for evaluating algorithms for similar problems in other papers \cite{BhavnaniAD+07, BBS11, BellalaBS12, NavidiKN20}. For each symptom-chemical pair, the data specifies whether or not someone exposed to the chemical exhibits the given symptom. However, the WISER data has ‘unknown’ entries for some pairs. In order to obtain instances for $\ODT$  from  this, we generate 10 different datasets by assigning random binary values to the ‘unknown’ entries. Then we remove all identical scenarios: to ensure that the $\ODT$ instance is feasible. We use the uniform probability distribution for the scenarios. Given these 10 datasets, we consider two cases. The first case considered is one where all tests have unit cost. We refer to this as the $\WISERu$ case. For the second case, we generate costs randomly for each test from $\{1, 4, 7, 10\}$ according to the weight vector $[0.1, 0.2, 0.4, 0.3]$; for example, with probability $0.4$, a test is assigned cost $7$. Varying cost captures the setting where tests may have different costs, and we may not want to schedule an expensive test without sufficient information. We refer to this as $\WISERr$ case.

We also test our algorithm on synthetic data which we generate as follows. We set $s = 10000$ and $m = 100$. For each $y \in [s]$, we randomly generate a binary sequence of outcomes which corresponds to how $y$ reacts to the tests. We do this in two ways: for test $e$, we set $y \in T_e$ with probability $p \in \{0.2, 0.5\}$. If a sequence of outcomes is repeated, we discard the scenario to ensure feasibility of the $\ODT$ instance. We assign equal probability to each scenario. We generate instances using unit costs or by assigning a random cost from $\{1, 4, 7, 10\}$ to each test according to the weight vector $[0.1, 0.2, 0.4, 0.3]$ (as described above). Thus, we generate $4$ types of instances with synthetic data. We refer to the instance generated with $p=0.2$ and unit costs as $\SYNu-0.2$. We similarly name the other instances.

\paragraph{Results.} We test our $r$-round adaptive algorithm on all of the above mentioned datasets. We vary $r$ over integers in the interval $[1, \log s]$. For the $\WISERu$ datasets, we compute the expected cost of our algorithm over \emph{all} scenarios. We plot the expected costs over all rounds, and compare it to $\log(s)$ which is an information-theoretic lower bound for the $\ODT$ problem with unit costs and uniform probabilities over the scenarios. We observe that our algorithm gets very close to this lower bound in only $3$ rounds of adaptivity. See Figure~\ref{fig:comp-odt-unit}. In the case of $\WISERr$, we compute the expected cost incurred by our $r$-round adaptive algorithm for $r$ varying over integers in the interval $[1, \log s]$ (expectation taken over \emph{all} scenarios). We observe a sharp decrease in costs within $4$ rounds of adaptivity. We plot the results in Figure~\ref{fig:comp-odt-random}.

\ifICMLVersion
\else
\begin{figure}
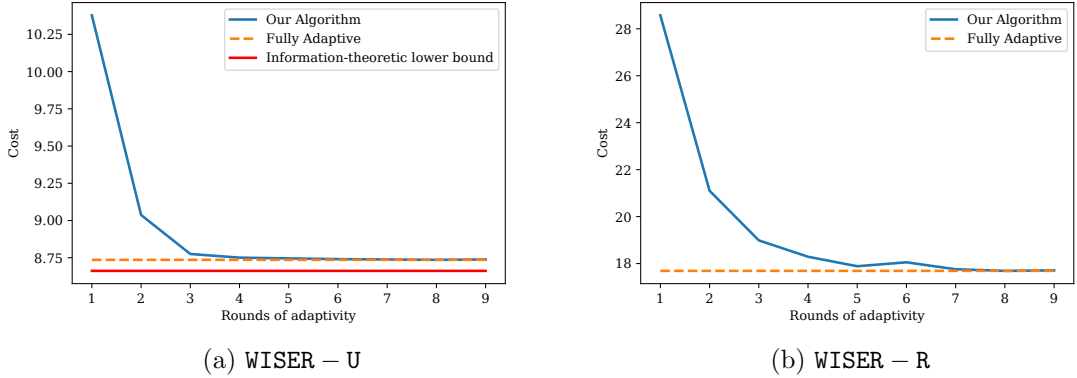

     \centering
     \begin{subfigure}[b]{0.45\textwidth}
         \centering
         \includegraphics[width=\textwidth]{figures_odt/wiser-0-unit-costs-all.pdf}
         \caption{$\WISERu$}
         \label{fig:comp-odt-unit}
     \end{subfigure}
     \begin{subfigure}[b]{0.45\textwidth}
         \centering
         \includegraphics[width=\textwidth]{figures_odt/wiser-0-rand-costs.pdf}
         \caption{$\WISERr$}
         \label{fig:comp-odt-random}
     \end{subfigure}
        \caption{Computational results for $\ODT$ on WISER dataset}
        \label{fig:comp-odt}
\end{figure}
\fi

For the synthetic data, we sample scenarios over $100$ trials (since $s=10000$, computing expectation over all $s$ would be very slow). As in the previous case, for each trial, we compute the cost incurred by our $r$-round adaptive algorithm for $r$ varying over integers in the interval $[1, \log s]$. Then, we average all these costs and use it as an estimate for the expected cost. For $\SYNu-0.2$ and $\SYNu-0.5$, we compare the results to $\log(s)$ which is an information-theoretic lower bound for the instances (since scenarios are sampled uniformly at random). We observe a improvement in the costs within $6$ rounds of adaptivity. Beyond this however, the costs do not improve (and we exclude it from our plot). We observe a similar trend for the case of $\SYNr-0.2$ and $\SYNr-0.5$. 
\ifICMLVersion
We plot the results in Figures~\ref{fig:comp-odt-syn-1}-\ref{fig:comp-odt-syn-4}.
\else
We plot the results in Figures~\ref{fig:comp-odt-syn-u} and \ref{fig:comp-odt-syn-r}.
\fi

\ifICMLVersion

\else
\begin{figure}
     \centering
     \begin{subfigure}[b]{0.45\textwidth}
         \centering
         \includegraphics[width=\textwidth]{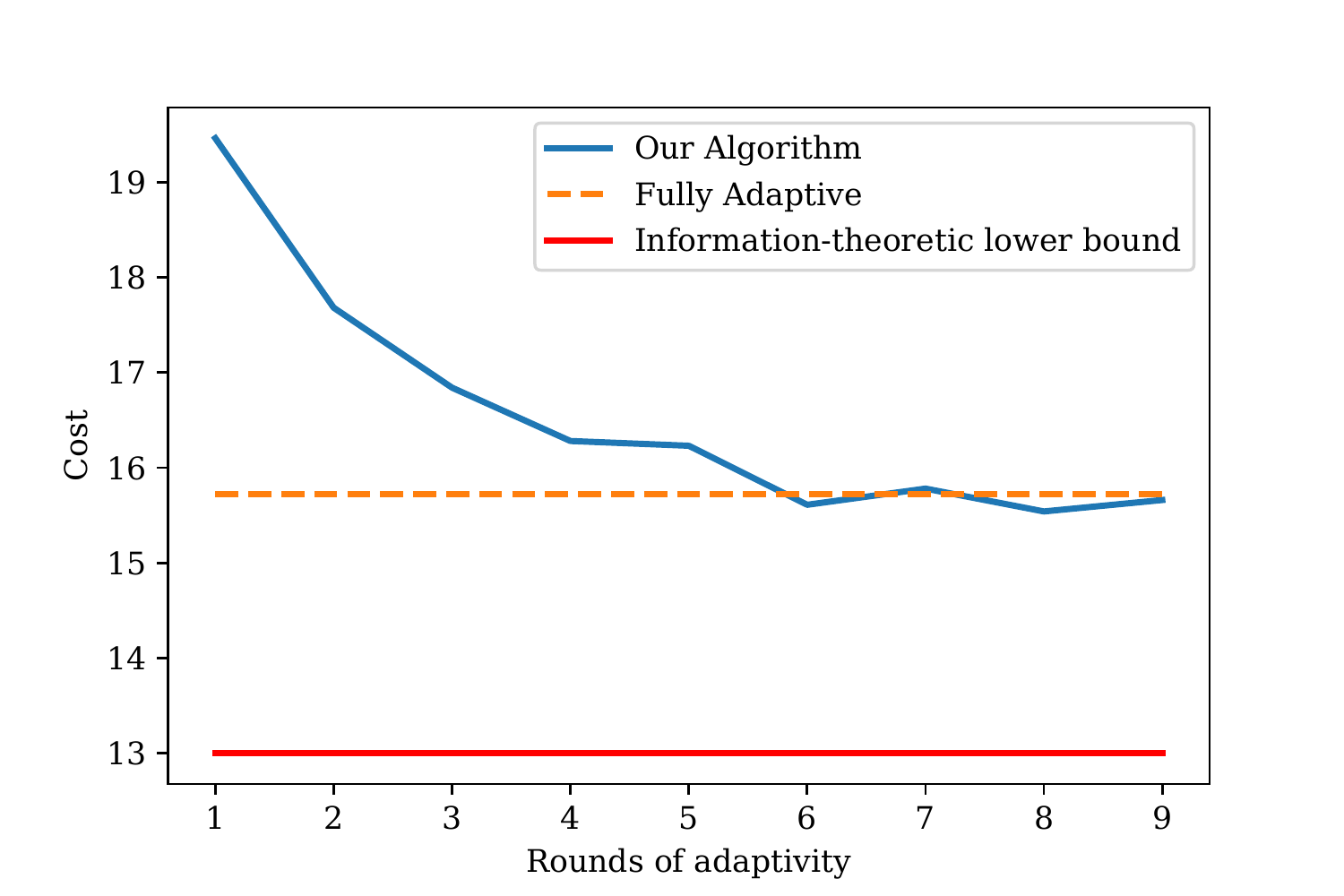}
         \caption{$\SYNu-0.2$}
     \end{subfigure}
     \begin{subfigure}[b]{0.45\textwidth}
         \centering
         \includegraphics[width=\textwidth]{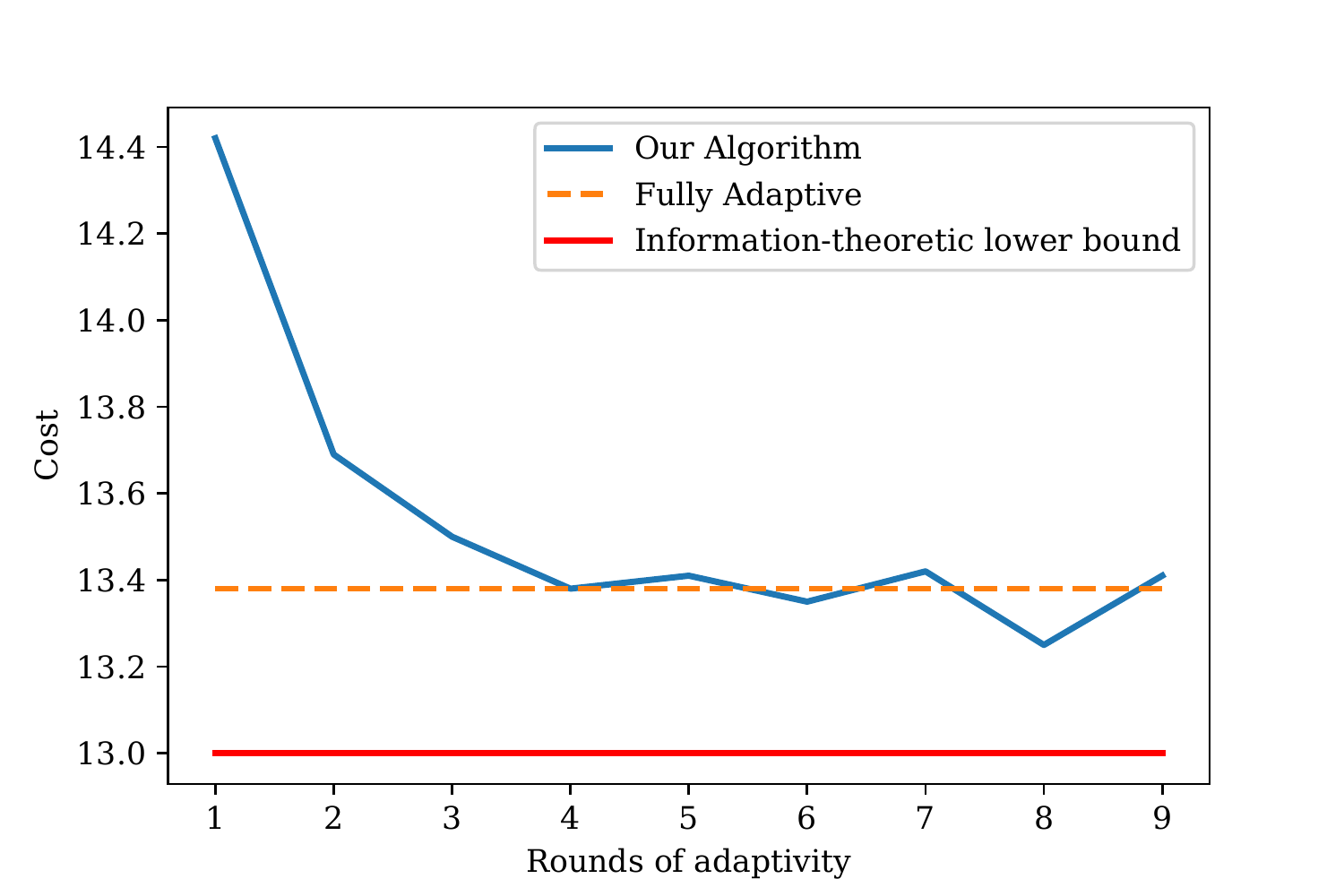}
         \caption{$\SYNu-0.5$}
     \end{subfigure}
        \caption{Computational results for $\ODT$ on synthetic data with unit costs}
        \label{fig:comp-odt-syn-u}
\end{figure}

\begin{figure}[h!]
     \centering
     \begin{subfigure}[b]{0.45\textwidth}
         \centering
         \includegraphics[width=\textwidth]{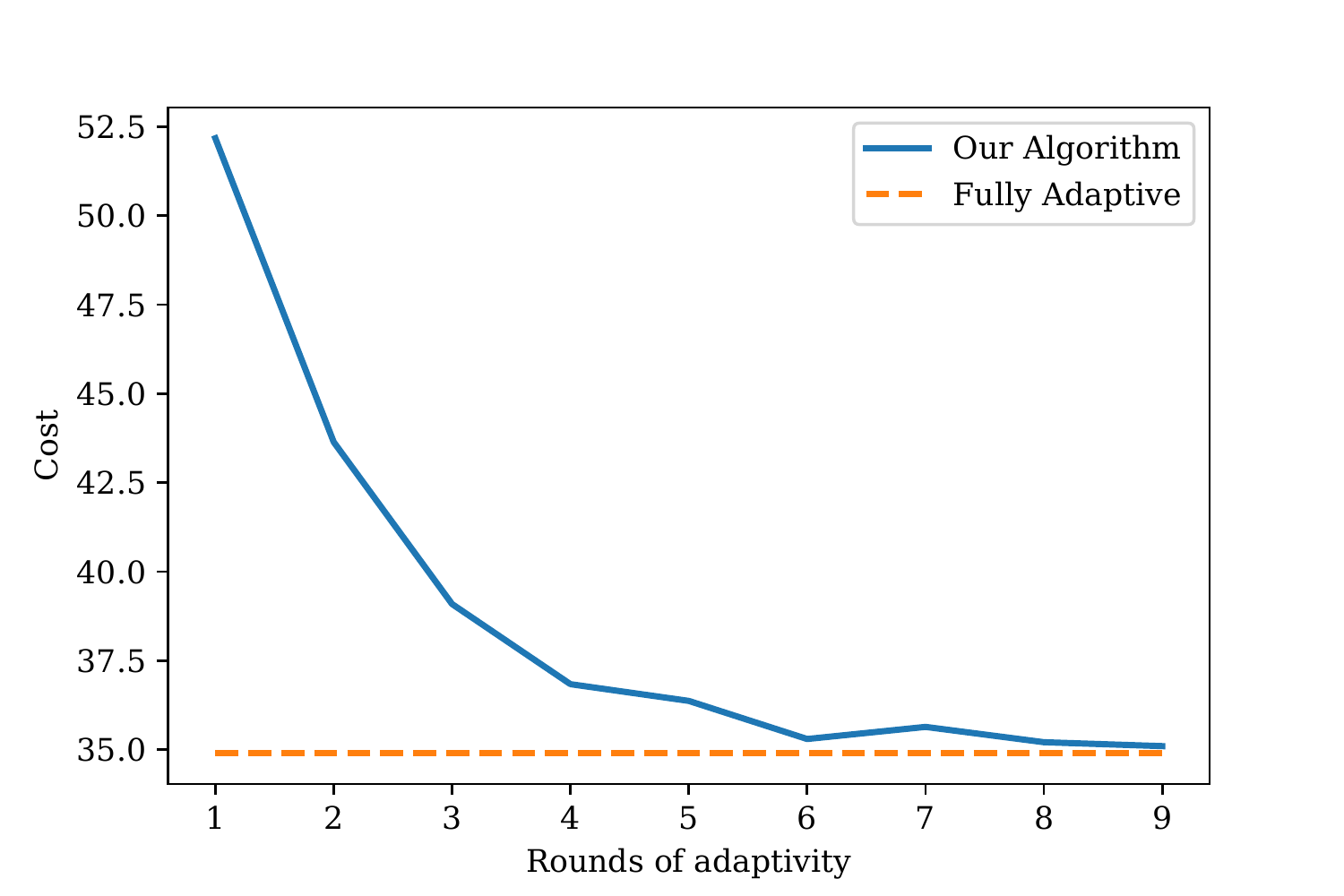}
         \caption{$\SYNr-0.2$}
     \end{subfigure}
     \begin{subfigure}[b]{0.45\textwidth}
         \centering
         \includegraphics[width=\textwidth]{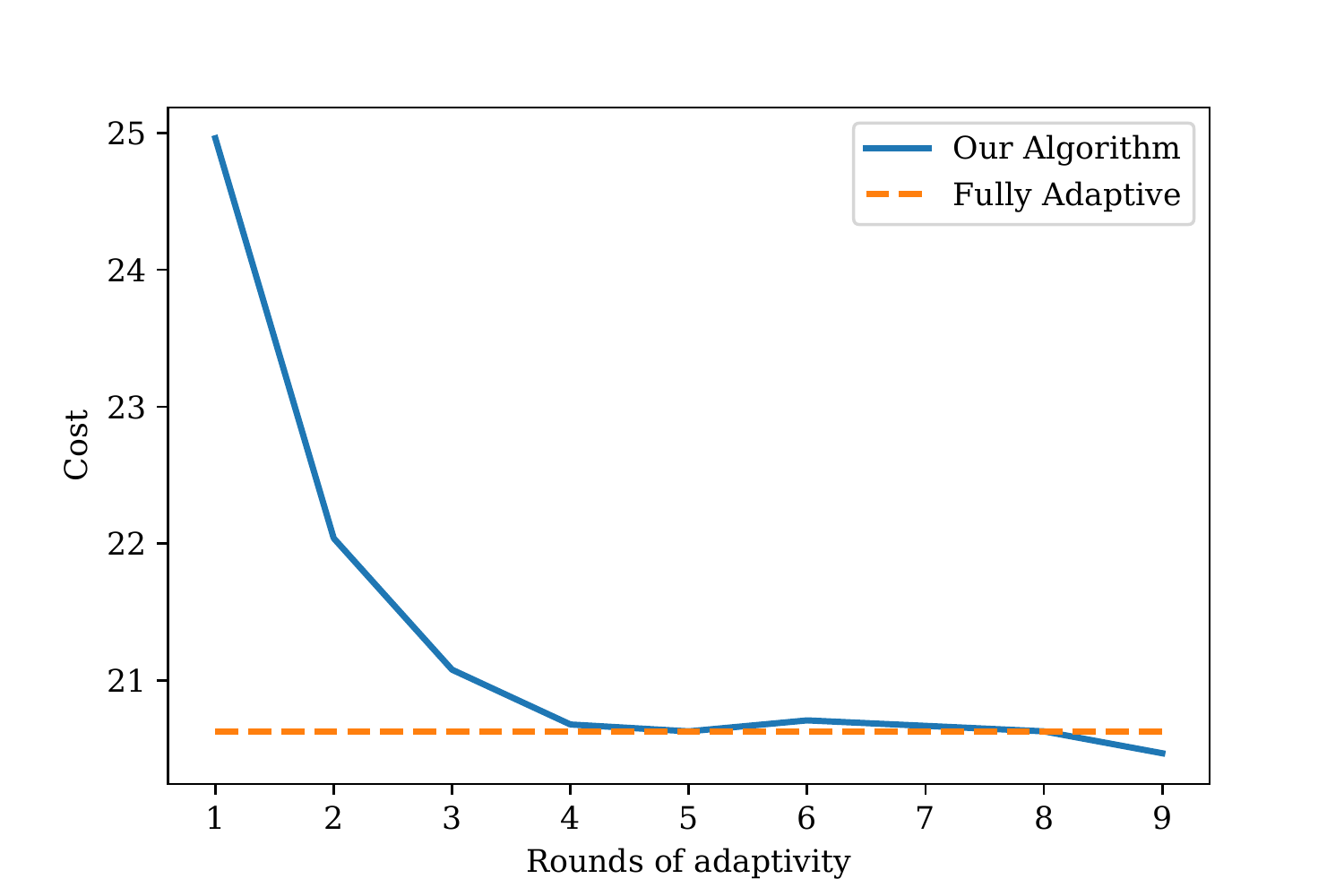}
         \caption{$\SYNr-0.5$}
     \end{subfigure}
        \caption{Computational results for $\ODT$ on synthetic data with non-uniform costs}
        \label{fig:comp-odt-syn-r}
\end{figure}
\fi

\ifICMLVersion
Note that we only plot results related to the first dataset for the $\WISERu$ and $\WISERr$ cases. We include plots for all $10$ datasets in the full version of the paper.

\else
Note that we only plot results related to the first dataset for the $\WISERu$ and $\WISERr$ cases. We include plots for all $10$ datasets in Appendix~\ref{app:computations}.
\fi

{\small
\bibliographystyle{alpha}
\bibliography{main}
}

\appendix

\section{Set-based Model for Rounds}\label{app:set-based}

Here we discuss the ``set based'' model for adaptive rounds, where each round probes a fixed subset of items (and incurs their total cost). In this model, as noted in \cite{AAK19}, we can no longer require
the function $f$ to be covered with probability one. We also provide an example below where the $r$-round-adaptivity gap is very large if we require coverage with probability one.  Therefore, we consider solutions that are only required to cover the function with high probability. We will still compare to the (fully) adaptive optimum $\OPT$ that always covers the function. 

Formally, an $r$-round adaptive solution in the set-based model proceeds as follows. For each round $i=1,\ldots,  r$, it specifies a subset $\rv{S}_i$ of items (that depends on all realizations in rounds $1,\ldots,  i-1$) that are probed in parallel. The cost incurred in round $i$ is $c(\rv{S}_i)$ the total cost of all probed items in that round. 

\paragraph{Example:} Consider an instance with $|U|=1$, and function $f$ with  $f(U)=1$ and $f(\emptyset)=0$. So parameter $Q=1$. There are $m$  items, where each item $i\in \{1,\ldots,  m-1\}$ has  cost $2^{i}$ and  instantiates to $U$ with probability $\frac12$. The last item $m$ costs $2^m$  and instantiates to $U$ with probability one. In the permutation model, there is a non-adaptive solution of cost $O(m)$, which just  probes   items in the order $1,2\ldots,  m$.  On the other hand,  in the set-based model with $r<m$ rounds we claim that the optimal  cost is $\Omega(2^{m/r})$ if we require $f$ to be covered with probability one. Note that any solution ${\cal B}$ in the set-based model with $r$ rounds is given by indices $0=i(0)\le i(1)\le \ldots,  i(r-1)\le i(r)=m$ where the $j^{th}$ round probes items $i(j-1)+1,\ldots,  i(j)$ (if $f$ is not already covered). The probability that round $j$ is needed in ${\cal B}$ is $2^{-i(j-1)}$. So the expected cost of ${\cal B}$ is at least $\sum_{j=1}^r 2^{-i(j-1)}\cdot 2^{i(j)}\ge 2^{m/r}$ where we used  $\sum_{j=1}^r (i(j)-i(j-1)) = i(r)-i(0)=m$. It follows that the $r$-round adaptivity gap in the set-based model is exponential in $m$  (while the adaptivity gap is constant for the permutation model). This is the reason our   set-based algorithms   only cover $f$ with high probability (rather than probability one).

\subsection{Conversion Theorems}
\label{sec:conversion-theorems}

While we give our main technical results in the permutation model to
make the analysis easier, 
we can
translate our $r$-round algorithm in the permutation model into one in
the set-based model, as we show next. Below, $\OPT$ is the cost of an optimal fully-adaptive solution.

\begin{theorem}\label{thm:set-small-r}
For any $r\ge 1$ and $\eta>0$, there is a set-based $r$-round adaptive
solution for stochastic submodular cover (resp. scenario submodular
cover) with  expected cost at most $O(\frac{r}{\eta} Q^{1/r} \log
Q)\cdot \OPT$ (resp. $O\big(\frac{r}{\eta} (s^{1/r} \log s+ r\log
Q)\cdot \OPT\big)$) and covers the function with probability at least
$1-\eta$. 
\end{theorem}
\begin{proof}
We will show the following ``black box'' reduction. 
Suppose there is an $r$-round adaptive algorithm in the permutation model with approximation ratio $\alpha$. Then, there is an algorithm in the set-based model that for any $\eta \in (0,1)$, finds an $r$-round adaptive solution that has expected cost at most $\frac{r\alpha}{\eta }\cdot \OPT$ and  covers the function with  probability at least $1-\eta$. The theorem would then follow from Theorems~\ref{thm:main-ssc} and \ref{thm:scn-main}.

Consider any  solution ${\cal A}$ in the permutation-model that always covers the function. For each round $i=1,\ldots,  r$ in ${\cal A}$, let $\rv{L}_i$ be the permutation to probe items and let $\mu_i$ be the expected cost of probed items (after applying the stopping rule). Note that $\rv{L}_i$  and $\mu_i$ depend on realizations in all previous rounds. 
Let $\rv{S}_i$ denote the maximal prefix of $\rv{L}_i$ that has total cost at most $\frac{r}{\eta}\mu_i$. We define the set-based solution ${\cal B}$ as follows. In round $i$, we probe all items $\rv{S}_i$ in parallel. Note that the cost of solution ${\cal B}$ in round $i$ is at most $\frac{r}{\eta}\mu_i$. 
Taking expectations and adding over all rounds, the expected cost of ${\cal B}$ is at most $\frac{r}{\eta}$ times that of ${\cal A}$. So the cost is at most $\frac{r\alpha}{\eta }\cdot \OPT$.

We refer to round-$i$ as a {\em failure} if  the round-$i$ stopping rule (in ${\cal A}$) {\em does not} apply by the end of $\rv{S}_i$.  By Markov's inequality, the probability that round-$i$ is a failure is at most $\frac{\eta}{r}$. 
Using a union bound over all $r$ rounds, the probability of a failure in any round is at most $\eta$. Note that, if there is no failure then solution ${\cal B}$ indeed covers the function fully: as the items probed by ${\cal B}$ is a superset of those probed by ${\cal A}$. Hence, solution ${\cal B}$ covers the function with probability at least $1-\eta$. 
\end{proof}

Note that the  approximation ratios for the set-based model are only a factor $r$ more than in the permutation model (assuming that  the ``failure probability'' $\eta$ is any constant). So this is a good result for small (constant) $r$. When the number of rounds $r$ is larger, we have a different approach that does not lose this factor $r$ in the approximation ratio, but uses $O(r)$ set-based rounds. 
\begin{theorem}\label{thm:set-large-r}
For any $r\ge 1$, in the  set-based model: 
\begin{itemize}
    \item  there is a 
$2r$-round adaptive  algorithm for stochastic submodular cover  with  expected cost at most $O(Q^{1/r} \log Q)\cdot \OPT$ 
\item  there is a 
$4r$-round adaptive  algorithm for scenario submodular cover  with  expected cost at most  $O(s^{1/r} \log (s Q))\cdot \OPT$  
\end{itemize}
that covers the function with probability at least $1-e^{-\Omega(r)}$. 
\end{theorem}
\begin{proof}
We first consider stochastic submodular cover. We make use of the partial cover algorithm   \textsc{ParCA}  (Theorem~\ref{thm:partial-cover}) in an iterative manner (similar to the $r$-round algorithm  in the permutation model). The difference is that we will run    \textsc{ParCA} for $2r$ set-based rounds (instead of $r$). 
The items probed in each set-based round is a prefix of the permutation-based round. Algorithm~\ref{alg:2r-round-set} gives a formal description. The {\em state} at any round is the tuple $(\rv{T},T)$ consisting of previously probed items and their realizations.  For the analysis, we view the iterative algorithm  as  a $2r$ depth tree, where the nodes at depth $i$ are the   states   in round $i$ and the branches out of each node  represent the observed realizations in that round. For any $i\in [2r]$, let $\Omega_i$ denote all the states in round $i$: note that these form a partition of the outcome space.
\begin{algorithm}[h]
\caption{$2r$-round set-based  algorithm for stochastic submodular cover} \label{alg:2r-round-set}
\begin{algorithmic}[1]
\State let $\rv{T}\gets \emptyset$ denote the probed items and $T\gets \emptyset$ their realization.
\For{$i=1,2,\ldots,  2r$}
\State run  \textsc{ParCA} $(\rv{X}\setminus \rv{T}, f_T, Q-f(T), Q^{-1/r})$. \label{step:set-round}
\State let $\rv{L}_i$ denote the permutation on items $\rv{X}\setminus \rv{T}$ from   \textsc{ParCA}.
\State let  $\mu_i$ be the expected cost of probing $\rv{L}_i$ (in the permutation model).
\State let $\rv{R}_i$ denote the maximal prefix of $\rv{L}_i$ of cost at most $4\mu_i$.
\State probe items $\rv{R}_i$ in the set-based solution and let $R_i$ denote their realizations.  
\State update $\rv{T}\gets \rv{T}\cup \rv{R}_i$ and $T\gets T\cup R_i$.
\EndFor
\end{algorithmic}
\end{algorithm}
\paragraph{Expected cost.} For any state $\omega\in \Omega_i$, let $\OPT(\omega)$ denote the expected cost of the optimal adaptive policy conditioned on the realizations in $\omega$.
By Theorem~\ref{thm:partial-cover}, conditioned on $\omega$, the expected cost $\mu_i(\omega)$ of the permutation $\rv{L}_i(\omega)$ is at most $\alpha\cdot \OPT(\omega)$ where $\alpha=O(\frac1r Q^{1/r} \log Q)$. Note that the cost of the set-based round is at most $4\cdot \mu_i (\omega)$.
Hence, the (unconditional) expected cost in round $i$ is at most  
$$4\sum_{\omega\in \Omega_i} \Pr[\omega]\cdot \mu_i(\omega) \le 4\alpha \sum_{\omega\in \Omega_i} \Pr[\omega]\cdot\OPT(\omega) = 4\alpha\cdot \OPT,$$
where we used that $\Omega_i$ is a partition of all outcomes. It follows that the expected cost of the set-based algorithm is at most $2r\cdot 4\alpha\cdot \OPT=8r\alpha\cdot \OPT$ as claimed.

\paragraph{Completion probability.} Consider any round $i$ and state $\omega\in \Omega_i$. Recall that $T$ denotes the realizations from prior rounds. The stopping rule for the permutation  $\rv{L}_i(\omega)$ in \textsc{ParCA} is that the residual target  drops by a factor $\delta=Q^{-1/r}$. In other words, if $S$ denotes the  realizations observed in permutation $\rv{L}_i(\omega)$ then the stopping rule is $Q-f(T\cup S)< \delta\cdot (Q-f(T))$. Recall that the set-based   round-$i$ involves probing the prefix $\rv{R}_i(\omega)$ of $\rv{L}_i(\omega)$ having cost $4\cdot \mu_i(\omega)$. This round
is said to be a {\em success} if  the stopping rule applies by the end
of $\rv{R}_i(\omega)$.  By the choice of $\rv{R}_i$ and Markov's
inequality, the probability that this round  is a success is at least
$\frac34$. So the expected number of successful rounds is at least
$2r\cdot \frac34=\frac32 r$. Moreover, the success events in different
rounds $i=1,\ldots,  2r$ are independent. Hence, by a Chernoff bound, the probability of seeing less than $r$ successes out of $2r$ rounds is at most $e^{-\Omega(r)}$. Finally, note that if there are at least $r$ successes then the function $f$ gets covered completely as each success reduces the residual target by a factor $Q^{-1/r}$. It follows that our algorithm covers $f$ with probability    $1-e^{-\Omega(r)}$ as claimed.

The proof for scenario submodular cover is identical, except for the use of \textsc{SParCA} (Corollary~\ref{cor:ssp}) instead of \textsc{ParCA}. 
\end{proof}

\section{Items realizing to subsets}\label{app:set-items}
Here, we consider a seemingly more general model for stochastic submodular cover, where each item realizes to a {\em subset} of elements. As before, there is  a monotone submodular function $f: 2^U \to \ZZ_{\geq 0}$ with $Q := f(U)$. There are $m$ items, where each item $\rv{X}_i$ has  cost $c_i$ and realizes  to a random
subset of $U$. (In the basic model, each item realizes to a single element.)  
  The goal is to select a set of items such that the
union $S$ of their corresponding realizations satisfies $f(S) = Q$, and the
expected cost is minimized. We can reduce this model to the usual one (considered in the paper) as follows. Let $W:=2^U$ denote an expanded groundset consisting of all subsets of $U$. Let function $g:2^W\rightarrow \ZZ_{\ge 0}$ be defined as $g(R)=f(\cup_{r\in R} r)$ for any $R\sse W$. Note that $g$ is monotone submodular as $f$ is. Note that each item $\rv{X}_i$ realizes to a single element of the expanded set $W$. Moreover, covering function $g$ is equivalent to covering $f$, and  both have maximal value $Q$.  So it suffices to solve the usual model with single-element realizations. However, the  difficulty with this reduction is that the number of elements $|W|$ is exponential. Using the fact that our algorithm's runtime is independent of the size of the groundset (see remark after Algorithm~\ref{alg:r-round-ind}), it follows that we obtain a polynomial (in $m,Q,c_{\max}$) time algorithm for this more general model as well.

\section{Estimating Scores in Algorithm  \textsc{ParCA}} \label{app:sampling}
The  partial covering algorithm \textsc{ParCA} (Algorithm~\ref{alg:parca}) relies on computing the maximum score according to \eqref{eq:greedy-choice}. We are not aware of a closed-form expression to calculate the score for general functions $f$. Instead, we show that sampling can be used to estimate these scores, which suffices for the overall algorithm. At any point in algorithm  \textsc{ParCA}, let $\rv{S}\sse \rv{X}$ denote the items already added to the non-adaptive list $L$.  Then, we need to compute the maximum $\text{score}(\rv{X}_e)=\frac{g_e}{c_e}$ over $\rv{X}_e\in \rv{X} \setminus \rv{S}$, where
$$g_e  :=   \sum_{S \sim \rv{S} : f(S)\le \tau } \pr(\rv{S}=S)\cdot \E_{X_e \sim \rv{X}_e} \left[  \frac{f(S \cup X_e) - f(S)}{Q - f(S)}\right]
 = \E_{\rv{S},\rv{X}_e} \left[ \mathbf{1}_{f(S)\le \tau} \cdot \frac{f(S \cup X_e) - f(S)}{Q - f(S)}\right].$$ 
To this end, we use $K=O(m^2 c_{\max} \log(mc_{\max}))$ independent samples to obtain an estimate $\bar{g_e}$ for each item $e$. Then, we choose the item that maximizes $\frac{\bar{g_e}}{c_e}$. Note that the time taken in each step is $O(m\cdot K)$, so the overall time taken by algorithm \textsc{ParCA} is $O(m^2K)=\text{poly}(m,c_{\max})$. 

Let $L$ denote the list produced by the (sampling based) \textsc{ParCA} algorithm, and let $\NA$ be the resulting non-adaptive solution. We now show that  the expected cost of  $\NA$ is $O(\frac1\delta \log \frac{1}{\delta})\cdot \OPT$. As before, $\OPT$ is the cost of an optimal fully-adaptive solution. This would prove Theorem~\ref{thm:partial-cover} even for the sampling-based \textsc{ParCA} algorithm.  

\begin{lemma}\label{lem:low-score}
At any step, with already added items $\rv{S}$, we have $\Pr[f(S)\le \tau] \le \sum_{\rv{X}_e\in \rv{X} \setminus \rv{S}} g_e$.
\end{lemma}
\begin{proof} We can write:
\begin{align*}
    \sum_{\rv{X}_e\in \rv{X} \setminus \rv{S}} g_e & = \sum_{\rv{X}_e\in \rv{X} \setminus \rv{S}} \E_{\rv{S},\rv{X}_e} \left[ \mathbf{1}_{f(S)\le \tau} \cdot \frac{f(S \cup X_e) - f(S)}{Q - f(S)}\right] \, =\, \E_{\rv{X}}\left[ \mathbf{1}_{f(S)\le \tau} \cdot\sum_{\rv{X}_e\in \rv{X} \setminus \rv{S}} \frac{f(S \cup X_e) - f(S)}{Q - f(S)} \right] \\
    &\ge \E_{\rv{X}}\left[ \mathbf{1}_{f(S)\le \tau} \cdot \frac{f(X) - f(S)}{Q - f(S)} \right]\, =\, \E_{\rv{X}}\left[ \mathbf{1}_{f(S)\le \tau} \right] \,=\, \Pr[f(S)\le \tau].
\end{align*}
Above, the inequality is by submodularity and the second-last equality uses the fact that $f(X)=Q$ with probability one. 
\end{proof}

Let $L'$ denote the maximal prefix of list $L$ where $\max_e g_e \ge \varepsilon := \frac{1}{m^2c_{\max}}$ at each step.  Let $\NA'$ denote the cost incurred by the 
non-adaptive solution on items in $L'$ and $\NA''$ be the cost incurred on items $L\setminus L'$. Clearly, $\E[\NA] = \E[\NA']+\E[\NA'']$.

By Lemma~\ref{lem:low-score}, it follows that after $L'$ ends, we have 
$\Pr[f(S)\le \tau] \le m\cdot \max_e g_e \le m\varepsilon \le \frac{1}{mc_{\max}}$. In other words, $\Pr[\NA \mbox{ continues after }L']\le \frac{1}{mc_{\max}}$ and hence  $$\E[\NA'']\le \Pr[\NA \mbox{ continues after }L']\cdot mc_{\max} =O(1).$$ 
It now remains to bound $\E[\NA']$. 

Let $\bar{\NA}$ denote any non-adaptive solution obtained by algorithm \textsc{ParCA} assuming that it always selects an item $\rv{X}_e$ with $\text{score}(\rv{X}_e)\ge \frac14 \max_i  \text{score}(\rv{X}_i)$. In  words, $\bar{\NA}$ is built using items  having $\frac14$-approximately maximum score. It can be easily verified that the analysis in \S\ref{subsec:analysis-PARCA} can be extended to prove $\E[\bar{\NA}]\le O(\frac1\delta \log \frac{1}{\delta})\cdot \OPT$. (The only change is an additional factor of $\frac14$ in the right-hand-side of \eqref{eq:G-LB}.)
We now bound $\E[\NA']$ using $\E[\bar{\NA}]$.
\begin{lemma}\label{lem:sampling}
Consider any item $\rv{X}_e$ in $L'$ (with $\rv{S}$ denoting all previous items). Then, we have $\text{score}(\rv{X}_e)\ge \frac14 \max_{\rv{X}_i\in \rv{X}\setminus \rv{S}}  \text{score}(\rv{X}_i)$ with  probability at least $1-\frac{1}{m^2c_{\max}}$. 
\end{lemma}
\begin{proof} We partition the remaining items $\rv{X}\setminus \rv{S}$ into  $I^+=\{i : g_i\ge \varepsilon/4\}$ and $I^- = \{j : g_j < \varepsilon/4\}$.

Consider any item $i\in I^+$: so  $g_i\ge \varepsilon/4$. As
$\bar{g_i}$ is the average of $K$ independent samples each of mean
$g_i$, using a Chernoff bound we obtain  $\Pr[\frac12 g_i \le \bar{g_i} \le 2g_i] \ge 1-e^{-\Omega(\varepsilon K)} \ge 1-\frac{1}{m^3 c_{\max}}$.
Now consider an item $j\in I^-$: so  $g_j < \varepsilon/4$. By a Chernoff bound again, $\Pr[\bar{g_j} > \varepsilon/2] \le e^{-\Omega(\varepsilon K)}\le \frac{1}{m^3 c_{\max}}$.
So, with probability at least $1-\frac{1}{m^2 c_{\max}}$, 
$$\frac12 g_i \le \bar{g_i} \le 2g_i \mbox{ for all }i\in I^+, \mbox{ and }\bar{g_j}\le \frac{\varepsilon}{2} \mbox{ for all }j\in I^-.$$ 
Below, we condition on this event. 

As  item $e$ is in $L'$, we know that $\max_i g_i \ge \varepsilon$. So $I^+\ne \emptyset$. Recall that $\bar{g_e} = \max_{i} \bar{g_i}$. 
Combined with the above,  it follows that $e\in I^+$. Moreover, 
$$g_e\ge \frac12 \bar{g_e} = \frac12 \max_{i} \bar{g_i} = \frac12  \max_{i\in I^+}  \bar{g_i} \ge \frac14  \max_{i\in I^+}   g_i = \frac14  \max_{i}   g_i.$$
This completes the proof.
\end{proof}
We say that the sampling was successful if the event in  Lemma~\ref{lem:sampling} occurs for every item in $L'$. As there are at most $m$ items, sampling is successful  
with probability at least $1-\frac{1}{mc_{\max}}$. Conditioned on the sampling being successful,  $L'$ is a prefix of some non-adaptive solution  $\bar{\NA}$ that 
always  selects an item   having $\frac14$-approximately maximum score. Hence,
$$\E[\NA'  \, | \mbox{sampling successful}] \le \E[\bar{\NA}] \le O(\frac1\delta \log \frac{1}{\delta})\cdot \OPT.$$
If sampling is not successful then $\NA'$ costs at most $mc_{\max}$. So,
$$\E[\NA'] \le \E[\NA'  \, | \mbox{sampling successful}] + \Pr(\mbox{sampling not successful})\cdot m c_{\max} \le O\bigg(\frac1\delta \log \frac{1}{\delta}\bigg)\cdot \OPT.$$
It now follows that $\E[\NA]=\E[\NA']+\E[\NA'']\le O(\frac1\delta \log \frac{1}{\delta})\cdot \OPT$ as desired.

\section{Additional Plots}\label{app:computations}
In the main body of the paper, we only include plots related to the first dataset for the $\WISERu$ and $\WISERr$ cases. The trends observed in the other datasets are similar. We include plots for all $10$ datasets here for completeness.

\begin{figure}
     \centering
     \begin{subfigure}[b]{0.4\textwidth}
         \centering
         \includegraphics[width=\textwidth]{figures_odt/wiser-0-unit-costs-all.pdf}
     \end{subfigure}
     \begin{subfigure}[b]{0.4\textwidth}
         \centering
         \includegraphics[width=\textwidth]{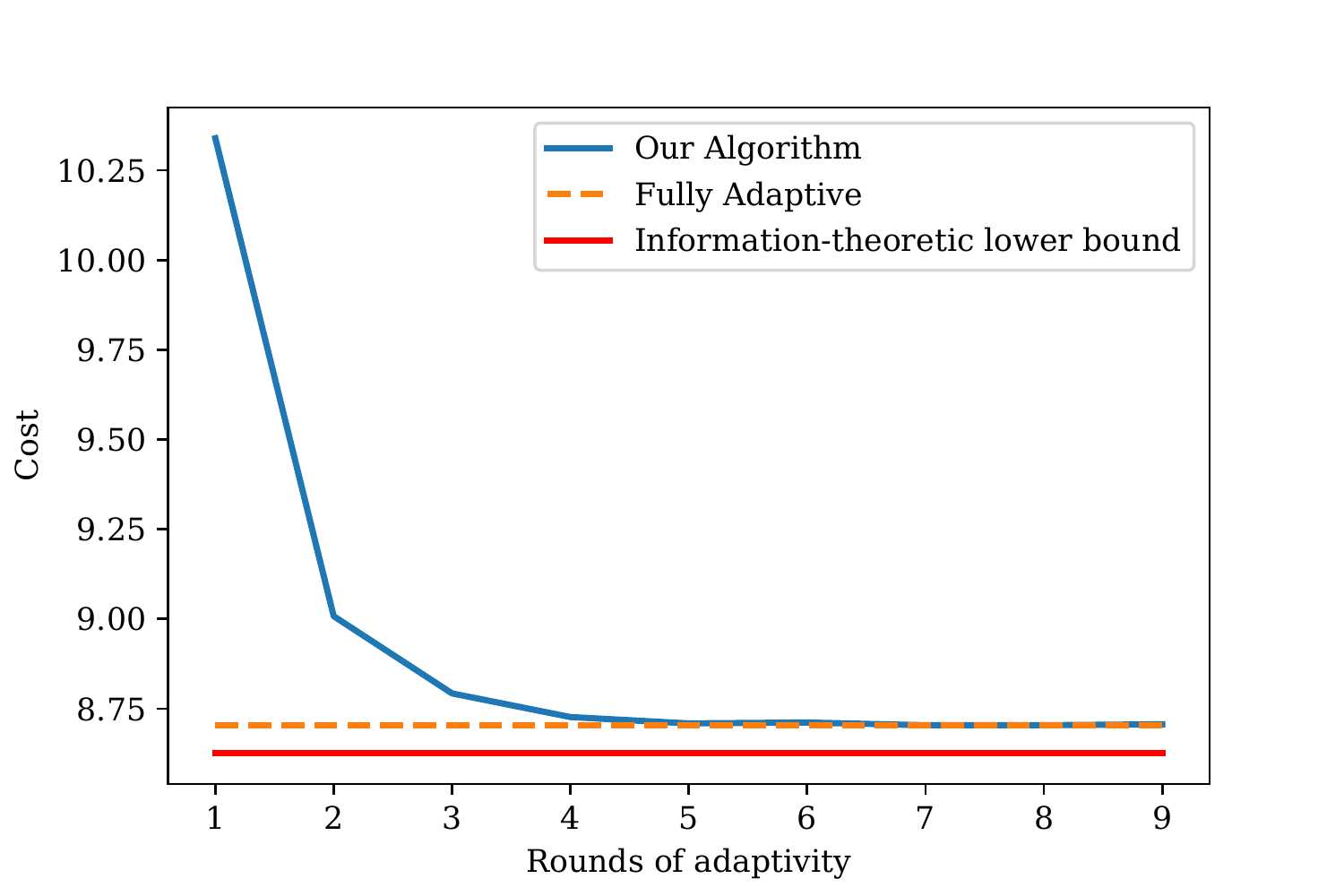}
     \end{subfigure}
     \begin{subfigure}[b]{0.4\textwidth}
         \centering
         \includegraphics[width=\textwidth]{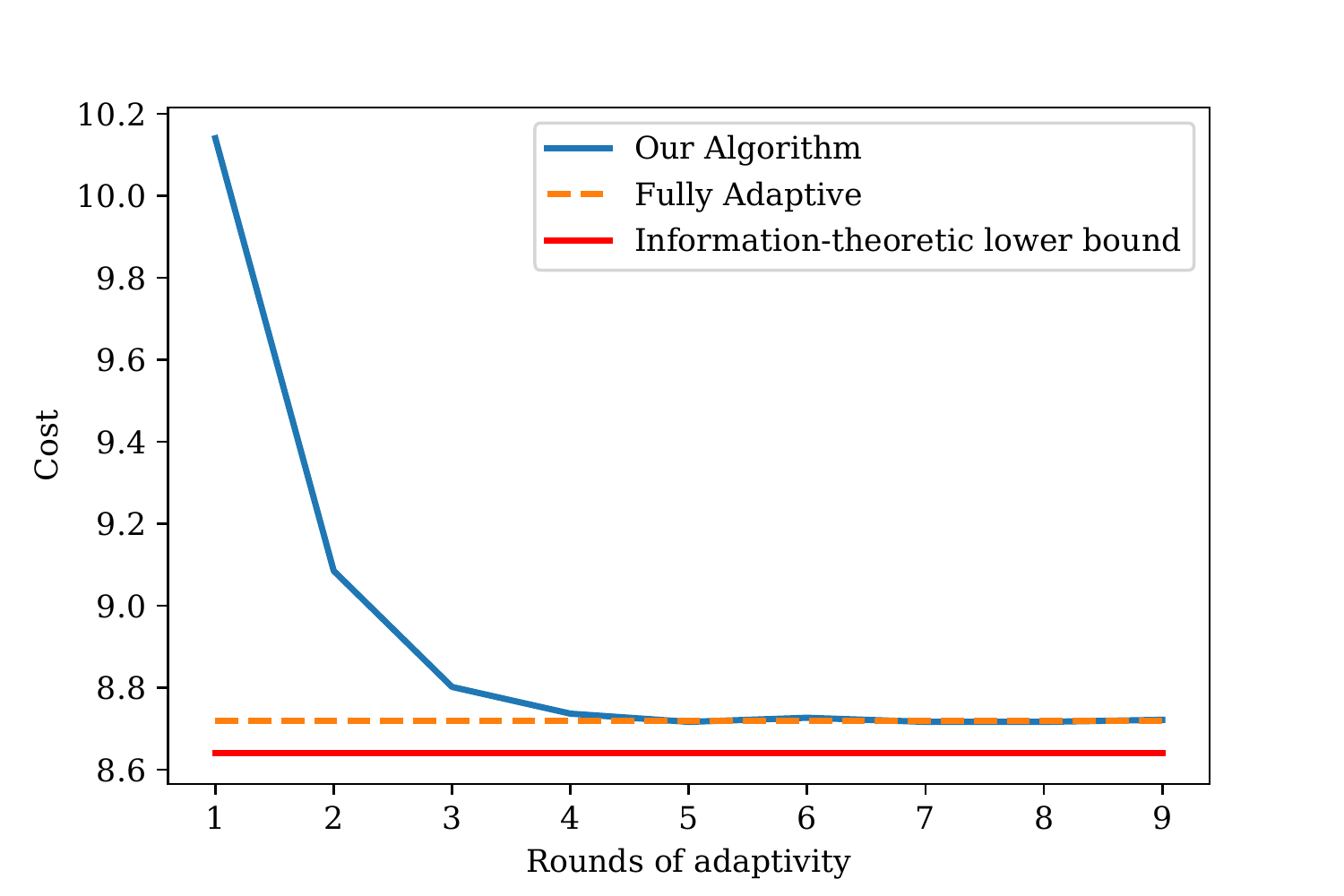}
     \end{subfigure}
     \begin{subfigure}[b]{0.4\textwidth}
         \centering
         \includegraphics[width=\textwidth]{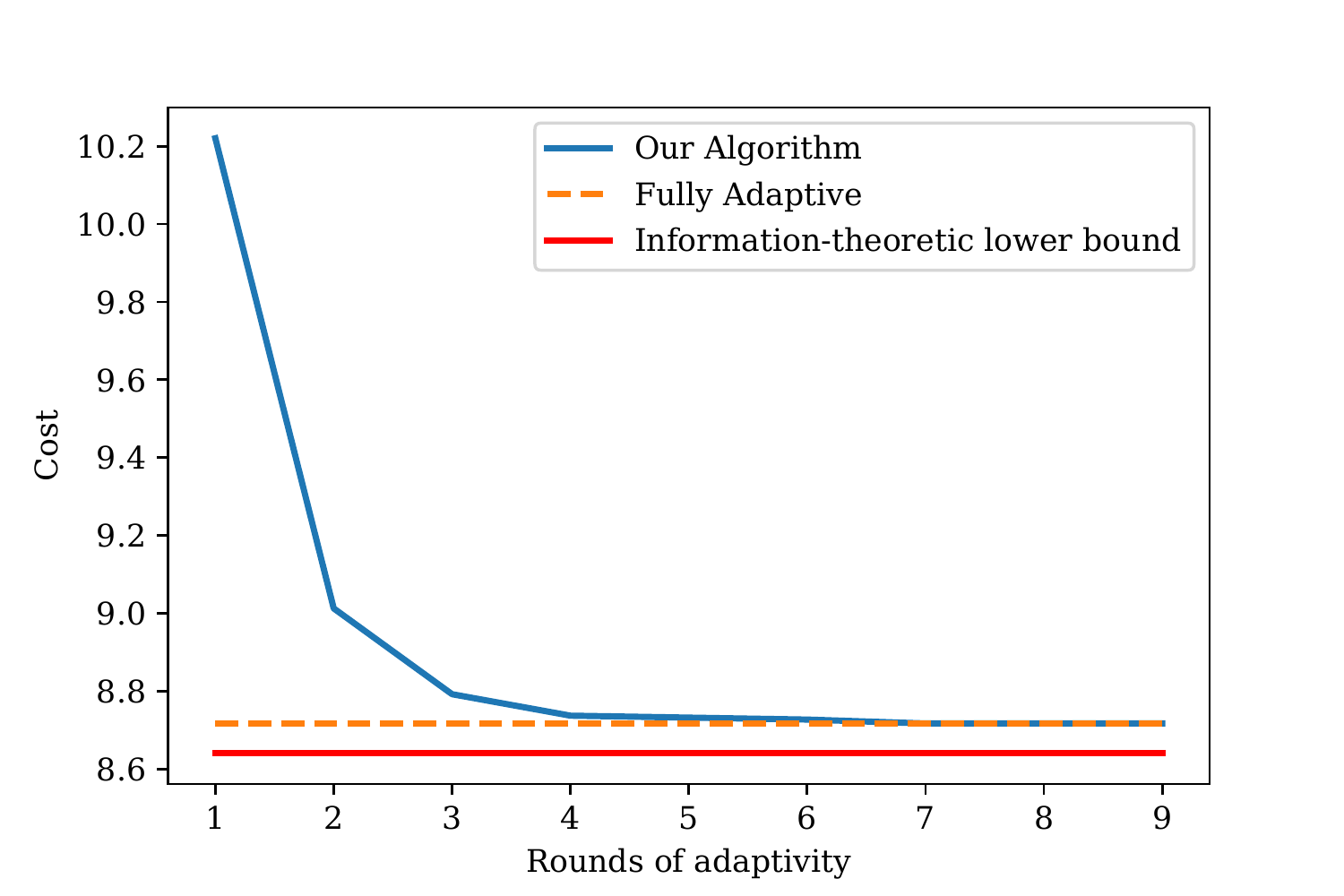}
     \end{subfigure}
     \begin{subfigure}[b]{0.4\textwidth}
         \centering
         \includegraphics[width=\textwidth]{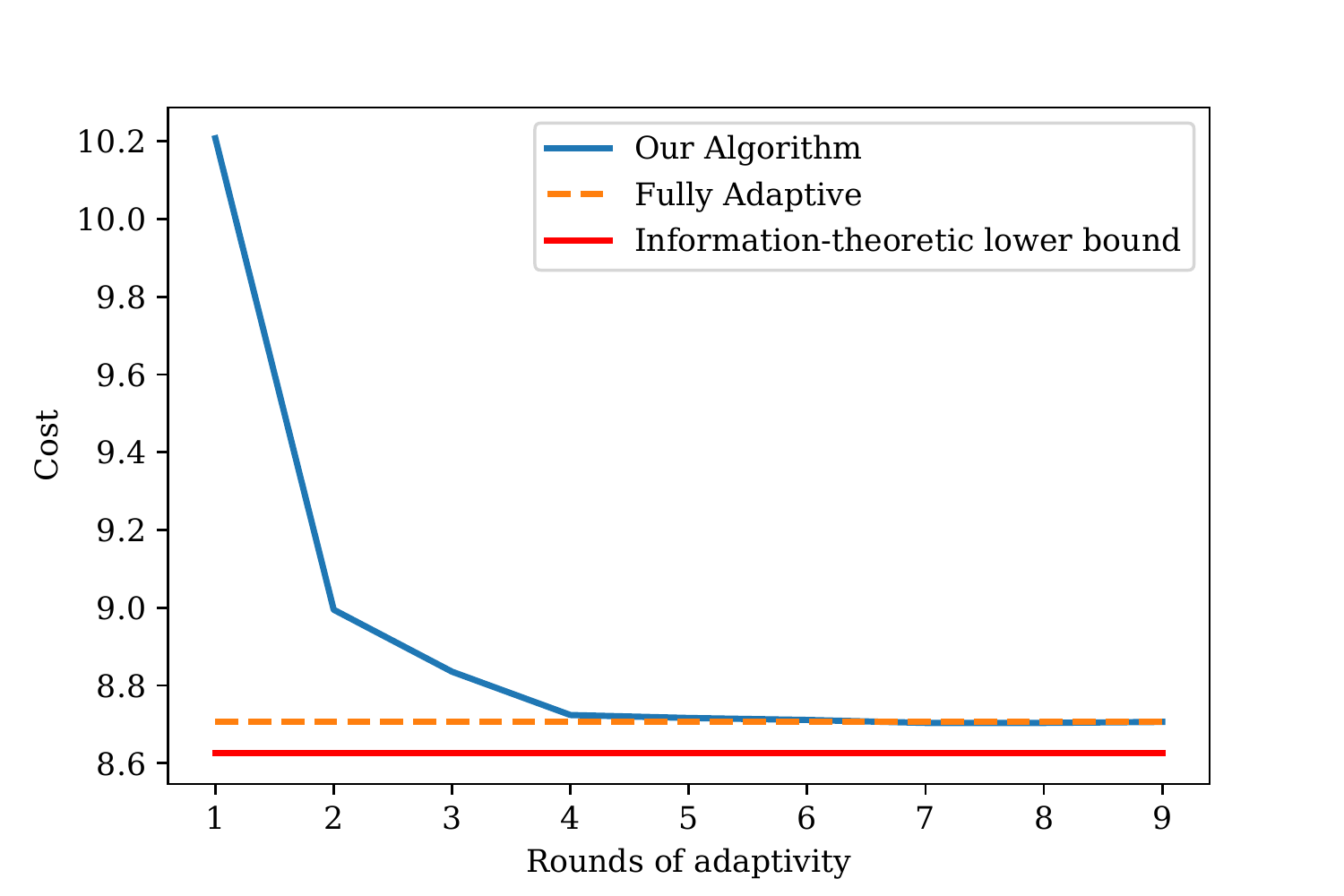}
     \end{subfigure}
     \begin{subfigure}[b]{0.4\textwidth}
         \centering
         \includegraphics[width=\textwidth]{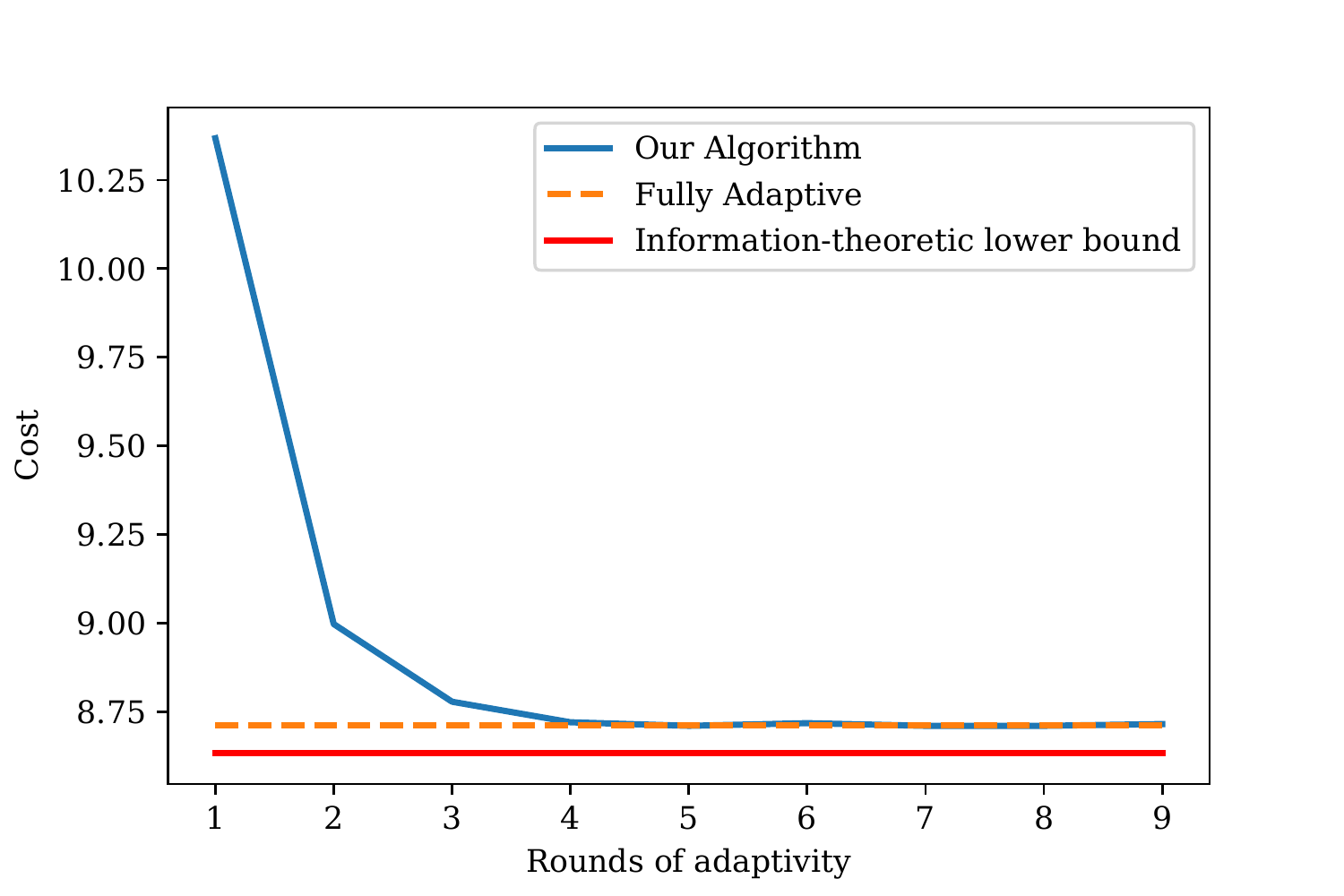}
     \end{subfigure}
     \begin{subfigure}[b]{0.4\textwidth}
         \centering
         \includegraphics[width=\textwidth]{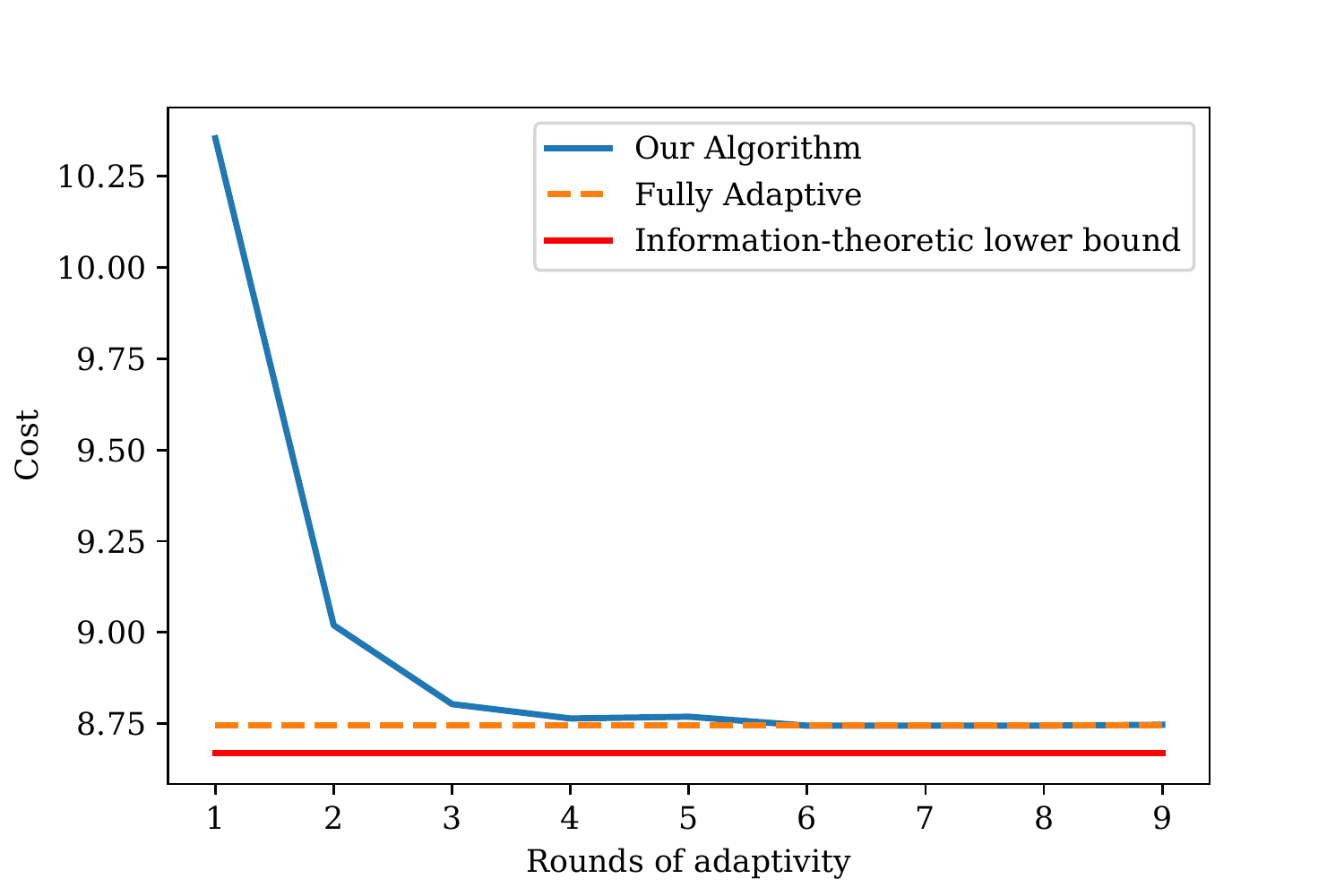}
     \end{subfigure}
     \begin{subfigure}[b]{0.4\textwidth}
         \centering
         \includegraphics[width=\textwidth]{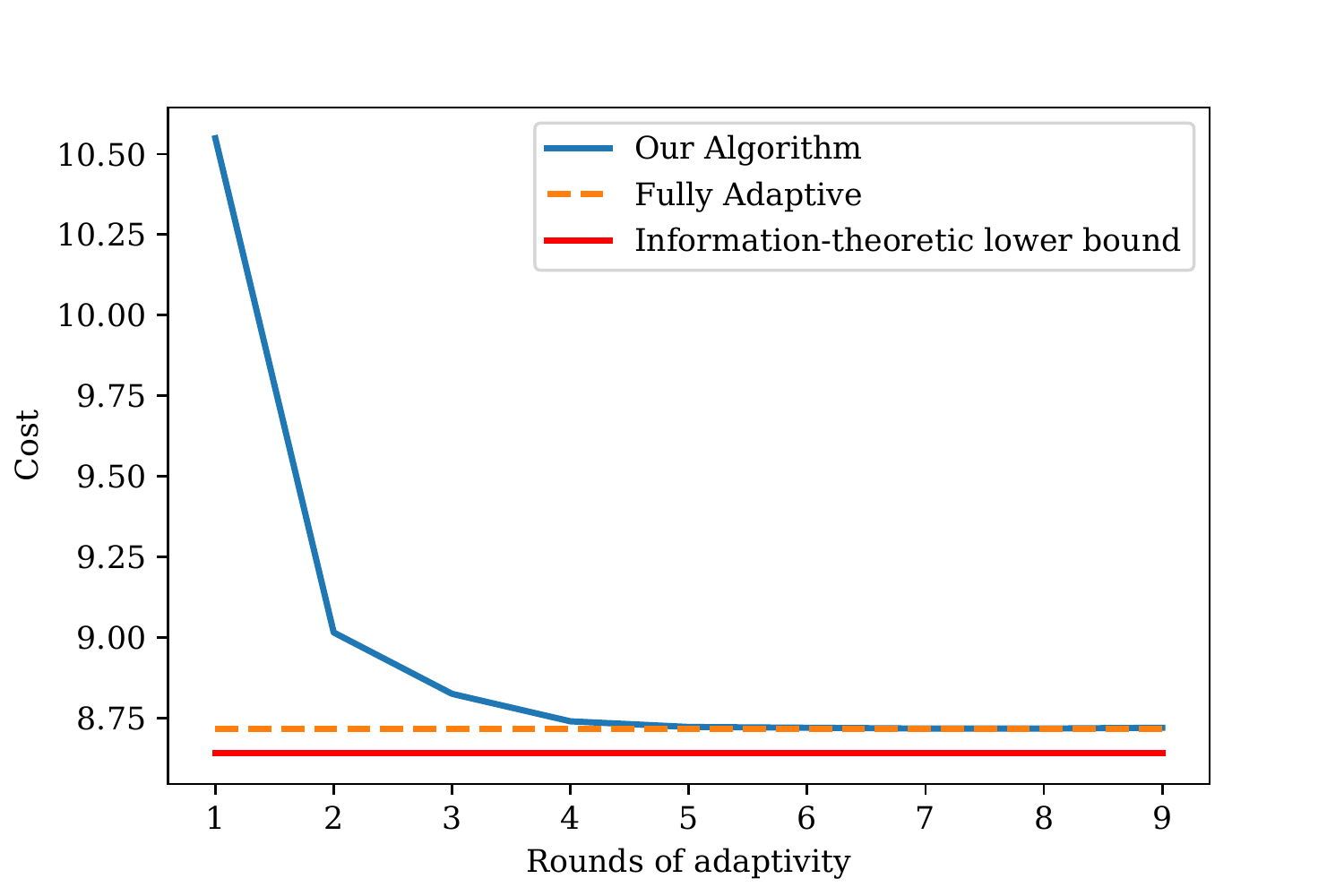}
     \end{subfigure}
     \begin{subfigure}[b]{0.4\textwidth}
         \centering
         \includegraphics[width=\textwidth]{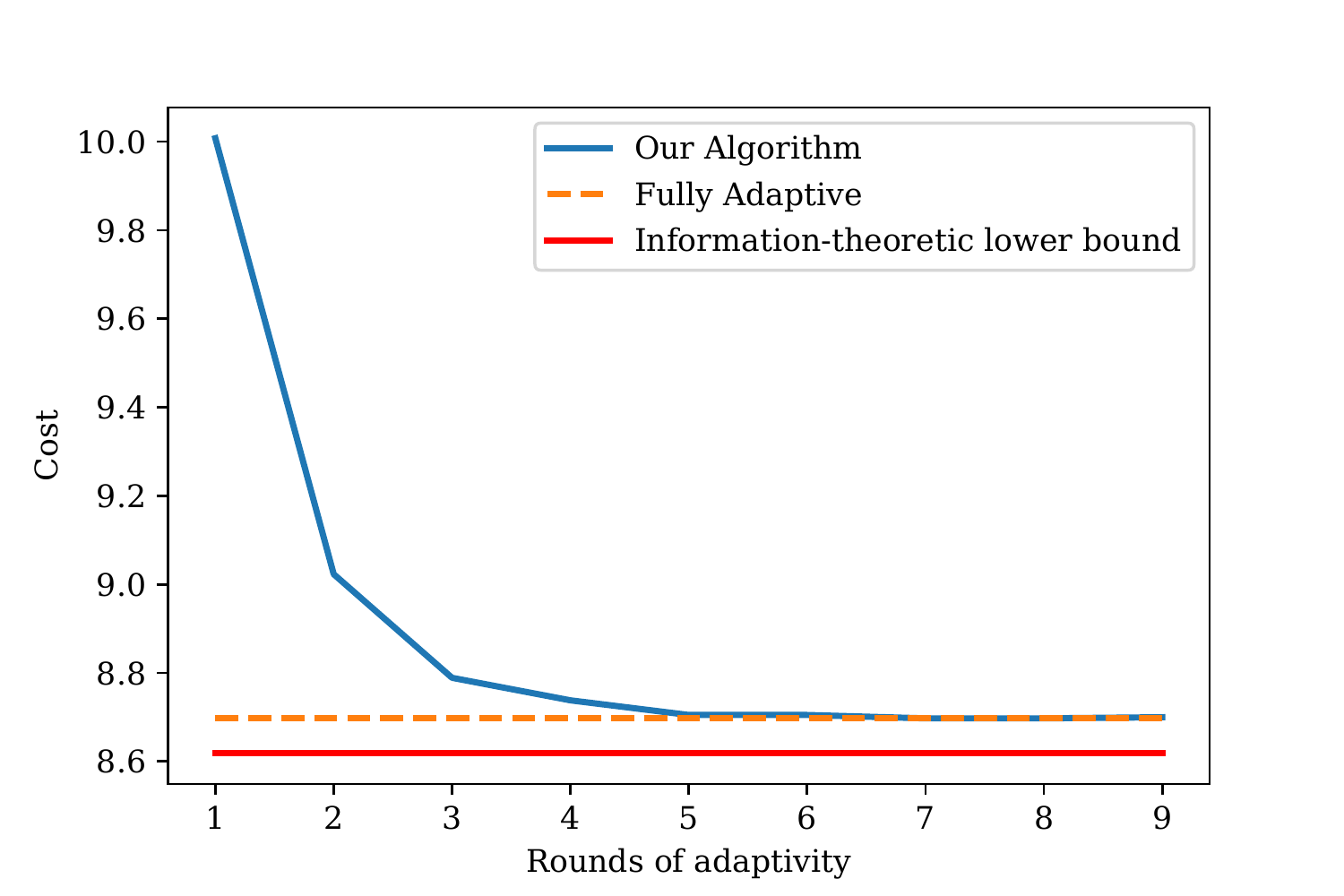}
     \end{subfigure}
     \begin{subfigure}[b]{0.4\textwidth}
         \centering
         \includegraphics[width=\textwidth]{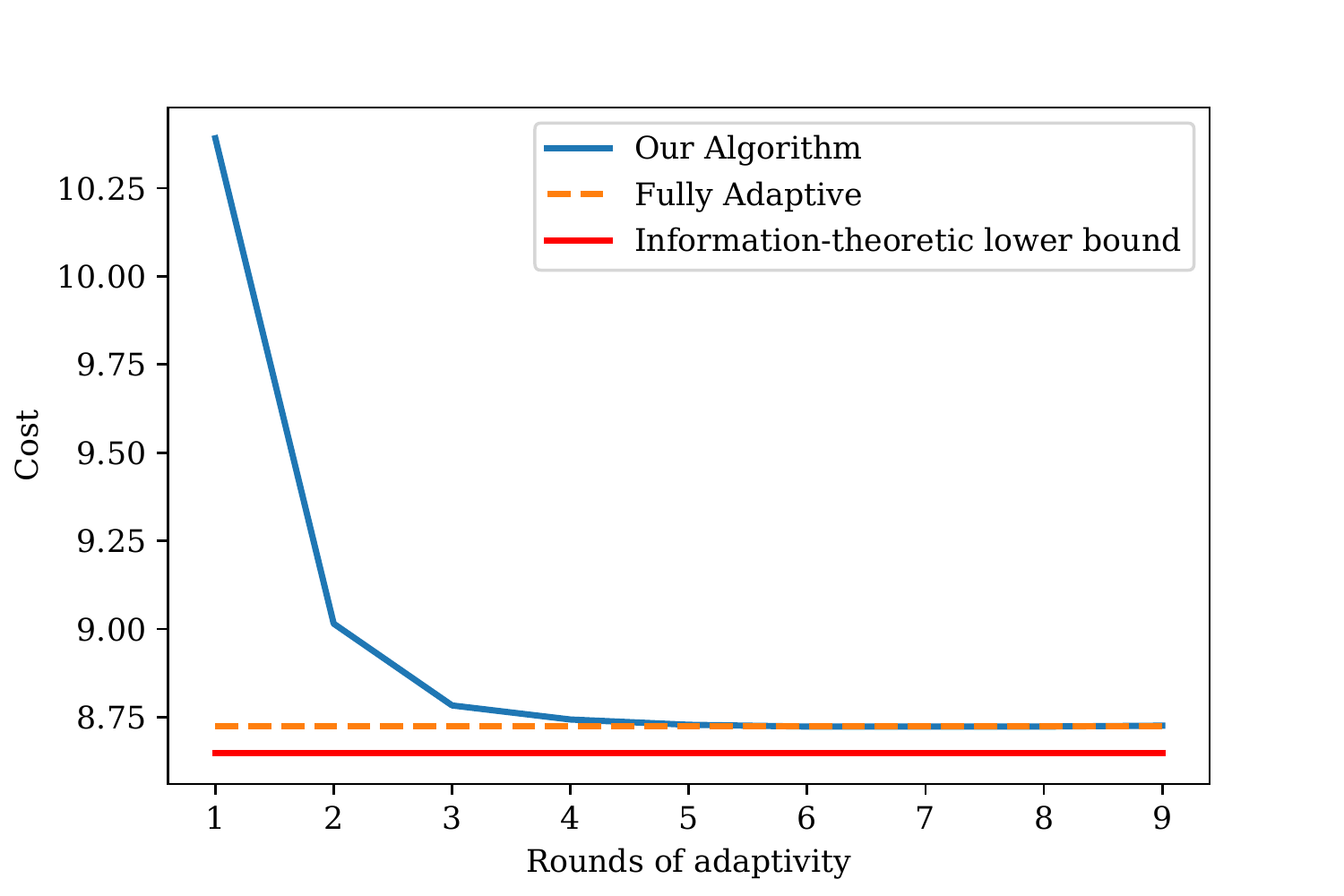}
     \end{subfigure}
     \hfill
        \caption{Computational results for $\ODT$ on $\WISERu$}
        \label{fig:comp-odt-wiser-unit-all}
\end{figure}

\begin{figure}
     \centering
     \begin{subfigure}[b]{0.4\textwidth}
         \centering
         \includegraphics[width=\textwidth]{figures_odt/wiser-0-rand-costs.pdf}
     \end{subfigure}
     \begin{subfigure}[b]{0.4\textwidth}
         \centering
         \includegraphics[width=\textwidth]{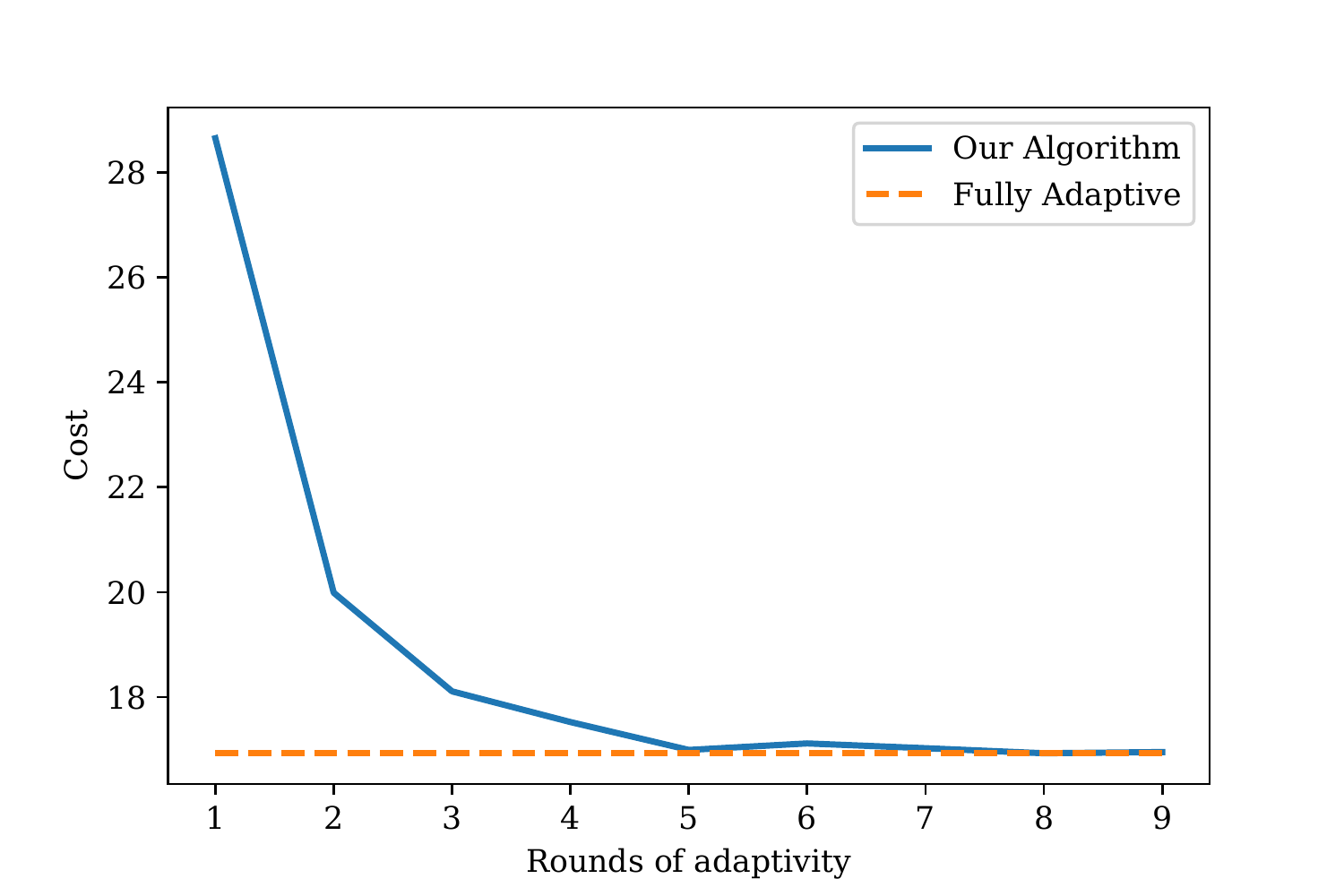}
     \end{subfigure}
     \begin{subfigure}[b]{0.4\textwidth}
         \centering
         \includegraphics[width=\textwidth]{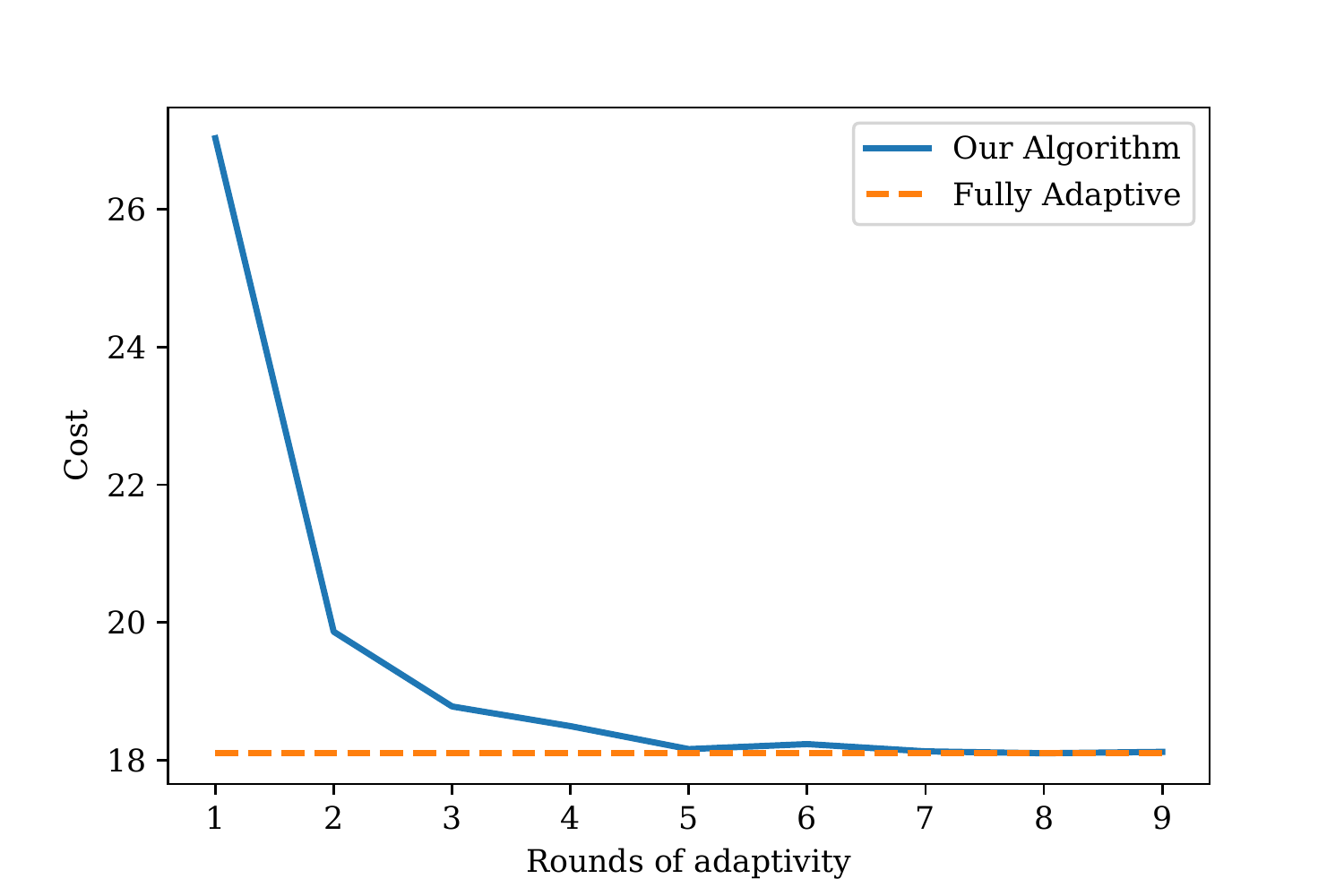}
     \end{subfigure}
     \begin{subfigure}[b]{0.4\textwidth}
         \centering
         \includegraphics[width=\textwidth]{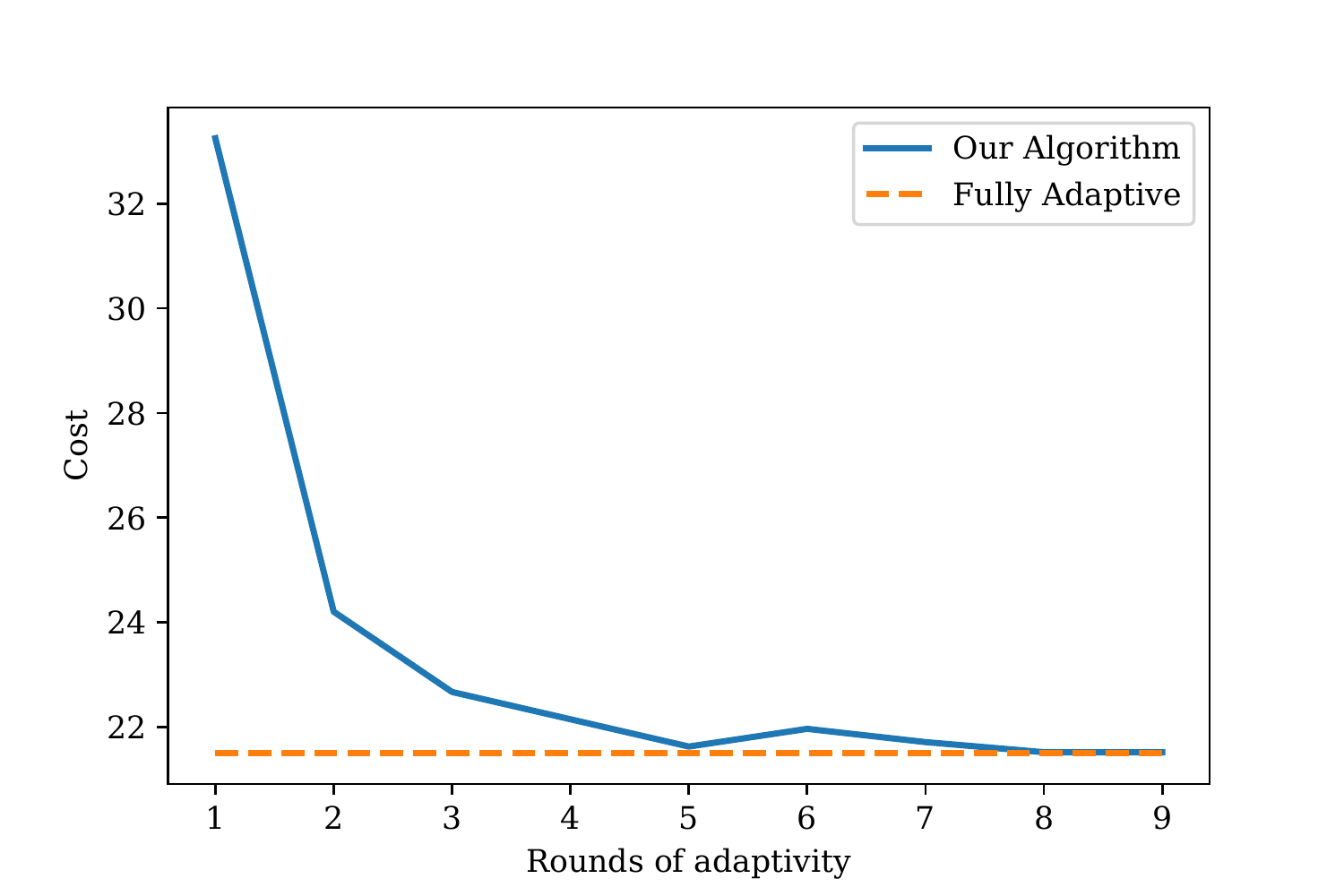}
     \end{subfigure}
     \begin{subfigure}[b]{0.4\textwidth}
         \centering
         \includegraphics[width=\textwidth]{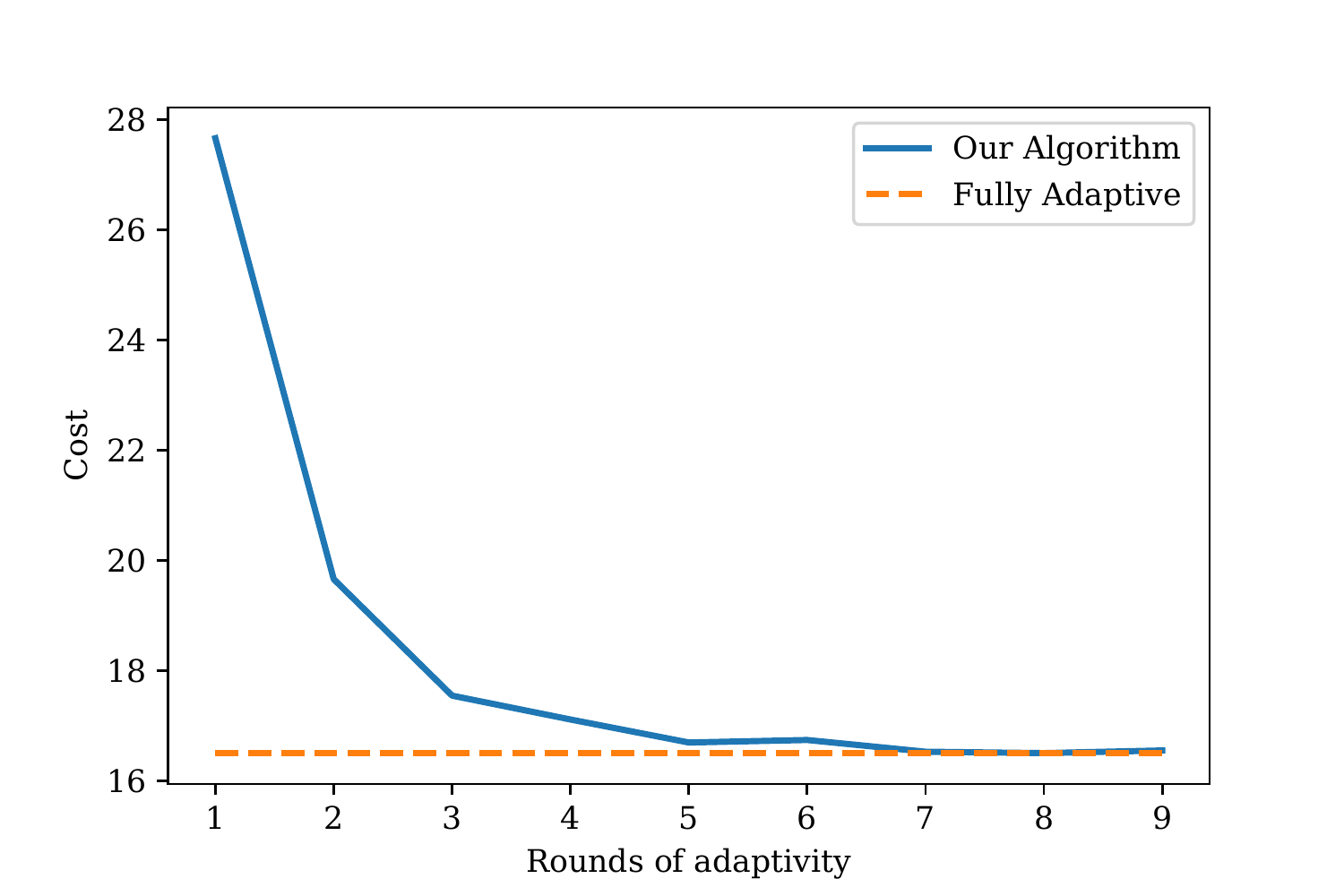}
     \end{subfigure}
     \begin{subfigure}[b]{0.4\textwidth}
         \centering
         \includegraphics[width=\textwidth]{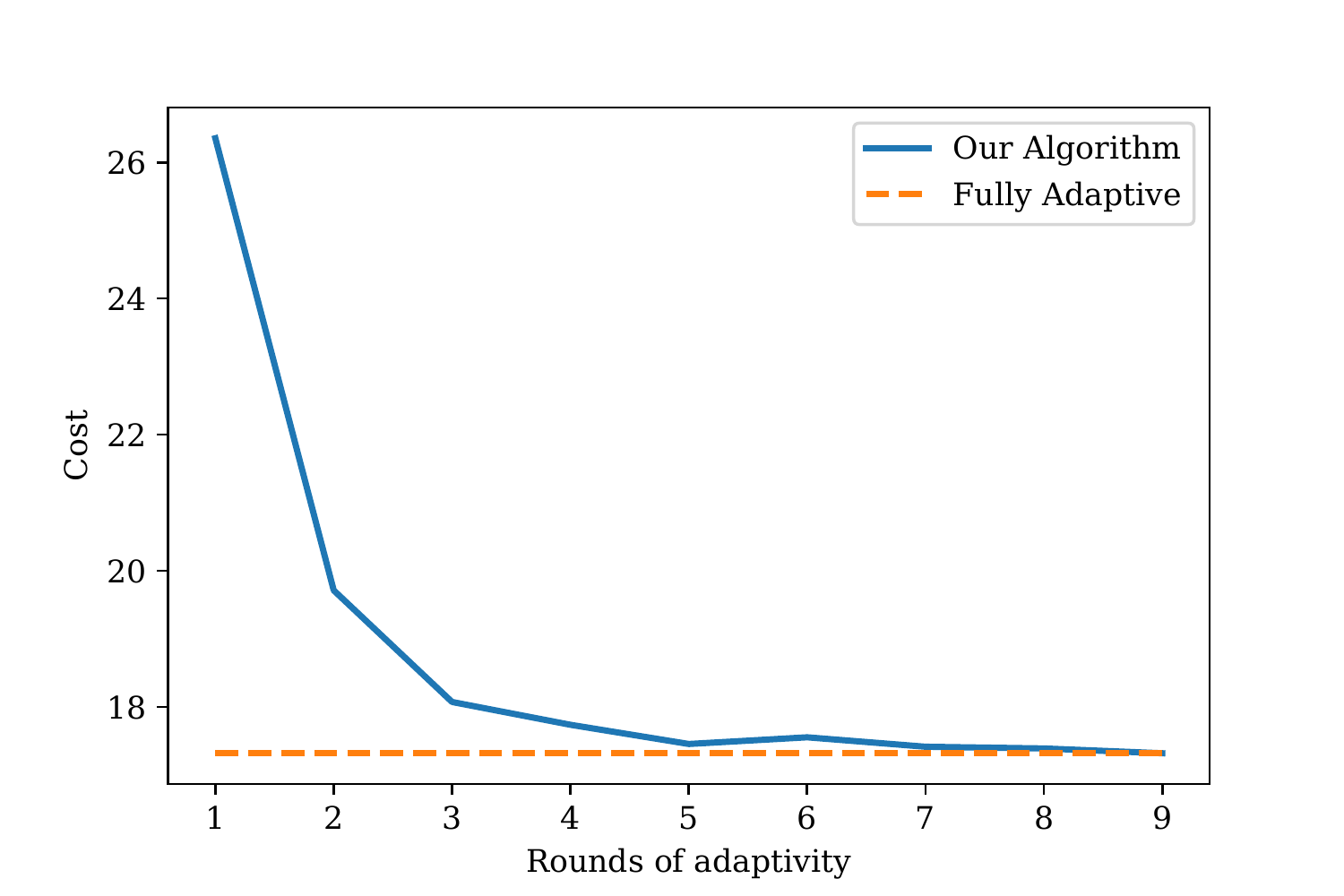}
     \end{subfigure}
     \begin{subfigure}[b]{0.4\textwidth}
         \centering
         \includegraphics[width=\textwidth]{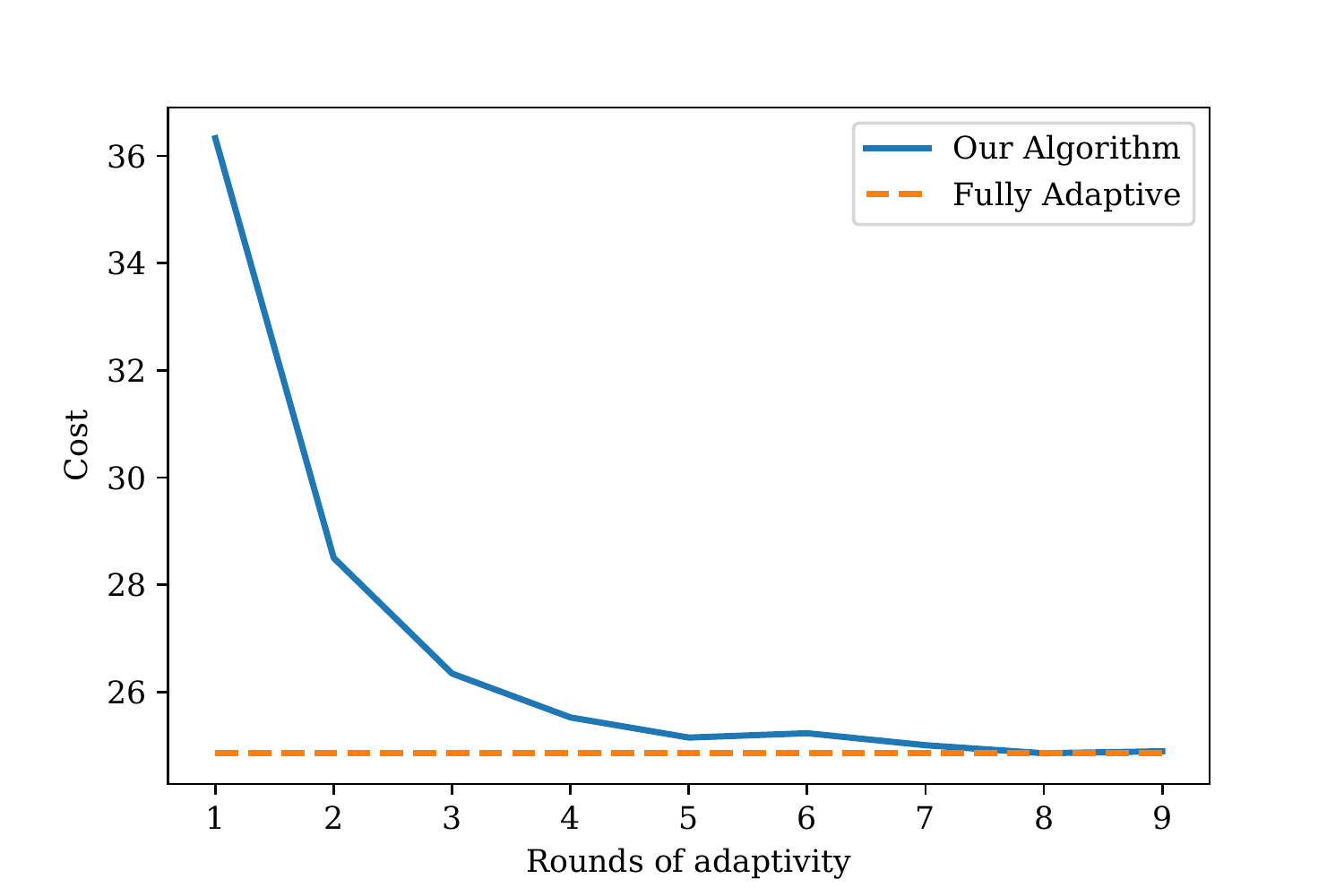}
     \end{subfigure}
     \begin{subfigure}[b]{0.4\textwidth}
         \centering
         \includegraphics[width=\textwidth]{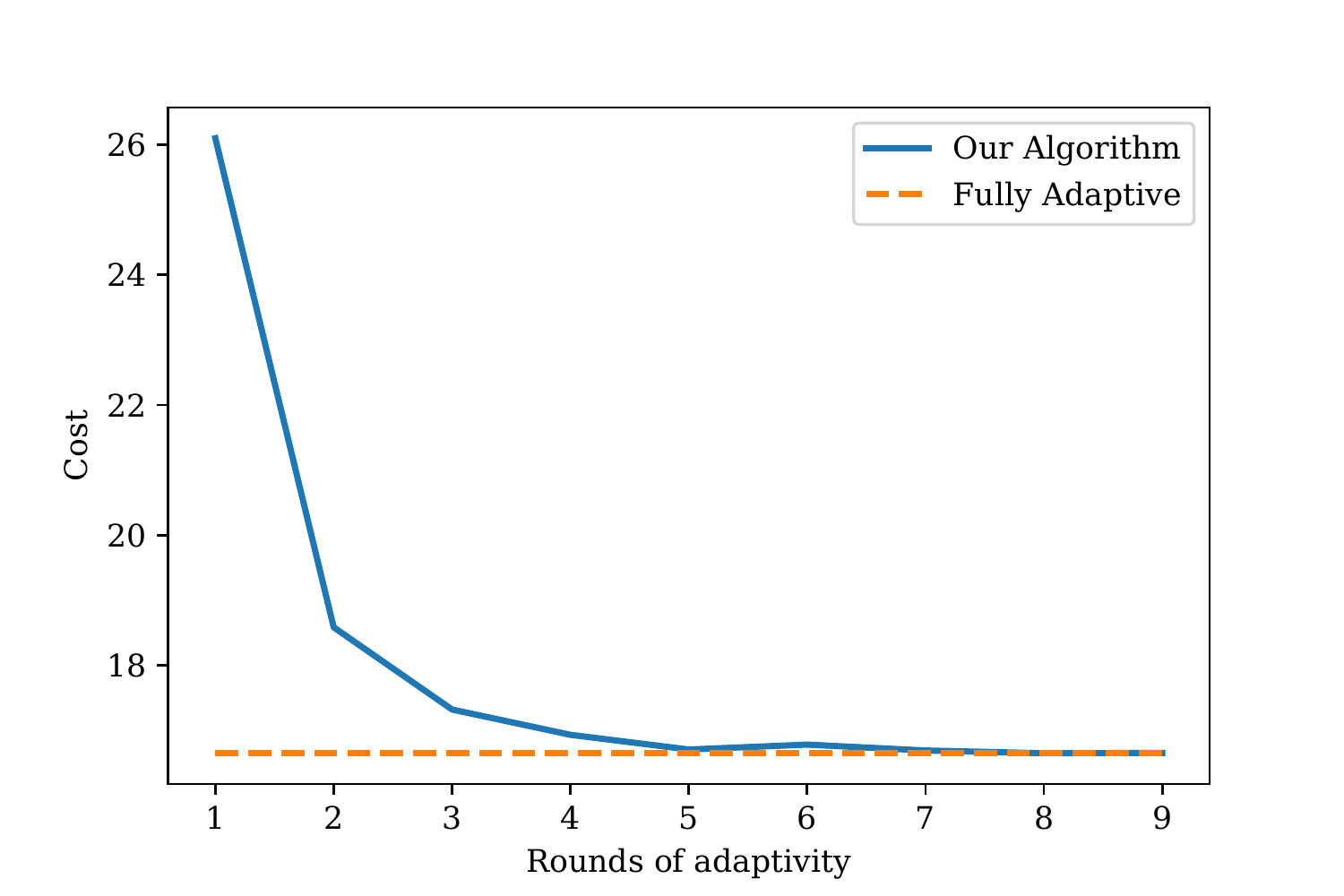}
     \end{subfigure}
     \begin{subfigure}[b]{0.4\textwidth}
         \centering
         \includegraphics[width=\textwidth]{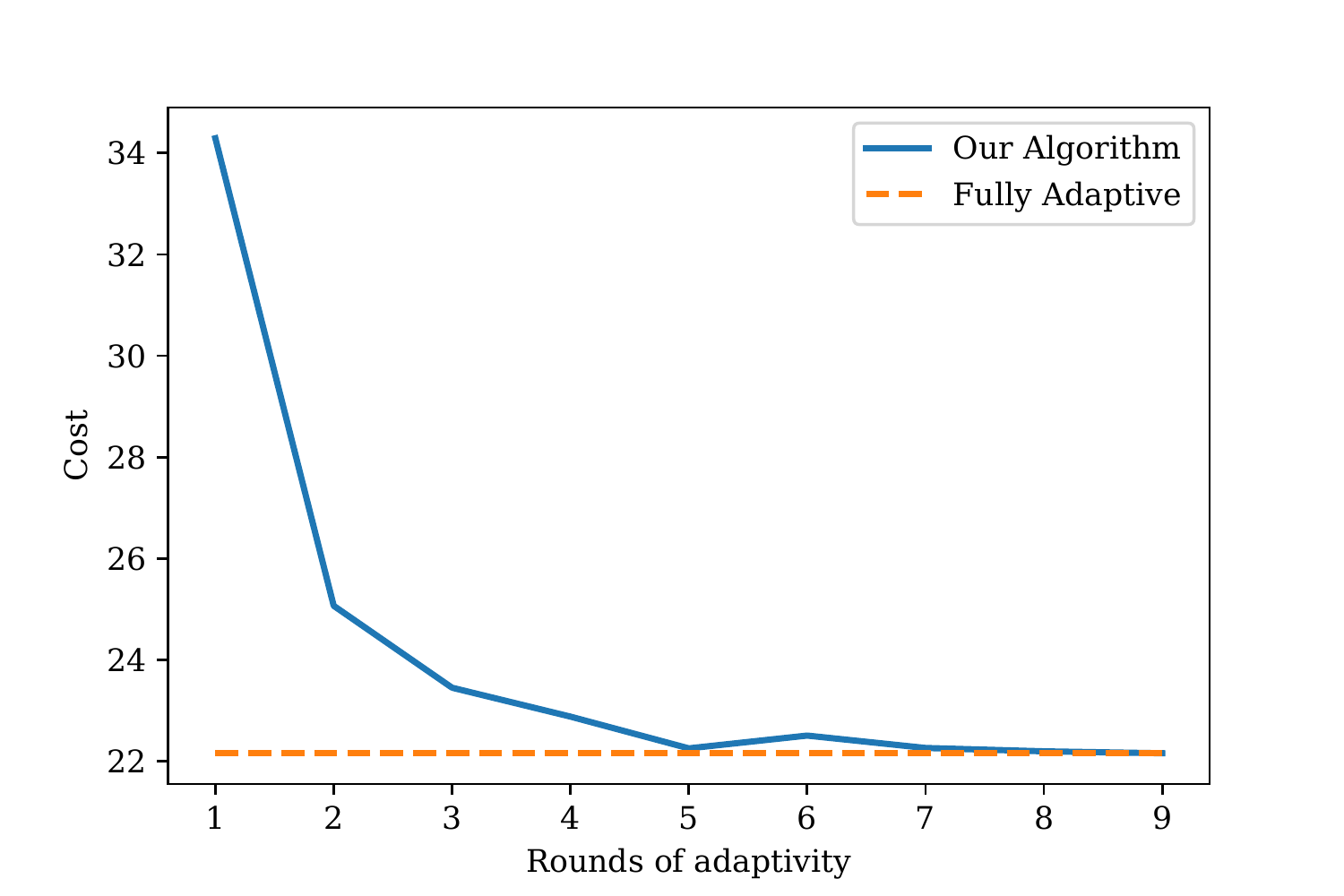}
     \end{subfigure}
     \begin{subfigure}[b]{0.4\textwidth}
         \centering
         \includegraphics[width=\textwidth]{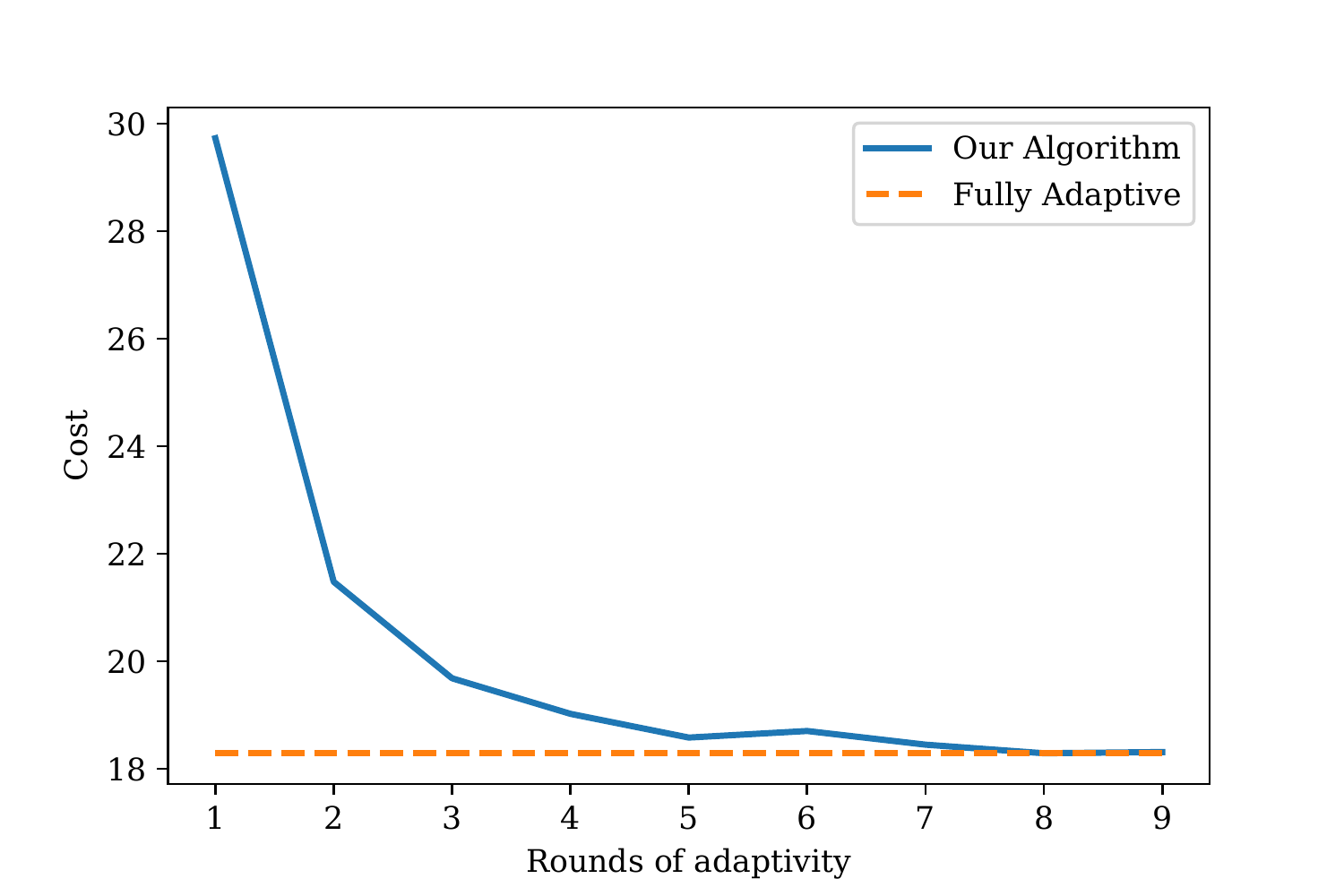}
     \end{subfigure}
     \hfill
        \caption{Computational results for $\ODT$ on $\WISERr$}
        \label{fig:comp-odt-wiser-random-all}
\end{figure}

\end{document}